\documentclass[11pt,fleqn]{article}

\usepackage{amsmath}
\usepackage{amssymb}
\usepackage{amsfonts}
\usepackage{amsthm}
\usepackage[authoryear]{natbib}
\usepackage{authblk}
\usepackage{enumerate}
\usepackage{epstopdf}
\usepackage{booktabs}
\usepackage{graphicx}
\usepackage{multirow}
\usepackage{subfigure}
\usepackage{placeins}
\usepackage{setspace}
\usepackage{footmisc}
\usepackage{rotating}
\usepackage{xcolor}
\usepackage{array}
\usepackage{bbm}
\usepackage{geometry}
\usepackage{lmodern}
\usepackage[hyperfootnotes=false]{hyperref}
\usepackage{xifthen}
\usepackage[title]{appendix}

\hypersetup{pdfborder = 0 0 0,%
    pdftitle = Autoregressive Wild Bootstrap Inference for Nonparametric Trends,%
    pdfauthor = {Marina Friedrich, Stephan Smeekes and Jean-Pierre Urbain},%
    setpagesize = false,%
    pdfpagelayout = SinglePage,%
    bookmarksopen = true,%
    bookmarksnumbered = true,%
    pdfstartview = Fit,%
    linkcolor = black,%
    anchorcolor = black,%
    filecolor = black,%
    citecolor = black,%
    menucolor = black,%
    urlcolor = blue,%
    colorlinks = true,%
}
\newtheoremstyle{rem}%
{10pt}%
{10pt}%
{}%
{}%
{\bf}%
{.}%
{.5em}%
{}%

\newtheoremstyle{assm}%
{\topsep}%
{\topsep}%
{\normalfont}%
{}%
{\bfseries}%
{}%
{\newline}%
{}%

\newtheoremstyle{alg}%
{\topsep}%
{\topsep}%
{\normalfont}%
{}%
{\bfseries}%
{}%
{\newline}%
{\thmname{#1}\thmnumber{ #2}:\thmnote{ #3}}%

\theoremstyle{plain}
\newtheorem{assumption}{Assumption}

\theoremstyle{rem}
\newtheorem{algorithm}{Algorithm}
\theoremstyle{plain}
\newtheorem{theorem}{Theorem}
\newtheorem{lemma}{Lemma}

\theoremstyle{rem}
\newtheorem{remark}{Remark}

\newcommand{\ssymbol}[1]{^{\@fnsymbol{#1}}}

\DeclareMathOperator{\E}{\mathbb{E}}

\newcommand{\dsum}[2][]{%
\ifthenelse{\isempty{#1}}{
\sum_{#2=1}^{n}}{\sum_{#2=1}^{n-#1}}}
\renewcommand{\P}{\mathbb{P}}
\DeclareMathOperator{\var}{\mathbb{V}ar}
\DeclareMathOperator{\cov}{\mathbb{C}ov}
\newcommand{\abs}[1]{\left\lvert#1\right\rvert}
\newcommand{\norm}[1]{\left\lVert#1\right\rVert}
\DeclareMathOperator*{\argmin}{arg~min}
\DeclareMathOperator*{\plim}{plim}
\DeclareMathOperator{\suptau}{\sup_{\tau \in [\delta, 1-\delta]}}
\DeclareMathOperator{\suptaun}{\sup_{\tau_0 \in [\delta, 1-\delta]}}
\DeclareMathOperator{\suptaus}{\sup_{\tau \in [\delta^*, 1-\delta^*]}}
\DeclareMathOperator{\suptauns}{\sup_{\tau_0 \in [\delta^*, 1-\delta^*]}}
\DeclareMathOperator{\suptaut}{\sup_{\tau \in [-1,1]}}
\DeclareMathOperator{\suptautt}{\sup_{\tau_1, \tau_2 \in [-1,1]}}

\makeatletter
\renewcommand\d[1]{\mspace{6mu}\mathrm{d}#1\@ifnextchar\d{\mspace{-3mu}}{}}
\makeatother

\hypersetup{pdfborder = 0 0 0,%
}

\title{\sc Autoregressive Wild Bootstrap Inference for Nonparametric Trends\thanks{Department of Quantitative Economics, Maastricht University, P.O. Box 616, 6200 MD Maastricht, The
Netherlands and Potsdam Institute for Climate Impact Research (PIK), Member of the Leibniz Association, P.O. Box 601203, 14412 Potsdam, Germany. E-mail: \href{mailto:friedrich@pik-potsdam.de}{friedrich@pik-potsdam.de}, \href{mailto:s.smeekes@maastrichtuniversity.nl}{s.smeekes@maastrichtuniversity.nl}.
This work is licensed under the Creative Commons CC-BY-NC-ND License. To view a copy of the license, visit \href{https://creativecommons.org/licenses/by-nc-nd/4.0/}{creativecommons.org/licenses/by-nc-nd/4.0/}}}
\author[a,b]{Marina Friedrich}
\author[a]{Stephan Smeekes} \author[a]{Jean-Pierre Urbain\footnote{Deceased on October 1, 2016}}

\affil[a]{Maastricht University, Department of Quantitative Economics}
\affil[b]{Potsdam Institute for Climate Impact Research (PIK)}
\begin{document}
\maketitle

\begin{abstract}
In this paper we propose an autoregressive wild bootstrap method to construct confidence bands around a smooth deterministic trend. The bootstrap method is easy to implement and does not require any adjustments in the presence of missing data, which makes it particularly suitable for climatological applications. We establish the asymptotic validity of the bootstrap method for both pointwise and simultaneous confidence bands under general conditions, allowing for general patterns of missing data, serial dependence and heteroskedasticity. The finite sample properties of the method are studied in a simulation study. We use the method to study the evolution of trends in daily measurements of atmospheric ethane obtained from a weather station in the Swiss Alps, where the method can easily deal with the many missing observations due to adverse weather conditions.
\end{abstract}

\section{Introduction}
The analysis of smoothly evolving trends is of interest in many fields such as economics and climatology. For instance, trend analysis in environmental variables is of major importance, as it can often directly be linked to climate change. Given that trends typically do not evolve in a linear way, fitting linear trends to the data does not uncover actual change accurately. Instead, one would like to use more flexible trend models that avoid making parametric assumptions on the form of the trend. A large body of statistical and econometric research therefore focuses on nonparametric trend modeling and estimation.

In addition, when modeling trends in temperature and emission data,  researchers also need to take into account that serial dependence and heteroskedasticity may be present in the data, see for example \citet{FV} and \citet{MV} who study parametric trend modeling in temperature series in the presence of serial dependence. Bootstrap methods provide an easy and powerful way to account for heteroskedasticity and autocorrelation. \citet{Buhlmann} shows the validity of the autoregressive sieve bootstrap for nonparametric trend modeling under general forms of dependence. \citet{Neumann} uses a wild bootstrap method to achieve robustness to heteroskedasticity for a similar model. 

The wild bootstrap approach is also suitable for dealing with missing data, as advocated for example by \citet{Shao}. It does not require any resampling and therefore the missing data points can keep their original date in a bootstrap sample. In particular in climatology, this feature of the wild bootstrap offers an important benefit over other methods, since there is no need of imputing missing data points. Missing data are a prominent feature in many climatological datasets due to instrument failure or adverse weather and measurement conditions; for instance, in our application, data are missing when cloud cover prevents measurements from being taken. The wild bootstrap, however, relies on independence of the error terms, which is a situation rarely encountered in practice. To relax this strong assumption, dependent versions of wild bootstrap methods have been proposed - see \citet{Shao}, \citet{LN} and \citet{SU} - but not in the context of nonparametric trend estimation. Moreover, so far no theory exists on the validity of such wild bootstrap methods in the presence of serial dependence, heteroskedasticty and missing data. In this paper we address this issue and propose an autoregressive wild bootstrap method that provides valid inference under general conditions for nonparametric trend modeling.

Next to the basic pointwise confidence intervals, we also study simultaneous confidence bands, which are often more informative about trend shapes than pointwise confidence intervals. Research questions, like whether upward trends are present over a certain period of time, should be addressed with simultaneous confidence bands as they involve multiple points in time at once. \citet{WuZhao} derive such bands for the nonparametric trend model that have asymptotically correct coverage probabilities, but do not consider bootstrap methods. \citet{Buhlmann} proposes sieve bootstrap-based simultaneous confidence bands that are not only asymptotically valid but also have good small sample performance. They can, however, not easily be adjusted to be applicable to time series with missing data, as the autoregressive sieve bootstrap requires imputation of the missing values through for instance the Kalman filter. While this is certainly possible, it complicates implmentation. One can also argue about how accurate imputation methods are when a majority of the data are missing, as we face in our climatological application. Instead, we provide a much simpler alternative that requires no adjustments at all in the presence of missing data.

To illustrate our methodology, we study a time series of atmospheric ethane emissions for which almost 70\% of the data points are missing. When weather conditions are unfavorable -- in particular due to cloud cover --  measurements cannot be taken. The series has previously been investigated by \citet{Franco}. Atmospheric ethane is an indirect greenhouse gas which can be used as an indicator of atmospheric pollution and transport. It is emitted during shale gas extraction and since shale gas has become more and more important as a source of natural gas, nonparametric trend analysis in ethane data provides geophysicists and climatologists with a tool to link trend changes to shale gas extraction activities, as well as study long-term climatological change.

The paper is organized as follows. In Section \ref{sec:model}, our trend model is introduced along with the missing data generating mechanism. Section \ref{sec:inference} describes the estimation procedure and the construction of bootstrap confidence bands. Subsequently, Section \ref{sec:theory} derives the asymptotic properties of our method. Finite sample performance is analyzed in Section \ref{sec:simulation} in a simulation study. Trends in atmospheric ethane are studied in Section \ref{sec:application}.  Section \ref{conclusion} concludes. All technical details including proofs are given in Appendix A, while Supplementary Appendices B to D provide further results.

Finally, a word on notation. We denote by $\xrightarrow{d}$ weak convergence and by $\xrightarrow{p}$ convergence in probability. Whenever a quantity has a subscript $^{\ast}$, it denotes a bootstrap quantity, conditional on the original sample. For instance, bootstrap weak convergence in probability is denoted by $\xrightarrow{d^*}_p$ \citep[cf.][]{GineZinn}. $\left\lfloor x\right\rfloor$ stands for the largest integer smaller than or equal to $x$. For any functions $f(x)$ and $g(x)$, defined on the same domain, $f^{(i)}(x) = \frac{d^i}{d x^i} f(x)$ and $\left[f g \right]^{(i)} (x) = \frac{d^i}{d x^i} f(x) g (x)$.

\section{Trend Model with Missing Data} \label{sec:model}
Consider the following data generating process (DGP):
\begin{equation*}
y_t = m\left(\frac{t}{n} \right) + z_t \qquad t= 1,\ldots,n,
\end{equation*}
where $m\left(\cdot \right)$ is a smooth deterministic trend function and $z_t=\sigma_t u_t$ is a weakly dependent stochastic component. $\sigma_t$ captures unconditional heteroskedasticity and $\left\{u_t\right\}$ is a linear process
\begin{equation} \label{eq:MA_model}
u_t = \sum_{j=0}^{\infty} \psi_j \epsilon_{t-j}, \qquad \psi_0=1,
\end{equation}
with autocovariance function $R_U(k)= \E u_t u_{t+k}$ and long-run variance
\begin{equation*}
\Omega_U = \sum_{k=-\infty}^\infty \E u_t u_{t+k} = \sum_{k=-\infty}^\infty R_U(k).
\end{equation*}

Not all observations $y_1, \ldots, y_n$ are observed in practice. For this purpose, define the process $\{D_t\}$ as an indicator for whether the observations at each time are observed: \begin{equation*}
D_t = \left\{
\begin{array}{ll}
1 & \text{if $y_t$ is observed} \\
0& \text{if $y_t$ is missing}
\end{array}
\qquad t = 1, \ldots, n, \right.
\end{equation*}

Assumptions \ref{as:smooth} to \ref{as:MD} contain the formal conditions that $\{y_t\}$ and $\{D_t\}$ satisfy.

\begin{assumption}\label{as:smooth}
$m:[0,1]\rightarrow\mathbb{R}$ is a twice continuously differentiable deterministic function on $(0,1)$ with $\sup_{0<\tau<1}\abs{m^{(i)}(\tau)}<\infty$ for $i=0,1,2$.
\end{assumption}

\begin{assumption}\label{as:sigma}
$\sigma:[0,1]\rightarrow\mathbb{R}^{+}$ is a Lipschitz continuous deterministic function.
\end{assumption}

\begin{assumption}\label{as:LP}
$\left\{u_t\right\}$ is generated by \eqref{eq:MA_model}, where
\begin{enumerate}[(i)]
\item $\left\{\epsilon_t\right\}$ is i.i.d.~with $\E(\epsilon_t)=0$, $\E(\epsilon_t^2)=\sigma_\varepsilon^2= \left(\sum_{j=0}^\infty \psi_j^2 \right)^{-1}$, and $\E(\epsilon_t^4)<\infty$.
\item $\sum_{j=0}^\infty j|\psi_j|<\infty$ and the lag polynomial $\Psi(z)=\sum_{j=0}^{\infty}\psi_j z^j \neq 0$ for all $z\in \mathbb{C}$ and $|z|\leq 1$.
\end{enumerate}
\end{assumption}

\begin{assumption}\label{as:MD}
For all $t = 1, \ldots, n$, $D_t$ satisfies the conditions:
\begin{enumerate}[(i)]
\item Let $\E_t (\cdot) = \E (\cdot | \mathcal{F}_t)$ where $\mathcal{F}_t = \sigma(\ldots, (y_{t-1}, D_{t-1})^\prime, (y_t, D_t)^\prime)$. For all $s \leq t$ and $i\geq 0$, $\E [\E_{t-i} D_s D_t - \E D_s D_t ]^2 \leq \zeta_i^2$, where $\sum_{i=0}^{\infty} i \zeta_i < \infty$.
\item $\E D_t = \P (D_t = 1) = p(t/n)$, where $p:[0,1] \rightarrow [\epsilon ,1]$, for some $\epsilon > 0$, is a twice continuously differentiable function on $(0,1)$ with $\sup_{0<\tau<1}\abs{p^{(i)}(\tau)}<\infty$ for $i=1,2$.
\item $\cov(D_t, D_{t+i}) = R_{D,i}\left(\frac{t}{n}, \frac{t+i}{n}\right)$, where each function $R_{D,i}:[0,1]^2 \rightarrow \mathbb{R}$, $i\geq 0$, is Lipschitz continuous.
\item For all $s_1,s_2 \in \{1, \ldots,n\}$, $\E (u_{s_1} | D_t) = 0$ and $\E (u_{s_1} u_{s_2} | D_t) = \E u_{s_1} u_{s_2}$.
\end{enumerate}
\end{assumption}

Assumption \ref{as:smooth} postulates that the trend $m(\cdot)$ is sufficiently smooth, which is the fundamental assumption for the estimation method to work. While it rules out abrupt structural breaks, this does not appear to be particularly restrictive for  climatological applications, as many climatological processes tend to be such that change occurs gradually. In particular, as many series are measured daily or even multiple times a day, only instantaneous breaks, which are extremely unlikely in atmospheric processes, would not be covered by the smooth trend model. 

Assumption \ref{as:sigma} allows for a wide array of unconditional heteroskedasticity. While excluding abrupt breaks, these can be allowed for by generalizing the function $\sigma(\cdot)$ to be piecewise Lipschitz as in \citet{SU}. However, given the limited relevance of abrupt breaks for our climatological focus, we do not pursue this in the current paper.

Assumption \ref{as:LP} is a standard linear process assumption that ensures that sufficient moments of $\left\{u_t\right\}$ exist and $\{u_t\}$ is weakly dependent and strictly stationary. These assumptions are satisfied by a large class of processes including, but not limited to, all finite order stationary ARMA models. The assumption also implies that $\Omega_U = \sigma_{\varepsilon}^2 \sum_{i=-\infty}^\infty \sum_{j=0}^\infty \psi_j \psi_{j + \abs{i}} < \infty$ (cf.~Lemma \ref{lem:cov}). While our current assumption does not allow for conditional heteroskedasticity, this could be relaxed at the expense of increasing the complexity of the theoretical arguments, by allowing $\epsilon_t$ to be a martingale difference sequence. Similarly, alternative dependence concepts such as mixing, which is considered in the same bootstrap context by \citet{SU}, could be used as well. However, as conditional heteroskedasticity is not the focus of our paper, we do not consider these extensions.

Assumption \ref{as:MD} allows the missing data generating mechanism to be weakly dependent and non-stationary. The mixingale assumption (i) along with the summability condition on $\zeta_i$ assures weak dependence and summable autocovariances (Lemma \ref{lem:cov}). For technical reasons, we need to put the mixingale assumption on the product $D_s D_t$, but it directly implies that $D_t$ is a mixingale, as well, by setting $s=t$. By (ii) and (iii), the first two moments of the missing data process are allowed to vary smoothly over time, which thereby allows for instance for smooth periodical changing probabilities (e.g.~due to seasonal variation), or long-term changes related to climate change. Smoothness is required as our estimator performs an implicit nonparametric estimate of the missing probability, and thus must behave similarly smooth as the trend function. Assumptions (i)-(iii) are met by a large class of  generating processes, including many Markov chains with smoothly varying transition probabilities. 

Assumption (iv) can be interpreted as an exogeneity assumption on the missing data generating mechanism, which for instance is satisfied if $\{D_t\}$ is independent of $\{u_t\}$. While this assumption could be argued to be restrictive, it does not appear to be problematic for our focus. Though inconclusive, some recent research has found evidence of a relation between greenhouse gases and the occurrence of cloud cover through climate change, see e.g. \citet{Norris}. As cloud cover may cause missing observations, our assumption might appear restrictive, but this kind of long-run dependence can be accommodated through the slowly varying trend affecting both $\{y_t\}$ and $\{D_t\}$. As such, the exogeneity assumption  mostly rules out short-run effects of ethane on cloud cover and vice versa, which we argue is reasonable.

\section{Inference on Trends} \label{sec:inference}
Our goal is to conduct inference on the trend function $m(\cdot)$ defined in Section \ref{sec:model}. We first describe point estimation of $m(\cdot)$, followed by our bootstrap method, and finally treat the construction of the confidence bands.

\subsection{Estimation of the Trend Function} \label{sec:estimation}
We consider local polynomial estimation which is common in the nonparametric regression literature. In particular, we focus on the local constant or Nadaraya-Watson estimator \citep{Nadaraya, Watson}, defined as
\begin{equation} \label{eq:estimator}
\begin{split}
\hat{m}(\tau) &= \argmin_{m(\tau)} \sum_{t=1}^n K\left(\frac{t/n-\tau}{h}\right) D_t \left\{y_t-m(\tau)\right\}^2 \\
&=\left[\sum_{t=1}^n K\left(\frac{t/n-\tau}{h}\right) D_t \right]^{-1} \sum_{t=1}^n K\left(\frac{t/n-\tau}{h}\right) D_t y_t,\qquad \tau\in \left(0,1\right),
\end{split}
\end{equation}
where $K(\cdot)$ is a kernel function and $h>0$ is a bandwidth, which should satisfy Assumptions \ref{as:kernel} and \ref{as:bandwidth} given below. Note that by construction of the $\{D_t\}$ series, the formulation in \eqref{eq:estimator} implies that the estimator only depends on the actually observed data.

\begin{assumption}\label{as:kernel}
$K(\cdot)$ is a symmetric, Lipschitz continuous function with compact support, where we define $\kappa_k = \int_{\mathbb{R}} K(\omega)^k \d \omega$, $\kappa (\tau) = \int_{\mathbb{R}} K(\omega) K(\omega - \tau) \d \omega$ and $\mu_k = \int_{\mathbb{R}} \omega^k K(\omega) \d \omega$.
\end{assumption}

\begin{assumption}\label{as:bandwidth}
The bandwidth $h=h(n)$ satisfies $n h^7 \rightarrow 0$ and $n h^2 \rightarrow \infty$ as $n\rightarrow \infty$.
\end{assumption}

Assumption \ref{as:kernel} is a standard assumption in the nonparametric kernel smoother literature, and is satisfied by many commonly used kernels. Assumption \ref{as:bandwidth} provides the range in convergence rates allowed for $h$ to ensure consistency and asymptotic normality of the kernel estimator.

The bandwidth, or smoothing parameter $h$ plays an important role. Large bandwidths produce a very smooth estimate, while small bandwidth produce a rough, wiggly trend estimate. Although Assumption \ref{as:bandwidth} gives some guidance, it does not provide us with a practical bandwidth choice. Data-driven bandwidth selection is therefore important for implementation. Leave-one-out cross-validation is the most popular data-based method for bandwidth selection, but it is designed for independent observations and therefore inappropriate for time series data. \citet{ChuMarron} show that in the presence of positive correlation, this criterion systematically selects very small bandwidths, producing estimates which are too wiggly. With negative correlation, bandwidths will be large and the estimate too smooth. Therefore, \citet{ChuMarron} propose to use a time series version of this criterion, called modified cross-validation (MCV). It is based on minimizing the criterion function $\frac{1}{n}\sum_{t=1}^n D_t \left(\hat{m}_{k,h}\left(\frac{t}{n}\right)-y_t\right)^2$ with respect to $h$, where
\begin{equation} \label{eq:MCV}
\hat{m}_{k,h}(\tau)=\frac{(n-2k-1)^{-1}\sum_{t:|t-\tau n|>k} K\left(\frac{t/n-\tau}{h}\right) D_t y_t}{(n-2k-1)^{-1}\sum_{t:|t-\tau n|>k} K\left(\frac{t/n-\tau}{h}\right) D_t}
\end{equation}
is a leave-$(2k+1)$-out version of the leave-one-out estimator of ordinary cross-validation, which leaves out the observation receiving the highest weight. Next to formal selection methods, visual inspection of the estimated trend function for a range of different bandwidths can help determining an appropriate bandwidth.

\begin{remark} \label{rem:MD_eps}
The assumption that $p(\cdot)$ is bounded away from zero implies that in every subinterval of $(0,1)$, we have enough observed data points as $n$ grows large. This assumption can be relaxed at the expense of more involved notation by restricting our attention to those compact subsets of $(0,1)$, where the probability of observing data is larger than or equal to some $\epsilon>0$.

In small samples, we may have points $\tau$ around which no data are observed in an $h$-neighborhood. As this is merely a small sample issue, for the theoretical analysis we implicitly assume that $n$ is large enough such that, for every $\tau \in (0,1)$, $\sum_{t=1}^n K\left(\frac{t/n - \tau}{h} \right) D_t \geq \epsilon^*$ for some $\epsilon^* > 0$. That is, sufficient data are available around $\tau$. This is not restrictive, as by Assumption \ref{as:bandwidth}, $h$ decreases more slowly than $n$ increases. In practice, the points around which insufficient data are available simply have to be excluded from the set of $\tau$ values considered for inference.
\end{remark}

\begin{remark} \label{rem:llest}
While the Nadaraya-Watson estimator locally approximates the trend function by a constant function, the local linear estimator locally fits a linear function to the data around a given point $\tau\in\left(0,1\right)$:
\begin{equation}\label{eq:ll_estimator}
\begin{split}
\left(
\begin{matrix}
\hat{m}_{ll} (\tau)\\
\hat{m}_{ll}^{(1)}(\tau)
\end{matrix}
\right) &= \argmin_{(m(\tau), m^{(1)}(\tau))} \; \sum_{t=1}^n K\left(\frac{t/n-\tau}{h}\right) D_t \left\{y_t-m(\tau)-m^{(1)}(\tau)\left(t/n-\tau)\right)\right\}^2\\
&= \left(\sum_{t=1}^n K\left(\frac{t/n-\tau}{h}\right)\mathbf{x}_t(\tau)\mathbf{x}_t(\tau)' D_t \right)^{-1} \sum_{t=1}^n K\left( \frac{t/n-\tau}{h} \right) \mathbf{x}_t(\tau) D_t y_t
\end{split}
\end{equation}
where $\mathbf{x}_{t} (\tau)\equiv\left(1, t/n - \tau \right)^\prime$. The local linear estimator is more accurate than the local constant estimator at points which are close to the boundaries of the sample. At these points, the local constant estimator suffers from boundary effects which the local linear estimator does not \citep{Cai,Fan}. Our analysis can be extended to the local linear estimator. However, notation becomes significantly more cumbersome and therefore proofs more complicated, so we focus on the local constant estimator in the theoretical part.
\end{remark}

\begin{remark} \label{rem:estmd}
An alternative way to specify the estimators in the presence of missing data is to work with unequally spaced data. Assuming we observe data on times $t_i$ for $i=1,\ldots, n_1$, where $n_1$ is the effective number of observations, the local constant estimator can be written as
\begin{equation*}
\hat{m}(\tau)=\left[\sum_{i=1}^{n_1} K\left(\frac{(t_i - t_1)/(t_{n_1} - t_1)-\tau}{h}\right) \right]^{-1}\sum_{i=1}^{n_1} K\left(\frac{(t_i - t_1)/(t_{n_1} - t_1)-\tau}{h}\right) y_{t_i}.
\end{equation*}
This formulation has the advantage that it does not require an underlying regular frequency at which the data are observed. However, in many applications it is not hard to define such an underlying frequency (days in our application), and then the two formulations are equivalent.

In the remainder of the paper we will continue to work with the first formulation, which leads to clearer notation and is easier to handle in the proofs, as the randomness in the missing data is modeled through the explicit $D_t$ variable rather than being ``hidden'' in the (now random) $t_i$ dates.
\end{remark}

\subsection{Autoregressive Wild Bootstrap}
\label{sec:bootstrap}
To construct confidence bands around the trend estimate, we modify the wild bootstrap, originally designed to handle heteroskedastic data \citep{DF}, to account fo serial dependence. The wild bootstrap generates bootstrap errors as $
z_t^* = \xi_t^* \hat{z}_t$, where $\hat{z}_t$ are residuals of the nonparametric trend regression. In the standard wild bootstrap, the random variables $\left\{\xi_t^*\right\}$ are i.i.d.~and thus, any dependence present in the data gets removed in the bootstrap errors. To overcome this drawback, \citet{Shao} proposed the \emph{dependent wild bootstrap} (DWB) in which $\left\{\xi_t^* \right\}$ are generated as $\ell$-dependent random variables with $\cov(\xi_s^*, \xi_t^*) = K_{DWB} \left(\frac{s - t}{\ell} \right)$, where $K_{DWB}(\cdot)$ is a kernel function. As the tuning parameter $\ell$, it has to be selected by the user.

Building on this idea, \citet{SU} propose the \emph{autoregressive wild bootstrap} (AWB) where $\left\{\xi_t^* \right\}$ is generated as an AR(1) process with parameter $\gamma = \gamma(n)$. The AWB has as advantage over the DWB that it is easier to implement and has a more intuitive interpretation. Moreover, as $\left\{\xi_t^* \right\}$ is not $\ell$-dependent, the AWB has the potential to capture more serial correlation and to be less sensitive to the choice of tuning parameter $\gamma$. In the context of unit root testing, \citet{SU} show that the AWB generally has a superior finite sample performance compared to the DWB. For these reasons, we mainly focus on the AWB in the following, although we consider the DWB in our simulation study as well. The AWB algorithm can be described as follows.

\begin{algorithm}[Autoregressive Wild Bootstrap]
$\phantom{1}$
\begin{enumerate}
\item Let $\tilde{m}(\cdot)$ be defined as in \eqref{eq:estimator}, but using bandwidth $\tilde{h}$. Obtain residuals $\hat{z}_t = D_t[y_t-\tilde{m}(t/n)]$ for $t=1,\ldots,n$.

\item For $0 < \gamma < 1$, generate $\nu_1^\ast,\ldots,\nu_n^\ast$ as i.i.d.~$\mathcal{N}(0,1-\gamma^2)$ and let $\xi_t^* = \gamma \xi_{t-1}^* + \nu_t^*$ for $t=2,\ldots,n$. Take $\xi_1^* \sim \mathcal{N}(0,1)$ to ensure stationarity of $\{\xi_t^*\}$.

\item Calculate the bootstrap errors $z_t^*$ as $z_t^* = D_t \xi_t^* \hat{z}_t$ and generate the bootstrap observations by $y_t^* = D_t [\tilde{m}(t/n) + z_t^*]$ for $t=1,\ldots,n$, where $\tilde{m}(\cdot)$ is the same estimate as in the first step.

\item Obtain the bootstrap estimator $\hat{m}^*(\cdot)$ as defined in \eqref{eq:estimator} using the bootstrap series $\{y_t^*\}$, with the same bandwidth $h$ as used for the original estimate $\hat{m}(\cdot)$.

\item Repeat Steps 2 to 4 $B$ times, and let 
\begin{equation} \label{quantile}
\hat{q}_{\alpha} (\tau) = \inf\left\{u \in \mathbb{R}:\mathbb{P}^* \left[\hat{m}^*(\tau) - \tilde{m}(\tau)\leq u\right] \geq \alpha\right\}
\end{equation}
denote the $\alpha$-quantile of the $B$ centered bootstrap statistics $\hat{m}^*(\tau) - \tilde{m}(\tau)$. These bootstrap quantiles are then used to construct confidence bands as described below.
\end{enumerate}
\end{algorithm}

Note that in Step 3 we only have to draw bootstrap observation for the dates where $D_t = 1$ and we observed the realization $y_t$; in the description the missing ones are artificially set to zero, but they are actually not used anywhere. In Step 2, we generate $\left\{\xi_t^* \right\}$ for all $t=1,\ldots,n$, although subsequently we only use the subset that corresponds to the actually observed data points. The missing data structure is preserved in the bootstrap sample, while the correlation between consecutive non-missing observations is determined only by their distance, which ensures a coherent bootstrap sample. In this way, the missing data structure is automatically taken into account in the bootstrap without any need for modifications.

Although we suggest to generate $\{\nu_t^*\}$ as a sequence of normally distributed random variables, inspection of the proofs shows normality is not needed; all one needs is a sequence of i.i.d.~random variables with $\E^* \nu_t^* = 0$, $\E^* \nu_t^{*2} = 1 - \gamma^2$ and $\E^{*4} \nu_t^{*4} < \infty$. Normality is simply a convenient option that is easy to implement; alternatively one could implement a variant of the Rademacher distribution (corrected for the right variance), which for the independent wild bootstrap has good properties \citep{DF}.

For the tuning parameter $\gamma$, we follow \citet{SU} and let $\gamma=\theta^{1/\ell}$ where $\ell$ is the ``block length'' parameter also found in the DWB and $0 < \theta < 1$ is a fixed parameter. This specification has the advantage that $\ell$ can be interpreted in a similar way as the block length parameter in a block bootstrap; its choice constitutes a trade-off between capturing more of the dependence with a large value of the tuning parameter, and allowing for more variation in the bootstrap samples with a smaller value for $\ell$. Additionally, it provides a convenient framework for studying the theoretical properties of our method. Specifically, we need $\ell \rightarrow \infty$ as $n \rightarrow \infty$, such that $\gamma \rightarrow 1$. This is analogous to the block (and dependent wild) bootstrap, where the block size must increase to capture more dependence when the sample size increases. Assumption \ref{as:ell} postulates the formal conditions that $\ell$ needs to satisfy. They imply that $\gamma \rightarrow 1$, but not too fast.

\begin{assumption}\label{as:ell}
The bootstrap parameter $\ell = \ell(n)$ satisfies $\ell\rightarrow \infty$ and $\ell /\sqrt{n h} \rightarrow 0$ as $n \rightarrow \infty$.
\end{assumption}

Note that we propose to use a different bandwidth $\tilde{h}$ in Step 1 of the algorithm. This is a common feature in the literature on bootstrap methods for nonparametric regression. By either selecting a larger (oversmoothing) or smaller bandwidth (undersmoothing) than used for the estimator, one can account for the asymptotic bias that is present in the local polynomial estimation, see \citet[Section 1.4]{HaHo} for an extensive literature review. While undersmoothing, such as used in the related paper by \citet{NP}, aims at making the bias asymptotically negligible, oversmoothing aims at producing a consistent estimator of the (non-negligible) bias. Both have advantages and disadvantages, see the extensive discussion in \citet{HaHo}. We follow \citet{Buhlmann} and consider a solution based on oversmoothing, which we find to work well in practice; also see \citet{HM}. After presenting our theoretical results in Section \ref{sec:theory}, Remark \ref{rem:oversmoothing} provides an intuition of why oversmoothing allows to consistently estimate the asymptotic bias.\footnote{\citet{HaHo} propose an alternative bootstrap approach that requires neither under- nor oversmoothing, however their approach only delivers pointwise intervals, and is therefore not considered in this paper.} We now state the formal conditions that $\tilde{h}$ must satisfy in Assumption \ref{as:bandwidth2}; one is that $h / \tilde{h} \rightarrow 0$ as $n \rightarrow \infty$, which ensures the oversmoothing.

\begin{assumption}\label{as:bandwidth2}
The oversmoothing bandwidth $\tilde{h} = \tilde{h}(n)$ satisfies $\max\left\{\tilde{h}, h / \tilde{h}, n h^{5} \tilde{h}^4 \right\} \rightarrow 0$ and $\ell \max\left\{\tilde{h}^{4}, 1/n\tilde{h} \right\} \rightarrow 0$ as $n \rightarrow \infty$.
\end{assumption}

\begin{remark}
There are a number of ways to choose the AWB parameter $\gamma$ in practice. Using the relation $\gamma = \theta^{1/\ell}$, one can fix $\theta$ and choose $\ell$ as a deterministic function of the sample size. \citet{SU} found that $\theta = 0.01$ paired with $\ell = 1.75 n^{1/3}$ performed well in their simulation study on AWB unit root testing; for the local polynomial estimation we might adapt this to let $\ell$ be a function of $nh$, see also Remark \ref{rem:pqr}. Alternatively, one may vary $\gamma$ over a range of reasonable values, choosing a value in the range where the bands are most stable (as a function of $\gamma$), akin to the minimum volatility method proposed by \citet{PRW}. Ideally, one would like to have a data-driven method for choosing an ``optimal'' $\gamma$. However, development of such a method requires a deeper study of higher-order asymptotic properties, which is outside the scope of the paper. 
\end{remark}

\begin{remark} \label{rem:pqr}
To give some intuition for the interaction between the three tuning parameters $h$, $\tilde{h}$ and $\ell$, consider taking $h = c n^{-p}$, for some $c > 0$ and $\frac{1}{7} < p < \frac{1}{2}$. This satisfies Assumption \ref{as:bandwidth}. Now let $\tilde{h} = c n^{-q}$, where $p > q$ as $h/\tilde{h} \rightarrow 0$. Furthermore, to satisfy $n h^5 \tilde{h}^4 \rightarrow 0$, we need that $5 p + 4 q > 1$. For $p \geq \frac{1}{5}$ this is satisfied for any $q$; with $p$ approaching $\frac{1}{7}$, $q > \frac{1}{14}$ suffices, which is not very strict.

Finally, take $\ell = c n^r$. From Assumption \ref{as:ell} we know that $r < \frac{1}{2} - p/2$, which for $p$ approaching $\frac{1}{2}$ means $r < \frac{1}{4}$, while for $p$ approaching $\frac{1}{7}$, it implies that $r < \frac{3}{7}$. To satisfy the restrictions from Assumption \ref{as:bandwidth2}, we need that $r < 4 q$ and $r < 1 - q$. As $q < p < \frac{1}{2}$, the second condition is non-binding. The first condition is only restrictive when $q$ is small, say in the vicinity of $\frac{1}{14}$. Then $r < \frac{2}{7}$, which is stricter than the condition implied by Assumption \ref{as:ell}. However, for $q \geq \frac{3}{28}$, this restriction is not binding.

As an example, consider the ``classical'' rates $p = \frac{1}{5}$ and $q = \frac{1}{9}$, cf.~\citet{Buhlmann}. These are allowed under our theory, and additionally imply that $0 < r < \frac{2}{5}$, therefore also allowing for $r = \frac{1}{3}$ as advocated in \citet{SU}.
\end{remark}

\begin{remark} \label{rem:MWB}
Instead of the autoregressive wild bootstrap, one could equally imagine a moving-average wild bootstrap (MA-WB), where $\xi_t^* = \sum_{j=0}^{\ell} \psi_{n,j} \nu_t^*$ and $\nu_t^*$ are i.i.d.~random variables as before. By letting $\psi_{n,j} \rightarrow 1$, if $n \rightarrow \infty$ and $j$ fixed, and $\psi_{n,j} \rightarrow 0$, if $n$ is fixed and $j \rightarrow \infty$, one could show validity of such an MA-WB as well. For practical purposes, one could take $\psi_{n,j} = f(j/\ell)$, for instance with $f(x) = 1 - x^r$ to simplify implementation.

Such a moving-average representation is closely related to the DWB, which can be seen as a two-sided MA process, with both bootstrap methods delivering $\ell$-dependent bootstrap samples. In his paper, \citet{Shao} shows the asymptotic equivalence of the variance estimator of the DWB with the tapered block bootstrap \citep{PapPol}. Similalry, one can show that the MA-WB has a close link to the extended tapered wild bootstrap of \citet{Shao2}, and depending on the distribution chosen for $\{\nu_t\}$, the smooth extended tapered block bootstrap of \citet{GLN15,GLN18}. If $\{\nu_t\}$ is drawn from a continuous distribution like the normal, we have an automatic smoothing in the wild bootstrap variants. As \citet{GLN15,GLN18} show that smoothing helps in the context of the block bootstrap, it may similarly do so for the wild bootstrap and provide a potential reason to prefer the normal distribution over discrete distributions like the Rademacher distribution.
\end{remark}

\subsection{Bootstrap Confidence Bands}
Pointwise bootstrap confidence intervals with a confidence level of $\left(1 - \alpha \right)$ for $m(\tau)$, are denoted by $I_{n,\alpha}^{(p)} (\tau)$ and constructed with the objective that
\begin{equation} \label{eq:pw_int}
\liminf_{n\to\infty} \mathbb{P}\left[\left(m(\tau)\in I_{n,\alpha}^{(p)}(\tau)\right)\right]\geq 1-\alpha \qquad \tau \in(0,1).
\end{equation}
Using our bootstrap algorithm, we can construct such pointwise intervals as
\begin{equation*}
I_{n,\alpha}^{(p)}(\tau) = \left[\hat{m}(\tau) - \hat{q}_{1-\alpha/2}(\tau), \hat{m}(\tau) - \hat{q}_{\alpha/2}(\tau) \right].
\end{equation*}

As these intervals are constructed separately for each $\tau$, links over time cannot be established with these intervals. Therefore, we next consider how to construct simultaneous confidence bands. Let $I_{n,\alpha}^G (\tau)$, for $\tau \in G$, denote a confidence band that is simultaneous over the set $G$. Formally, we seek to construct $I_{n,\alpha}^G (\tau)$ such that
\begin{equation} \label{eq:sim_int}
\liminf_{n\to\infty}\left[\mathbb{P} \left( m(\tau) \in I_{n,\alpha}^G(\tau) \quad \forall \tau \in G \right)\right] \geq 1 - \alpha.
\end{equation}
Our practical implementation follows the three-step procedure proposed by \citet{Buhlmann}:

\begin{enumerate}
\item For all $\tau \in G$, obtain pointwise quantiles $\hat{q}_{\alpha_p/2}(\tau),\hat{q}_{1-\alpha_p/2}(\tau)$ for varying $\alpha_p \in [1/B,\alpha]$.

\item Choose $\alpha_s= \argmin_{\alpha_p \in [1/B,\alpha]}\abs{\mathbb{P}^*\left[\hat{q}_{\alpha_p/2}(\tau) \leq \hat{m}^*(\tau)-\tilde{m}(\tau) \leq \hat{q}_{1-\alpha_p/2}(\tau) \quad \forall \tau \in G\right]-(1-\alpha)}$.

\item Construct the simultaneous confidence bands as \begin{equation*}
I_{n,\alpha}^G (\tau)=\left[\hat{m}(\tau)-\hat{q}_{1-\alpha_{s}/2}(\tau),\hat{m}(\tau)-\hat{q}_{\alpha_{s}/2}(\tau)\right]\qquad \tau \in G.
\end{equation*}
\end{enumerate}

In the second step, a pointwise error $\alpha_s$ is found for which a fraction of approximately $\left(1-\alpha\right)$ of all centered bootstrap estimates falls within the resulting confidence intervals, for all points of the set $G$. As such, the confidence intervals with pointwise coverage $\left(1-\alpha_s\right)$ become simultaneous confidence bands with coverage $\left(1-\alpha\right)$.\footnote{We provide R code to implement the estimator and bootstrap confidence bands on \href{http://www.stephansmeekes.nl}{www.stephansmeekes.nl}.}

\begin{remark}
As an alternative to the variable-size bands proposed by \citet{Buhlmann}, one could consider Kolmogorov-Smirnov-type simultaneous confidence bands of fixed size. They would be of the form $I^{\ast}_{\alpha}(\tau)=\left[\hat{m}(\tau)-t^{\ast}_{1-\alpha},\hat{m}(\tau)+t^{\ast}_{1-\alpha}\right]$, where the quantile $t^{\ast}_{1-\alpha}$ is determined as the $(1-\alpha)$-quantile of the distribution of the quantity $U^{\ast}_n=\sup_{\tau \in G}\left\{\left|\hat{m}^{\ast}(\tau)-\tilde{m}(\tau)\right|\right\}$. \citet{NP} establish the asymptotic validity of such simultaneous bands under serial independence. We do not go in this direction, because we believe confidence bands with variable width to be more informative. They have the feature of becoming wider at points with more variability and more narrow for periods with less variability. To obtain variable width intervals with the Kolmogorov-Smirnov approach, one has to estimate the variance of the estimator at each $\tau$ and bootstrap a pivotal quantity, see e.g.~\citet[Section 2.2]{NP}. This adds additional complications in order to achieve consistent variance estimation.
\end{remark}

\section{Asymptotic Theory}
\label{sec:theory}
We first provide the pointwise limiting normal distribution of the local constant estimator $\hat{m}(\cdot)$. Although the result is similar to the non-bootstrap part of Theorem 3.1 in \citet{Buhlmann}, we extend the asymptotic theory for the local constant estimator to allow for the presence of nonstationary volatility and missing data. As we feel this is a noteworthy result in its own right, we present it in Theorem \ref{th:asdis}.

\begin{theorem}\label{th:asdis}
Under Assumptions \ref{as:smooth}-\ref{as:bandwidth}, for any $\tau\in(0,1)$, we have as $n\rightarrow\infty$:
\begin{equation*}
\sqrt{nh}\left(\hat{m}(\tau)- m(\tau) - h^2 B_{as}(\tau) \right) \xrightarrow{d} \mathcal{N}\left(0,\sigma^2_{as}(\tau)\right),
\end{equation*}
where
\begin{equation}
B_{as}(\tau) = \mu_2 p(\tau)^{-1} \left[m p \right]^{(2)} (\tau) \qquad \text{and} \qquad \sigma_{as}^2 (\tau) = p(\tau)^{-1} \sigma(\tau)^2 \Omega_U \kappa_2, \label{eq:as_pars}
\end{equation}
\end{theorem}

The term $B_{as}(\tau)$ reflects the familiar asymptotic bias generally found in local polynomial estimators, although the exact form is different due the presence of the missing data parameter $p(\tau)$. The asymptotic variance $\sigma^2_{as}(\tau)$ is not only affected by $p(\tau)$, but also by the volatility process $\sigma^2 (\tau)$. If one were to use these distributions directly for inference, one would need to plug in consistent estimators of these nuisance parameters. However, as we show next, in the bootstrap these are automatically consistently estimated, and we have consistency of the autoregressive wild bootstrap method for the local constant estimator.

\begin{theorem}\label{th:basdis}
Under Assumptions \ref{as:smooth}-\ref{as:bandwidth2}, for any $\tau\in(0,1)$, we have as $n\rightarrow\infty$:
\begin{equation*}
\sqrt{nh}\left(\hat{m}^*(\tau)- \hat{m}(\tau) - h^2 B_{as}(\tau)\right) \xrightarrow{d^*}_p \mathcal{N}\left(0,\sigma^2_{as}(\tau)\right),
\end{equation*}
where $B_{as}(\tau)$ and $\sigma_{as}^2 (\tau)$ are defined in \eqref{eq:as_pars}.
\end{theorem}

The pointwise validity of the bootstrap confidence intervals in the sense of \eqref{eq:pw_int} follows directly from this pointwise convergence result. Note that, as the bias term $B_{as}(\tau)$ is the same in both theorems, it is consistently estimated by the bootstrap. As such, we do not need the bias to disappear, which happens when undersmoothing if $n h^{5} \rightarrow 0$, or to be $O(1)$, when $n h^{5} \rightarrow c$. Even if $n h \rightarrow \infty$, and the asymptotic bias dominates the stochastic variation, the bootstrap correctly mimics this and can be used for asymptotically valid inference. As such, we can relax the assumption in \citet[p.~55]{Buhlmann} that $h \sim C n^{-1/5}$ to allow for a wider range of bandwidths. In practice, this means that the bootstrap provides additional protection against a misspecified bandwidth, by letting the widths of confidence bands automatically adapt.

Next, to study the validity of simultaneous confidence bands as in \eqref{eq:sim_int}, we consider $h$-neighborhoods around time points $\tau$. We do so because estimates $\hat{m}(\tau_1)$ and $\hat{m}(\tau_2)$ are asymptotically independent for $\tau_1 \neq \tau_2$ being two fixed distinct time points. When the distance between $\tau_1$ and $\tau_2$ is of order $h$, the estimators show a non-zero correlation. Therefore a major benefit of ``zooming in'' on local $h$-neighborhoods is that we can study how the bootstrap mimics the correlation between close points, a feature which is lost when considering simultaneity globally.

\begin{theorem} \label{th:uniform}
For any $\tau_0 \in (0,1)$, let
\begin{align*}
Z_{\tau_0,n}(\tau) &= \sqrt{nh}\left(\hat{m}(\tau_0+\tau h)-m(\tau_0+\tau h)\right), &\quad
Z_{\tau_0,n}^{*} (\tau) &= \sqrt{nh} \left(\hat{m}^{*}(\tau_0 + \tau h) - \tilde{m} (\tau_0 + \tau h) \right).
\end{align*}
Then, under Assumptions \ref{as:smooth}-\ref{as:bandwidth2}, we have for all $\tau_0 \in (0,1)$
\begin{align*}
\left\{Z_{\tau_0,n}(\tau) - B_{as}(\tau_0) \right\}_{\tau \in [-1,1]} &\Rightarrow \left\{W(\tau)\right\}_{\tau\in [-1,1]},\\
\left\{Z_{\tau_0,n}^*(\tau) - B_{as}(\tau_0)\right\}_{\tau\in [-1,1]} &\Rightarrow_p \left\{W(\tau)\right\}_{\tau\in [-1,1]},
\end{align*}
where $\left\{W(\tau)\right\}_{\tau\in [-1,1]}$ is a Gaussian process with $\E W(\tau) = 0$ and
\begin{align*}
\cov(W(\tau_1),W(\tau_2))&= \sigma_{W,\tau_0} (\tau_1, \tau_2) = p(\tau_0)^{-1} \sigma(\tau_0)^2 \Omega_U \kappa (\tau_1 - \tau_2).
\end{align*}
Here, $\Rightarrow$ denotes weak convergence in the space of continuous real-valued functions on $\left[-1,1\right]$ endowed with the sup-norm.
\end{theorem}

Theorem \ref{th:uniform} establishes the uniform validity of the bootstrap within an $h$-neighborhood around any point $0<\tau_0<1$, where, since $h=o(1)$, we assume without loss of generality that $m(\tau_0+\tau h)$ is always defined. Note that the interval $[-1,1]$ is mainly chosen out of convenience, and the results can trivially be shown to hold over any interval $[\tau_0 - a h, \tau_0 + b h]$ with $0<a,b < \infty$. Moreover, it follows directly from Theorem \ref{th:uniform} that the bootstrap will be valid uniformly on sets that contain a union of any finite number of such $h$-neighborhoods, see e.g.~\citet[Corollary 3.3]{Buhlmann}. While, in finite samples, one can always take $h$ and the intervals such that the full sample is covered in $G$, this kind of ``too large'' simultaneity should be considered with caution, as this is not what the asymptotic analysis covers. Although simultaneity over such local sets might appear less attractive than simultaneity over the whole sample, it can nevertheless be of great interest in applications. For example, constructing confidence bands with simultaneous coverage over two time periods - one located early in the sample and the other one at the end - is useful when judging if there was an upward (or downward) movement of the trend at the end of the time period when compared to the beginning. This allows the empirical researcher to draw conclusions about developments spanning time stretches, which is not possible with pointwise confidence intervals.

\begin{remark}\label{rem:oversmoothing}
To provide some intuition for the required oversmoothing with bandwidth $\tilde{h}$ in the bootstrap, note that from Theorem \ref{th:asdis} (or more formally Lemma \ref{lem:m_decomp}) we can deduce that the estimator used in the first step of the bootstrap algorithm satisfies
\begin{equation} \label{eq:mtilde_dec}
\tilde{m} (\tau) - m(\tau) = \tilde{h}^2 B_{as} (\tau) + Z_n  (\tau) / \sqrt{n\tilde{h}} + o_p \left(1/\sqrt{ n\tilde{h}}\right),
\end{equation}
where $Z_n (\tau) \xrightarrow{d} N(0, \sigma_{as}^2 (\tau)$. Furthermore, letting $w_{t,n} (\tau) = K\left(\frac{t/n - \tau} {h} \right) D_t / \left[\sum_{t=1}^n K\left(\frac{t/n - \tau} {h} \right) D_t \right]$, and using that $y_t^* = \tilde{m} (t/n) + z_t^*$, we can write
\begin{equation*}
\begin{split}
\hat{m}^*(\tau) - \tilde{m} (\tau) &= \left[\sum_{t=1}^n w_{t,n} (\tau) \tilde{m} (t/n) - \tilde{m} (\tau) \right] + \sum_{t=1}^n w_{t,n} (\tau) z_t^* + o_p \left(1/\sqrt{ n\tilde{h}}\right).
\end{split}
\end{equation*}
While the second term, $\sum_{t=1}^n w_{t,n} (\tau) z_t^*$, mimics the stochastic variation in the trend estimation, and ensures the asymptotic normality of the bootstrap trend estimator, the bias arises from the first term, $\sum_{t=1}^n w_{t,n} (\tau) \tilde{m} (t/n) - \tilde{m} (\tau)$. Using \eqref{eq:mtilde_dec}, this term can be decomposed as
\begin{equation*}
\begin{split}
\sum_{t=1}^n w_{t,n} (\tau) \tilde{m} (t/n) - \tilde{m} (\tau)
&= \left[\sum_{t=1}^n w_{t,n} (\tau) m (t/n) - m(\tau) \right] + \tilde{h}^2 \left[\sum_{t=1}^n w_{t,n} (\tau) B_{as} (t/n) - B_{as} (\tau) \right]\\
&\quad + \left[\sum_{t=1}^n w_{t,n} (\tau) Z_n (t/n) - Z_n (\tau)\right]/\sqrt{ n\tilde{h}} + o_p \left(1/\sqrt{ n\tilde{h}}\right).
\end{split}
\end{equation*}
The asymptotic bias arises, as for the original estimator, from the first term of the decomposition. As shown in the proof of Lemma \ref{lem:mb_decomp}, by the smoothness of $B_{as}(\tau)$, the second term converges to zero, thus canceling out the bias in $\tilde{m}(\tau)$. However, to make the third, stochastic, term vanish when multiplying by $\sqrt{nh}$, it must be that $h / \tilde{h} \rightarrow 0$, such that this term is of small enough magnitude. This is achieved by oversmoothing.
\end{remark}

\begin{remark}\label{rem:boundary}
Although it is not visible from the theorems -- as we only consider pointwise $\tau$ (or $\tau_0$ in Theorem \ref{th:uniform}) -- $\tau$ needs to be bounded away from the boundaries (0 and 1) to make the results hold. As the estimator exhibits edge effects, the quality of $\tilde{m}(\tau)$ can only be guaranteed for $\tau$ away from 0 or 1. Formally, we need to take a small $\delta > 0$ and then consider $\tau \in [\delta, 1- \delta]$. Consequently, in the bootstrap we then obtain the limit distribution in a slightly smaller set $\tau \in [\delta^*, 1- \delta^*]$, for some $\delta^* > \delta$. However, as we can take $\delta$ and $\delta^*$ as small as we like, this does not affect the pointwise statements of the theorems.\footnote{See Lemmas \ref{lem:m_decomp} and \ref{lem:mb_decomp} in Appendix A for statements where these constants do appear.}

Because of the above reasons, \citet[p.~53]{Buhlmann} suggests in his bootstrap algorithm to obtain residuals $\hat{z}_t = D_t[y_t-\tilde{m}(t/n)]$ for $t=[n \delta] + 1,\ldots, [n (1 - \delta)]$, as the residuals too close to the boundary may contaminate the bootstrap sample when they are sampled. This is not a problem for our method, as the AWB doesn't involve resampling: boundary residuals remain at the boundary in the bootstrap sample, and can therefore not affect results away from the boundaries.
\end{remark}

\section{Simulation Study} \label{sec:simulation}
For the simulation exercise, we simulate time series with a trending behavior in both mean and variance, inspired by patterns observed in climatological time series, and allow for similar patterns of missing data. We will first describe the setting and then present and discuss the results.

\subsection{Simulation Setup}
We consider the following smooth transition model:
\begin{equation}
y_t = m(t/n) + \sigma_t u_t, \qquad m(\tau) = \beta_1 \tau + \beta_2 \tau G(\tau,\lambda,c),
\label{eq:model_simu}
\end{equation}
where for $\lambda>0$,
\begin{equation}
G(\tau,\lambda,c)=\left(1+\exp\left\{-\lambda(\tau-c)\right\}\right)^{-1}.
\label{eq:trans_function}
\end{equation}
The error term $\left\{u_t\right\}$ follows an ARMA$(1,1)$ model
\begin{equation}
u_t=\phi u_{t-1}+\psi \epsilon_{t-1}+\epsilon_t\hspace{20mm}\epsilon_t\sim \mathcal{N}\left(0,\frac{(1-\phi^2)/4}{1+\psi^2-2\phi\psi}\right),
\end{equation}
where we vary the parameters $\phi$ and $\psi$ to investigate the impact of serial correlation on our method. The variance of $\epsilon_t$ is normalized such that the signal to noise ratio does not depend on the specific choice of the AR and MA parameter. Furthermore, we introduce heteroskedasticity with the process $\left\{\sigma_t\right\}$. We consider two scenarios, where  $\sigma_t$ is constant over time or $\sigma_t = \sigma(t/n)$, with the volatility process $\sigma(\tau)$ given by
\begin{equation}\label{eq:variance}
\sigma (\tau) = \sigma_{0} + (\sigma_{\ast}-\sigma_{0})(\tau)+a\cos\left(2\pi k \tau \right).
\end{equation}
Equation \eqref{eq:model_simu} is a shifting mean model as considered by \citet{GT}, and can be seen as a smooth transition version of a broken trend model with one break. The function $G(\tau,\lambda,c)$ as given in \eqref{eq:trans_function} is the transition function with time as transition variable. Its inputs apart from time are the location of the shift -- the parameter $c$ -- as well as the smoothness of the shift, determined by $\lambda$. For large values of $\lambda$ the shift happens almost instantaneous, while it is smoother for smaller values of this parameter. In our simulations, we fix $\lambda=10$. The other parameters of our DGP will be chosen in such a way that the time series experiences a downward trend during the first three quarters which turns into a steeper upward trend in the last quarter.

\begin{figure}[tbh]
\centering
\subfigure[The trend function $m(\tau)$]
{\includegraphics[width=0.47\linewidth]{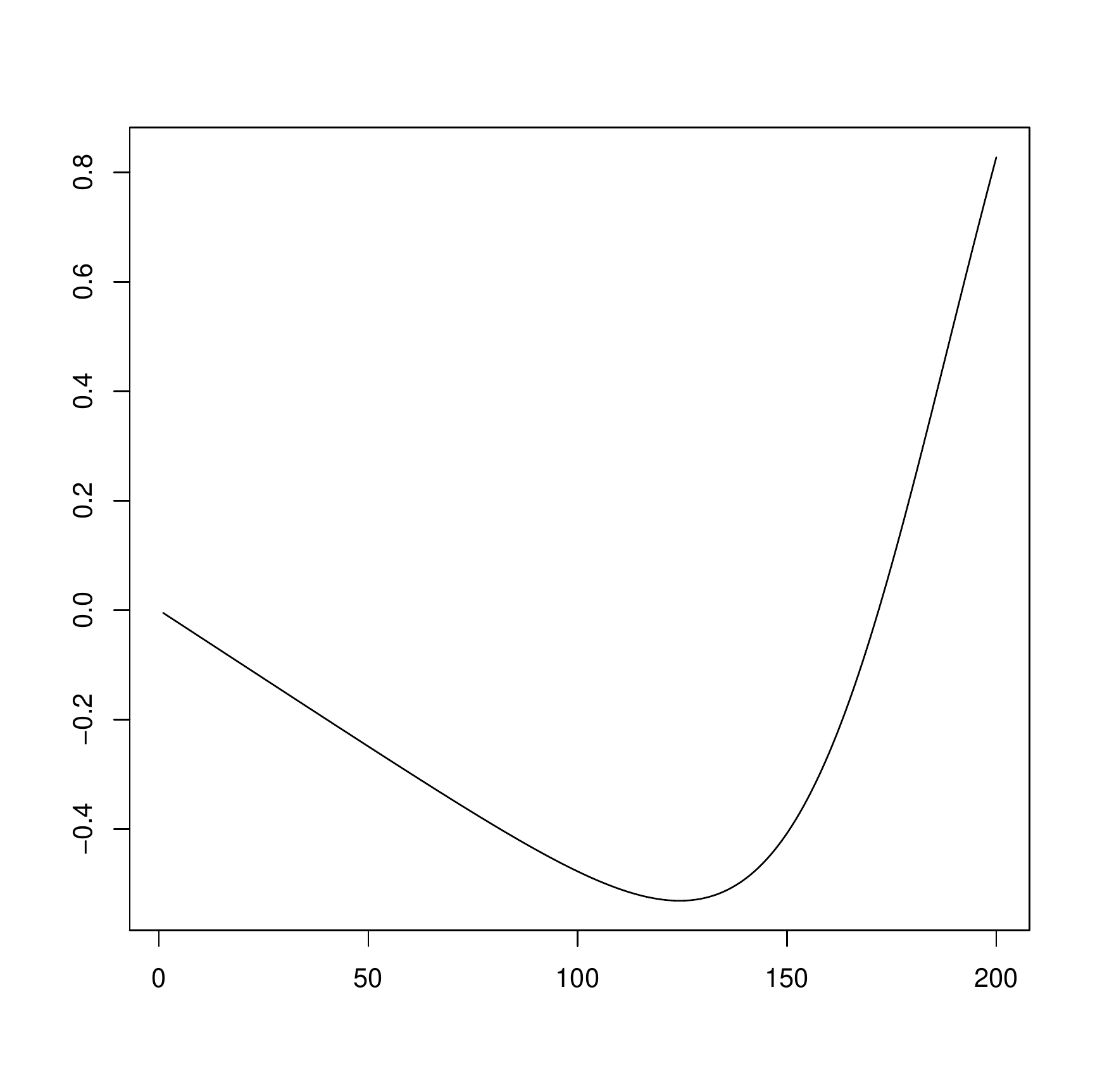}
\label{fig:trend}}
\subfigure[The variance process $\sigma(\tau)$ with $k=4$ and $a=0.5$]
{\includegraphics[width=0.47\linewidth]{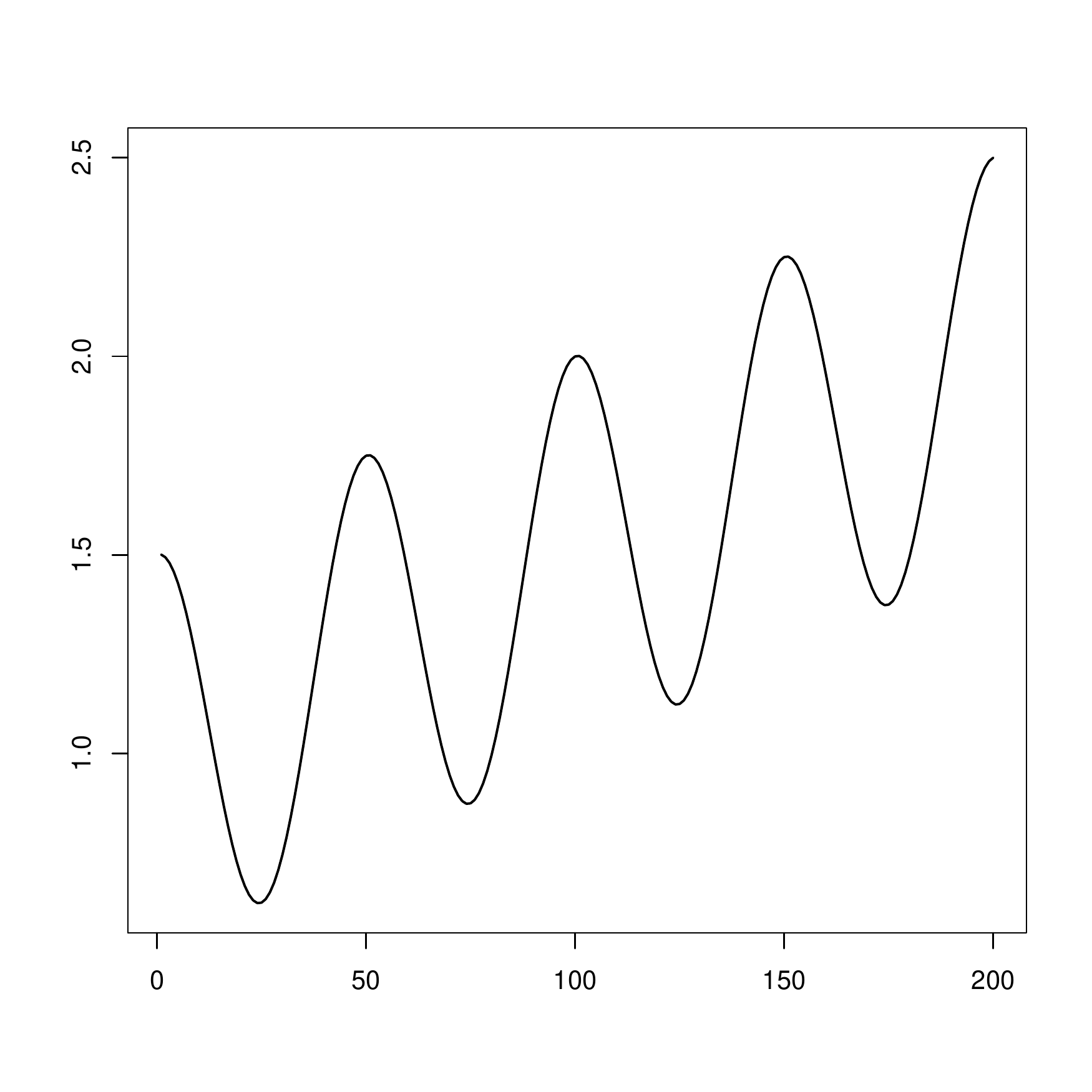}
\label{fig:sigma}}	
\caption{Trends in mean and variance in the simulation DGP}
\label{fig:DGP}
\end{figure}

This mimics the general pattern which is expected to occur in atmospheric ethane time series and therefore fits our application well. More specifically, this means we set the location of the shift to occur at $c=0.9$. The slope of the trend gradually changes from $\beta_1=-1$ before the shift to $\beta_2=2.5$ after the shift. This is illustrated in Figure \ref{fig:trend}. For the variance process, inspired by the series considered in the empirical application, we consider a cyclical component with trend. We have to choose four parameters in \eqref{eq:variance}: the start and end point of the trend -- $\sigma_0$ and $\sigma_{\ast}$ -- as well as the specifics of the cyclical component. The parameter $a$ fixes the amplitude of the cycle, while $k$ determines how many cycles there are. We set $\sigma_0=1$, $\sigma_{\ast}=2$ and consider different combinations of values for $a$ and $k$. We let $a=0.3,0.5,0.7$ and $k=2,3,4$. An example of this process is displayed in Figure \ref{fig:sigma}.

We also consider different degrees of dependence by varying the AR and MA parameters. For the AR parameter we take $\phi=0,0.2,0.5,-0.5$, while the MA parameter varies between $\psi=0$, $\psi=0.2$ and $\psi=0.5$. We only look at pure AR or MA processes with these coefficient values. The different specifications will be abbreviated in the tables with self-explanatory names, e.g.~we write $AR_{-0.5}$ for $\phi=-0.5$, $\psi=0$ and use $MA_{0.5}$ when $\phi=0$ and $\psi=0.5$.

In addition, we consider cases of missing data for which we generate a missing pattern that is representative for the ethane data. We implement a first-order Markov Chain for $D_t$ with transition probabilities
\begin{equation} \label{eq:missings_data}
\bordermatrix{
& D_{t} = 0 & D_t = 1  \cr
D_{t-1} = 0 & 0.80 & 0.20 \cr
D_{t-1} = 1 & 0.45 & 0.55 \cr
},
\end{equation}
which are estimated from the ethane time series considered in Section \ref{sec:application}. This transition matrix results in an average fraction of around 70\% missing observations.

In the estimation step, we apply the local constant estimator based on the Epanechnikov kernel which is given by the function $K(x)=\frac{3}{4}(1-x^2)\mathbbm{1}_{\left\{|x|\leq 1\right\}}$. For the bandwidth parameter $h$ we use $h=0.02$, $h=0.04$ and $h=0.06$. In the first step of the bootstrap procedure we follow the recommendation of \citet{Buhlmann} to use $\tilde{h}=Ch^{5/9}$ with $C=2$. In the second step of the bootstrap, we consider different values for the AR parameter $\gamma$; next to $\gamma=0$, which reduces the AWB to a standard wild bootstrap (WB), we also consider $\gamma=0.2,0.4,0.6$.

For each specification, we run 5000 Monte Carlo simulations. We report average pointwise as well as simultaneous coverage for a sample size of $n=200$, based on $B=999$ bootstrap replications. For ease of comparison, we choose the sample size in cases with missing data such that we have approximately 200 data points remaining. Given the large fraction of missings, an expected effective sample size of 200 translates into a original sample size for our Markov Chain of $n=666$.

The nominal coverage in all cases is 95\%. We also report the average median length of the confidence intervals in parenthesis underneath the respective coverage. For simultaneous coverage, the trend curve has to lie within the confidence bands for all points of the considered set $G$, for which we take the two sets $G_{sub}$ and $G$ considered by \citet{Buhlmann}, where $G_{sub}=U_1(h)\cup U_4(h)$ and $G=\bigcup_{i=1}^4U_i(h)$, with $U_i(h)=\left\{(i/5)-h+j/100;\; j=0,...,[200h]\right\}$.

To compare the performance of our AWB method to related bootstrap methods, we also implement the dependent wild bootstrap (DWB) and the sieve wild bootstrap (SWB), with a standard normal distribution for generation of the wild bootstrap errors. Since the SWB cannot easily be adapted to work with missing data, we provide results for this method only for the cases with no missing data. For the DWB, we convert the $\gamma$ parameter into the corresponding value for the tuning parameter $\ell$, using the formula $\gamma=\theta^{1/\ell}$. \citet{SU} found in their simulation study that $\theta=0.01$ provides a sensible conversion between the AWB and DWB, in the sense of yielding comparable performance of the two methods, therefore we use $\theta=0.01$ as well to convert the AWB parameter into the DWB parameter. This tuning parameter does not exist in the case of the SWB; instead, the lag length has to be selected, for which we use AIC.

\subsection{Simulation Results}

First, we report results for equally spaced data with no missing observations and a variance process with the same specifications ($k=4$ and $a=0.5$) as displayed in Figure \ref{fig:sigma}. The results on pointwise coverage are given in Table \ref{pointwise}, while Tables \ref{simultaneous1} and \ref{simultaneous2} show simultaneous coverage probabilities for the two sets $G.sub$ and $G$, respectively. The tables consist of three main blocks, one for each bandwidth. Within each block, the individual rows contain results for different combinations of AR and MA parameters. Results for other choices of the variance parameters $k$ and $a$, including the homoskedastic case, are available in Supplementary Appendix C. Qualitatively, these settings yield the same conclusions as the ones considered here.

Table \ref{pointwise} shows that the autoregressive wild bootstrap provides confidence intervals with good pointwise coverage in the presence of heteroskedasticity and mild autocorrelation. For the independent case and the cases with negative or small positive correlation, the coverage probabilities are close to the nominal level. The only specifications for which the coverage lies below the nominal level are when $\phi=0.5$ and $\psi=0.5$. In these cases, the data deviate from the trend line in clusters due to the strong positive correlation. This causes the nonparametric estimate to go through these clusters and thus, to deviate significantly from the true trend. The confidence bands are in these situations not wide enough to cover the true trend, resulting in too low coverage. Interestingly, all methods/tuning parameters are similarly affected.

Concerning the autoregressive parameter of the wild bootstrap, we can observe that whenever the data are serially correlated, the autoregressive wild bootstrap ($\gamma\neq 0$) provides better coverage than the standard wild bootstrap ($\gamma=0$). In addition, with stronger correlation, a larger value for $\gamma$ should be preferred, except that the case $\gamma=0.6$ provides consistently lower coverage, indicating that simply going for a very large value of $\gamma$ is not sensible in practice. However, even if we do see these patterns, in general, the coverage probabilities do not vary substantially with the autoregressive parameter and therefore, the sensitivity to this parameter appears to be fairly limited.

When we look at the different blocks of Table \ref{pointwise}, the bootstrap shows a similar overall performance regardless of the value we select for the bandwidth parameter. Since the bandwidth plays such an important role in nonparametric estimation, yet there are no fully satisfactory ways to select optimal bandwidth from the data in most applications, robustness to bandwidth ``misspecifcation'' is an important finding. This implies that the bootstrap can correct for poorly chosen bandwidths.

We observe similar patterns in Tables \ref{simultaneous1} and \ref{simultaneous2}, while overall coverage is lower for the set $G$ than for $G_{sub}$. This is not surprising, since the former set covers twice as many points as the latter. Interestingly, the confidence bands are consistently more narrow with $G$ than they are with $G_{sub}$. This appears counterintuitive at first, as $G$ is twice as large as $G_{sub}$. However, while $G_{sub}$ is made up entirely of points relatively close to the boundaries, for which estimation is more variable, $G$ additionally contains ``stable'' regions closer to the center. It may be that the stability of these regions has an offsetting effect compared with the boundary regions, reducing the size of the intervals. As a side effect, overall coverage is also reduced. As such, if one is only interested in coverage near the beginning and the end of the sample, it may be wiser to only attempt to achieve uninformity over these regions, rather than over the full sample.

Similar to the pointwise coverage results, the simultaneous coverage is close to the nominal level for the two cases with negative correlation as well as the independent case. Weak positive correlation can also be handled decently. The cases $\phi=0.5$ and $\psi=0.5$ are more problematic, as coverage drops to around 60\% for $G$ and 70\% for $G_{sub}$. In these cases, the smallest bandwidth $h=0.02$ seems to be preferred.

Comparing the AWB to the other two bootstrap methods, we can see that in almost all cases, AWB and DWB show similar results and they outperform the sieve version. Often, the AWB results in slightly higher coverage with shorter intervals. Only when $\gamma=0.6$, the DWB displays better coverage. Further increases of this parameter did not lead to improvements. An exception is the DGP with negative correlation, where the DWB displays coverage that is too high, independent of the choice of $\gamma$. In such cases, the AWB is often more accurate for $\gamma=0.4$. In all other cases, we see that for both methods, the best performance is similar in magnitude but obtained at a different value of the tuning parameter. Since the tuning parameters do not have exactly the same meaning in both methods, there is no reason to expect identical variation. We chose $\theta=0.01$ to link the two methods; changing this value will likely change the relation between the methods as well.

\begin{table}
\centering
\begin{tabular}{@{}lllllllllll@{}}
\toprule
 &  & \multicolumn{1}{c}{$\gamma=0$} & \multicolumn{2}{c}{$\gamma=0.2$} & \multicolumn{2}{c}{$\gamma=0.4$} & \multicolumn{2}{c}{$\gamma=0.6$} & \multicolumn{1}{c}{$-$} \\ \cmidrule(l){3-3} \cmidrule(l){4-5} \cmidrule(l){6-7} \cmidrule(l){8-9} \cmidrule(l){10-10}
$h$ & \multicolumn{1}{l}{DGP}  & WB & AWB & DWB & AWB  & DWB & AWB & DWB & SWB\\
\midrule
\multirow{ 12}{*}{$0.02$} & $0$ & 0.946 & 0.952 & 0.946 & 0.942 & 0.945 & 0.905 & 0.945 & 0.925 \\
&& (0.732) & (0.715) & (0.733) & (0.658) & (0.733) & (0.555) & (0.733) & (0.719) \\
 & $AR_{0.2}$ & 0.902 & 0.918 & 0.902 & 0.912 & 0.902 & 0.871 & 0.902 & 0.861 \\
&& (0.710) & (0.711) & (0.711) & (0.672) & (0.711) & (0.582) & (0.711) & (0.688) \\
 & $AR_{0.5}$ & 0.778 & 0.818 & 0.776 & 0.828 & 0.778 & 0.791 & 0.777 & 0.672 \\
&& (0.604) & (0.630) & (0.605) & (0.621) & (0.604) & (0.562) & (0.603) & (0.536) \\
 & $AR_{-0.5}$ & 0.989 & 0.988 & 0.989 & 0.982 & 0.989 & 0.958 & 0.989 & 0.960 \\
&& (0.623) & (0.575) & (0.624) & (0.499) & (0.624) & (0.397) & (0.624) & (0.552) \\
 & $MA_{0.2}$ & 0.910 & 0.925 & 0.911 & 0.918 & 0.909 & 0.879 & 0.910 & 0.872 \\
&& (0.712) & (0.712) & (0.712) & (0.671) & (0.712) & (0.580) & (0.712) & (0.690) \\
 & $MA_{0.5}$ & 0.864 & 0.891 & 0.864 & 0.891 & 0.864 & 0.854 & 0.864 & 0.773 \\
&& (0.638) & (0.657) & (0.639) & (0.636) & (0.638) & (0.564) & (0.639) & (0.571) \\
\midrule
\multirow{ 12}{*}{$0.04$} & $0$ & 0.950 & 0.952 & 0.951 & 0.940 & 0.951 & 0.905 & 0.951 & 0.922 \\
&& (0.567) & (0.561) & (0.568) & (0.527) & (0.568) & (0.460) & (0.568) & (0.548) \\
 & $AR_{0.2}$ & 0.903 & 0.911 & 0.901 & 0.905 & 0.901 & 0.872 & 0.902 & 0.862 \\
&& (0.553) & (0.563) & (0.553) & (0.544) & (0.553) & (0.487) & (0.554) & (0.525) \\
 & $AR_{0.5}$ & 0.768 & 0.803 & 0.766 & 0.818 & 0.767 & 0.795 & 0.765 & 0.667 \\
&& (0.478) & (0.511) & (0.479) & (0.519) & (0.478) & (0.489) & (0.477) & (0.409) \\
 & $AR_{-0.5}$ & 0.996 & 0.994 & 0.996 & 0.989 & 0.996 & 0.969 & 0.996 & 0.970 \\
&& (0.492) & (0.457) & (0.492) & (0.405) & (0.492) & (0.338) & (0.493) & (0.423) \\
 & $MA_{0.2}$ & 0.912 & 0.922 & 0.913 & 0.913 & 0.912 & 0.880 & 0.912 & 0.873 \\
&& (0.553) & (0.563) & (0.554) & (0.542) & (0.554) & (0.485) & (0.554) & (0.526) \\
 & $MA_{0.5}$ & 0.867 & 0.890 & 0.869 & 0.890 & 0.868 & 0.861 & 0.867 & 0.784 \\
&& (0.501) & (0.528) & (0.503) & (0.522) & (0.502) & (0.479) & (0.503) & (0.436) \\
\midrule
\multirow{ 12}{*}{$0.06$} & $0$ & 0.955 & 0.957 & 0.956 & 0.946 & 0.956 & 0.917 & 0.956 & 0.925 \\
&& (0.492) & (0.492) & (0.493) & (0.471) & (0.493) & (0.422) & (0.493) & (0.465) \\
 & $AR_{0.2}$ & 0.910 & 0.917 & 0.907 & 0.913 & 0.907 & 0.887 & 0.909 & 0.865 \\
&& (0.480) & (0.494) & (0.481) & (0.486) & (0.480) & (0.447) & (0.481) & (0.444) \\
 & $AR_{0.5}$ & 0.781 & 0.814 & 0.777 & 0.831 & 0.779 & 0.815 & 0.776 & 0.669 \\
&& (0.420) & (0.455) & (0.421) & (0.468) & (0.420) & (0.453) & (0.420) & (0.347) \\
 & $AR_{-0.5}$ & 0.997 & 0.996 & 0.997 & 0.992 & 0.997 & 0.977 & 0.997 & 0.979 \\
&& (0.432) & (0.408) & (0.432) & (0.373) & (0.432) & (0.325) & (0.433) & (0.358) \\
 & $MA_{0.2}$ & 0.918 & 0.928 & 0.920 & 0.921 & 0.918 & 0.894 & 0.918 & 0.875 \\
&& (0.480) & (0.495) & (0.481) & (0.485) & (0.481) & (0.444) & (0.481) & (0.446) \\
 & $MA_{0.5}$ & 0.878 & 0.900 & 0.879 & 0.901 & 0.878 & 0.878 & 0.878 & 0.789 \\
&& (0.439) & (0.467) & (0.440) & (0.470) & (0.440) & (0.441) & (0.440) & (0.370) \\
\bottomrule
\end{tabular}
\caption{Pointwise coverage probabilities (average median interval length) for $k=4$ and $a=0.5$.}
\label{pointwise}
\end{table}

\begin{table}
\centering
\begin{tabular}{@{}lllllllllll@{}}
\toprule
 &  & \multicolumn{1}{c}{$\gamma=0$} & \multicolumn{2}{c}{$\gamma=0.2$} & \multicolumn{2}{c}{$\gamma=0.4$} & \multicolumn{2}{c}{$\gamma=0.6$} & \multicolumn{1}{c}{$-$} \\ \cmidrule(l){3-3} \cmidrule(l){4-5} \cmidrule(l){6-7} \cmidrule(l){8-9} \cmidrule(l){10-10}
$h$ & \multicolumn{1}{l}{DGP}  & WB & AWB & DWB & AWB  & DWB & AWB & DWB & SWB\\
\midrule
\multirow{ 12}{*}{$0.02$} & $0$ & 0.941 & 0.944 & 0.936 & 0.934 & 0.935 & 0.870 & 0.937 & 0.917 \\
&& (0.672) & (0.656) & (0.673) & (0.604) & (0.672) & (0.510) & (0.673) & (0.661) \\
 & $AR_{0.2}$ & 0.884 & 0.908 & 0.883 & 0.888 & 0.879 & 0.822 & 0.884 & 0.805 \\
&& (0.652) & (0.653) & (0.653) & (0.617) & (0.653) & (0.534) & (0.653) & (0.632) \\
 & $AR_{0.5}$ & 0.700 & 0.759 & 0.689 & 0.769 & 0.700 & 0.687 & 0.700 & 0.427 \\
&& (0.555) & (0.578) & (0.556) & (0.570) & (0.555) & (0.515) & (0.553) & (0.493) \\
 & $AR_{-0.5}$ & 0.982 & 0.981 & 0.983 & 0.974 & 0.987 & 0.920 & 0.984 & 0.985 \\
&& (0.572) & (0.528) & (0.573) & (0.459) & (0.573) & (0.365) & (0.573) & (0.507) \\
 & $MA_{0.2}$ & 0.899 & 0.913 & 0.891 & 0.900 & 0.888 & 0.835 & 0.895 & 0.820 \\
&& (0.654) & (0.653) & (0.654) & (0.616) & (0.654) & (0.533) & (0.654) & (0.634) \\
 & $MA_{0.5}$ & 0.830 & 0.860 & 0.821 & 0.853 & 0.810 & 0.792 & 0.821 & 0.619 \\
&& (0.586) & (0.604) & (0.587) & (0.584) & (0.586) & (0.518) & (0.587) & (0.524) \\
\midrule
\multirow{ 12}{*}{$0.04$} & $0$ & 0.930 & 0.931 & 0.929 & 0.907 & 0.932 & 0.818 & 0.933 & 0.902 \\
&& (0.521) & (0.516) & (0.522) & (0.484) & (0.522) & (0.423) & (0.522) & (0.504) \\
 & $AR_{0.2}$ & 0.849 & 0.858 & 0.836 & 0.834 & 0.835 & 0.759 & 0.848 & 0.774 \\
&& (0.508) & (0.517) & (0.508) & (0.500) & (0.508) & (0.448) & (0.509) & (0.483) \\
 & $AR_{0.5}$ & 0.592 & 0.648 & 0.564 & 0.685 & 0.577 & 0.604 & 0.573 & 0.376 \\
&& (0.439) & (0.470) & (0.440) & (0.476) & (0.440) & (0.449) & (0.439) & (0.376) \\
 & $AR_{-0.5}$ & 0.997 & 0.991 & 0.995 & 0.980 & 0.994 & 0.917 & 0.995 & 0.993 \\
&& (0.452) & (0.420) & (0.453) & (0.373) & (0.452) & (0.310) & (0.453) & (0.390) \\
 & $MA_{0.2}$ & 0.859 & 0.877 & 0.857 & 0.860 & 0.864 & 0.774 & 0.863 & 0.790 \\
&& (0.508) & (0.517) & (0.509) & (0.498) & (0.509) & (0.446) & (0.510) & (0.483) \\
 & $MA_{0.5}$ & 0.767 & 0.814 & 0.770 & 0.807 & 0.771 & 0.733 & 0.766 & 0.594 \\
&& (0.461) & (0.485) & (0.462) & (0.480) & (0.461) & (0.440) & (0.462) & (0.401) \\
\midrule
\multirow{ 12}{*}{$0.06$} & $0$ & 0.930 & 0.929 & 0.931 & 0.907 & 0.939 & 0.819 & 0.936 & 0.906 \\
&& (0.452) & (0.453) & (0.453) & (0.433) & (0.453) & (0.388) & (0.453) & (0.428) \\
 & $AR_{0.2}$ & 0.852 & 0.845 & 0.832 & 0.835 & 0.840 & 0.757 & 0.851 & 0.783 \\
&& (0.441) & (0.455) & (0.442) & (0.447) & (0.442) & (0.411) & (0.442) & (0.409) \\
 & $AR_{0.5}$ & 0.564 & 0.624 & 0.554 & 0.663 & 0.565 & 0.594 & 0.537 & 0.366 \\
&& (0.386) & (0.418) & (0.387) & (0.431) & (0.387) & (0.416) & (0.386) & (0.319) \\
 & $AR_{-0.5}$ & 0.998 & 0.995 & 0.998 & 0.987 & 0.998 & 0.937 & 0.998 & 0.995 \\
&& (0.397) & (0.376) & (0.398) & (0.343) & (0.397) & (0.299) & (0.398) & (0.330) \\
 & $MA_{0.2}$ & 0.862 & 0.870 & 0.857 & 0.856 & 0.865 & 0.771 & 0.862 & 0.796 \\
&& (0.442) & (0.455) & (0.443) & (0.446) & (0.443) & (0.409) & (0.443) & (0.410) \\
 & $MA_{0.5}$ & 0.768 & 0.801 & 0.761 & 0.811 & 0.773 & 0.732 & 0.762 & 0.611 \\
&& (0.404) & (0.430) & (0.405) & (0.432) & (0.404) & (0.405) & (0.405) & (0.341) \\
\bottomrule
\end{tabular}
\caption{Simultaneous coverage probabilities (average median interval length) over $G_{sub}$ for $k=4$ and $a=0.5$.}
\label{simultaneous1}
\end{table}

\begin{table}
\centering
\begin{tabular}{@{}lllllllllll@{}}
\toprule
 &  & \multicolumn{1}{c}{$\gamma=0$} & \multicolumn{2}{c}{$\gamma=0.2$} & \multicolumn{2}{c}{$\gamma=0.4$} & \multicolumn{2}{c}{$\gamma=0.6$} & \multicolumn{1}{c}{$-$} \\ \cmidrule(l){3-3} \cmidrule(l){4-5} \cmidrule(l){6-7} \cmidrule(l){8-9} \cmidrule(l){10-10}
$h$ & \multicolumn{1}{l}{DGP}  & WB & AWB & DWB & AWB  & DWB & AWB & DWB & SWB\\
\midrule
\multirow{ 12}{*}{$0.02$} & $0$ & 0.925 & 0.939 & 0.925 & 0.919 & 0.926 & 0.845 & 0.925 & 0.894 \\
&& (0.527) & (0.515) & (0.528) & (0.474) & (0.528) & (0.400) & (0.528) & (0.518) \\
 & $AR_{0.2}$ & 0.855 & 0.887 & 0.855 & 0.863 & 0.853 & 0.768 & 0.866 & 0.738 \\
&& (0.512) & (0.513) & (0.512) & (0.485) & (0.512) & (0.420) & (0.513) & (0.496) \\
 & $AR_{0.5}$ & 0.608 & 0.690 & 0.604 & 0.691 & 0.602 & 0.590 & 0.606 & 0.266 \\
&& (0.435) & (0.454) & (0.435) & (0.449) & (0.435) & (0.407) & (0.434) & (0.386) \\
 & $AR_{-0.5}$ & 0.982 & 0.979 & 0.979 & 0.972 & 0.986 & 0.915 & 0.984 & 0.984 \\
&& (0.448) & (0.413) & (0.449) & (0.359) & (0.449) & (0.285) & (0.449) & (0.398) \\
 & $MA_{0.2}$ & 0.874 & 0.896 & 0.869 & 0.876 & 0.865 & 0.792 & 0.870 & 0.770 \\
&& (0.513) & (0.513) & (0.513) & (0.484) & (0.513) & (0.419) & (0.513) & (0.497) \\
 & $MA_{0.5}$ & 0.780 & 0.833 & 0.779 & 0.812 & 0.764 & 0.728 & 0.776 & 0.498 \\
&& (0.459) & (0.474) & (0.460) & (0.459) & (0.460) & (0.408) & (0.460) & (0.411) \\
\midrule
\multirow{ 12}{*}{$0.04$} & $0$ & 0.914 & 0.915 & 0.915 & 0.882 & 0.917 & 0.762 & 0.917 & 0.875 \\
&& (0.401) & (0.397) & (0.402) & (0.373) & (0.402) & (0.326) & (0.402) & (0.388) \\
 & $AR_{0.2}$ & 0.812 & 0.822 & 0.793 & 0.798 & 0.799 & 0.688 & 0.806 & 0.697 \\
&& (0.391) & (0.398) & (0.391) & (0.385) & (0.391) & (0.345) & (0.392) & (0.372) \\
 & $AR_{0.5}$ & 0.461 & 0.544 & 0.439 & 0.571 & 0.435 & 0.480 & 0.441 & 0.206 \\
&& (0.338) & (0.362) & (0.339) & (0.368) & (0.338) & (0.347) & (0.338) & (0.289) \\
 & $AR_{-0.5}$ & 0.996 & 0.991 & 0.995 & 0.981 & 0.995 & 0.907 & 0.995 & 0.993 \\
&& (0.348) & (0.323) & (0.348) & (0.287) & (0.348) & (0.239) & (0.349) & (0.300) \\
 & $MA_{0.2}$ & 0.824 & 0.848 & 0.824 & 0.818 & 0.824 & 0.705 & 0.832 & 0.718 \\
&& (0.391) & (0.398) & (0.392) & (0.384) & (0.392) & (0.344) & (0.392) & (0.372) \\
 & $MA_{0.5}$ & 0.694 & 0.755 & 0.700 & 0.743 & 0.692 & 0.646 & 0.697 & 0.464 \\
&& (0.355) & (0.374) & (0.356) & (0.370) & (0.355) & (0.339) & (0.356) & (0.309) \\
\midrule
\multirow{ 12}{*}{$0.06$} & $0$ & 0.913 & 0.911 & 0.912 & 0.881 & 0.922 & 0.761 & 0.922 & 0.884 \\
&& (0.345) & (0.346) & (0.346) & (0.331) & (0.346) & (0.297) & (0.346) & (0.327) \\
 & $AR_{0.2}$ & 0.794 & 0.804 & 0.781 & 0.776 & 0.784 & 0.663 & 0.790 & 0.705 \\
&& (0.337) & (0.348) & (0.338) & (0.342) & (0.338) & (0.315) & (0.338) & (0.313) \\
 & $AR_{0.5}$ & 0.415 & 0.497 & 0.395 & 0.529 & 0.396 & 0.450 & 0.390 & 0.196 \\
&& (0.295) & (0.320) & (0.296) & (0.330) & (0.295) & (0.319) & (0.295) & (0.244) \\
 & $AR_{-0.5}$ & 0.999 & 0.995 & 0.997 & 0.988 & 0.998 & 0.932 & 0.998 & 0.996 \\
&& (0.304) & (0.287) & (0.304) & (0.262) & (0.304) & (0.228) & (0.304) & (0.252) \\
 & $MA_{0.2}$ & 0.815 & 0.835 & 0.817 & 0.807 & 0.827 & 0.682 & 0.826 & 0.740 \\
&& (0.338) & (0.348) & (0.339) & (0.341) & (0.338) & (0.313) & (0.338) & (0.313) \\
 & $MA_{0.5}$ & 0.680 & 0.740 & 0.686 & 0.729 & 0.683 & 0.632 & 0.683 & 0.480 \\
&& (0.308) & (0.329) & (0.310) & (0.331) & (0.309) & (0.310) & (0.309) & (0.260) \\
\bottomrule
\end{tabular}
\caption{Simultaneous coverage probabilities (average median interval length) over $G$ for $k=4$ and $a=0.5$.}
\label{simultaneous2}
\end{table}

Next, we consider the setting with missing data. Given the previous findings, we restrict ourselves to one bandwidth ($h=0.06$) but consider all AR and MA models. The results for pointwise as well as simultaneous coverage probabilities are given in Table \ref{missings}; further results, with similar conclusions, are available in Supplementary Appendix C. 

The first block presents pointwise coverage, while the second and third blocks show results for the sets $G_{sub}$ and $G$. The AWB performs well even if a significant proportion of the data are missing, as both pointwise and simultaneous coverage is close to the nominal level for almost all models. There is a significant increase in pointwise coverage for the cases with strong positive correlation, which is now close to 90\%. The same increase is visible for both $G_{sub}$ and $G$. This phenomenon does not come as a surprise, as the missing data points create space between consecutive observations, thus effectively reducing the serial dependence between observed points. Comparing the AWB with the DWB, we can see that the coverage is slightly closer to 95\% for the AWB in many cases. As before, the best performance is obtained at different values of $\gamma$ for the AWB and the DWB. The general pattern, however, is as in the previous setting. The DWB outperforms the AWB for higher values of $\gamma$, while the AWB obtains the most accurate coverage for smaller values of this parameter. A notable exception are again the cases with negative autocorrelation. Results for other bandwidths are similar, and presented in Supplementary Appendix C.

Overall, this simulation study indicates that the autoregressive wild bootstrap performs well in most of our considered scenarios. The pointwise confidence intervals show coverage close to the nominal level in the presence of heteroskedasticity and serial correlation. In addition, the method still performs well in the presence of missing data. The AWB provides simultaneous confidence bands with good coverage as long as the correlation does not become too strong. It outperforms the sieve wild bootstrap whenever we could compare results. In comparison to the dependent wild bootstrap, we saw that both methods provide very similar coverage probabilities, while the DWB produces slightly wider intervals. 

\begin{table}
\centering
\begin{tabular}{@{}llllllllll@{}}
\toprule
 &  & \multicolumn{1}{c}{$\gamma=0$} & \multicolumn{2}{c}{$\gamma=0.2$} & \multicolumn{2}{c}{$\gamma=0.4$} & \multicolumn{2}{c}{$\gamma=0.6$} \\ \cmidrule(l){3-3} \cmidrule(l){4-5} \cmidrule(l){6-7} \cmidrule(l){8-9}
 & \multicolumn{1}{l}{DGP}  & WB & AWB & DWB & AWB  & DWB & AWB & DWB\\
\midrule
\multirow{ 12}{*}{pw} & $0$ & 0.960 & 0.959 & 0.960 & 0.949 & 0.960 & 0.923 & 0.961 \\
&& (0.303) & (0.303) & (0.303) & (0.290) & (0.303) & (0.264) & (0.304) \\
 & $AR_{0.2}$ & 0.939 & 0.940 & 0.937 & 0.933 & 0.939 & 0.906 & 0.938 \\
&& (0.297) & (0.302) & (0.297) & (0.295) & (0.298) & (0.272) & (0.298) \\
 & $AR_{0.5}$ & 0.885 & 0.897 & 0.884 & 0.894 & 0.883 & 0.874 & 0.883 \\
&& (0.267) & (0.280) & (0.267) & (0.280) & (0.267) & (0.266) & (0.268) \\
 & $AR_{-0.5}$ & 0.989 & 0.986 & 0.989 & 0.978 & 0.989 & 0.957 & 0.989 \\
&& (0.271) & (0.263) & (0.272) & (0.246) & (0.272) & (0.221) & (0.272) \\
 & $MA_{0.2}$ & 0.942 & 0.945 & 0.942 & 0.936 & 0.943 & 0.911 & 0.943 \\
&& (0.298) & (0.302) & (0.298) & (0.294) & (0.297) & (0.271) & (0.298) \\
 & $MA_{0.5}$ & 0.924 & 0.931 & 0.925 & 0.926 & 0.924 & 0.903 & 0.924 \\
&& (0.276) & (0.285) & (0.275) & (0.282) & (0.275) & (0.263) & (0.276) \\
\midrule
\multirow{ 12}{*}{$G$} & $0$ & 0.937 & 0.936 & 0.941 & 0.906 & 0.937 & 0.844 & 0.943 \\
&& (0.237) & (0.237) & (0.237) & (0.227) & (0.237) & (0.207) & (0.238) \\
 & $AR_{0.2}$ & 0.898 & 0.899 & 0.903 & 0.863 & 0.885 & 0.797 & 0.888 \\
&& (0.233) & (0.236) & (0.233) & (0.231) & (0.233) & (0.213) & (0.233) \\
 & $AR_{0.5}$ & 0.769 & 0.797 & 0.773 & 0.782 & 0.767 & 0.716 & 0.766 \\
&& (0.209) & (0.219) & (0.209) & (0.219) & (0.209) & (0.208) & (0.209) \\
 & $AR_{-0.5}$ & 0.988 & 0.981 & 0.985 & 0.968 & 0.987 & 0.918 & 0.985 \\
&& (0.212) & (0.206) & (0.213) & (0.192) & (0.213) & (0.173) & (0.213) \\
 & $MA_{0.2}$ & 0.900 & 0.903 & 0.901 & 0.875 & 0.899 & 0.813 & 0.911 \\
&& (0.233) & (0.236) & (0.233) & (0.230) & (0.233) & (0.212) & (0.233) \\
 & $MA_{0.5}$ & 0.856 & 0.873 & 0.862 & 0.852 & 0.861 & 0.793 & 0.868 \\
&& (0.216) & (0.223) & (0.216) & (0.221) & (0.215) & (0.206) & (0.216) \\
\midrule
\multirow{ 12}{*}{$G_{sub}$} & $0$ & 0.945 & 0.949 & 0.953 & 0.923 & 0.942 & 0.890 & 0.951 \\
&& (0.273) & (0.273) & (0.273) & (0.262) & (0.273) & (0.238) & (0.274) \\
 & $AR_{0.2}$ & 0.923 & 0.919 & 0.917 & 0.899 & 0.911 & 0.858 & 0.913 \\
&& (0.268) & (0.273) & (0.268) & (0.266) & (0.268) & (0.245) & (0.269) \\
 & $AR_{0.5}$ & 0.828 & 0.855 & 0.843 & 0.844 & 0.835 & 0.809 & 0.834 \\
&& (0.241) & (0.252) & (0.241) & (0.253) & (0.241) & (0.240) & (0.241) \\
 & $AR_{-0.5}$ & 0.987 & 0.982 & 0.985 & 0.972 & 0.985 & 0.946 & 0.987 \\
&& (0.245) & (0.237) & (0.245) & (0.222) & (0.245) & (0.199) & (0.245) \\
 & $MA_{0.2}$ & 0.918 & 0.928 & 0.927 & 0.903 & 0.918 & 0.868 & 0.926 \\
&& (0.268) & (0.272) & (0.268) & (0.265) & (0.268) & (0.244) & (0.268) \\
 & $MA_{0.5}$ & 0.891 & 0.910 & 0.904 & 0.890 & 0.893 & 0.857 & 0.902 \\
&& (0.248) & (0.257) & (0.248) & (0.254) & (0.248) & (0.237) & (0.249) \\
\bottomrule
\end{tabular}
\caption{Coverage probabilities (average median interval length) with missing data ($h=0.06$).}
\label{missings}
\end{table}

\section{Trends in Atmospheric Ethane}
\label{sec:application}
We use our methodology to investigate the trending behavior of a time series of atmospheric ethane emissions which is derived from observations performed at the Jungfraujoch station in the Swiss Alps. This station can be found on the saddle between the Jungfrau and the M\"onch, located at 46.55$^{\circ}$ N, 7.98$^{\circ}$ E, 3580 m altitude. Ethane is the most abundant hydrocarbon gas in the atmosphere after methane and it is used as a measure of atmospheric pollution. It contributes to the formation of ground-level ozone and it influences the lifetime of methane which classifies it as an indirect greenhouse gas. This series, which has been studied by \citet{Franco}, is available from the Network for the Detection of Atmospheric Composition Change website at \href{ftp://ftp.cpc.ncep.noaa.gov/ndacc/station/jungfrau/hdf/ftir/}{ftp://ftp.cpc.ncep.noaa.gov/ndacc/station/jungfrau/hdf/ftir/}. It is argued in \citet{Franco} that the measurement conditions are very favorable at the Jungfraujoch location due to high dryness and low local pollution. Further details on the ground-based station and on how the measurements are obtained can be found in the aforementioned reference. It is a time series consisting of daily ethane columns (i.e. the number of molecules integrated between the ground and the top of the atmosphere) recorded under clear-sky conditions between September 1994 and August 2014 with a total of 2260 data points. Whenever more than one measurement is taken on one day, a daily mean is considered.

The average number of data points per year is 112.6 - giving an indication of the severity of the missing data problem present in this series. This shows that, in line with the above discussion, it is of major importance to use a bootstrap method which can replicate the missing data pattern correctly. The estimated transition probabilities of a first order Markov Chain reported in \eqref{eq:missings_data} already indicated the presence of (weak) serial dependence in the missing data generating mechanism. As a further exploration of this mechanism, note that the local constant estimator implicitly provides an estimator for the smoothly varying proportion of non-missing data $p(\tau)$. We can write \eqref{eq:estimator} as
\begin{equation*}
\begin{split}
\hat{m}(\tau) &= \hat{p}(\tau)^{-1} \frac{1}{nh} \sum_{t=1}^n K\left(\frac{t/n-\tau}{h}\right) D_t y_t, \qquad \text{where} \quad \hat{p} (\tau) = \frac{1}{nh} \sum_{t=1}^n K\left(\frac{t/n-\tau}{h}\right) D_t,
\end{split}
\end{equation*}
and $\hat{p}(\tau)$ can be seen as an estimator of $p(\tau)$.\footnote{Lemma \ref{lem:D_results} establishes the consistency of this estimator.} In Figure \ref{fig:Missingness}, we plot $\hat{p}(\cdot)$ as a diagnostic tool to investigate how data availability evolves over time. It fluctuates around the average proportion of 0.3, with a maximum of almost 0.4 and a minimum of 0.2. Although no overall trend appears to be present, the fluctuations are serious enough to cast doubt on the stationarity of the missing data generating mechanism; however, our method can handle this without problems.

\begin{figure}[htb]
\centering
\includegraphics[width=0.8\linewidth, clip, trim = {0 1.5cm 0 1cm}]{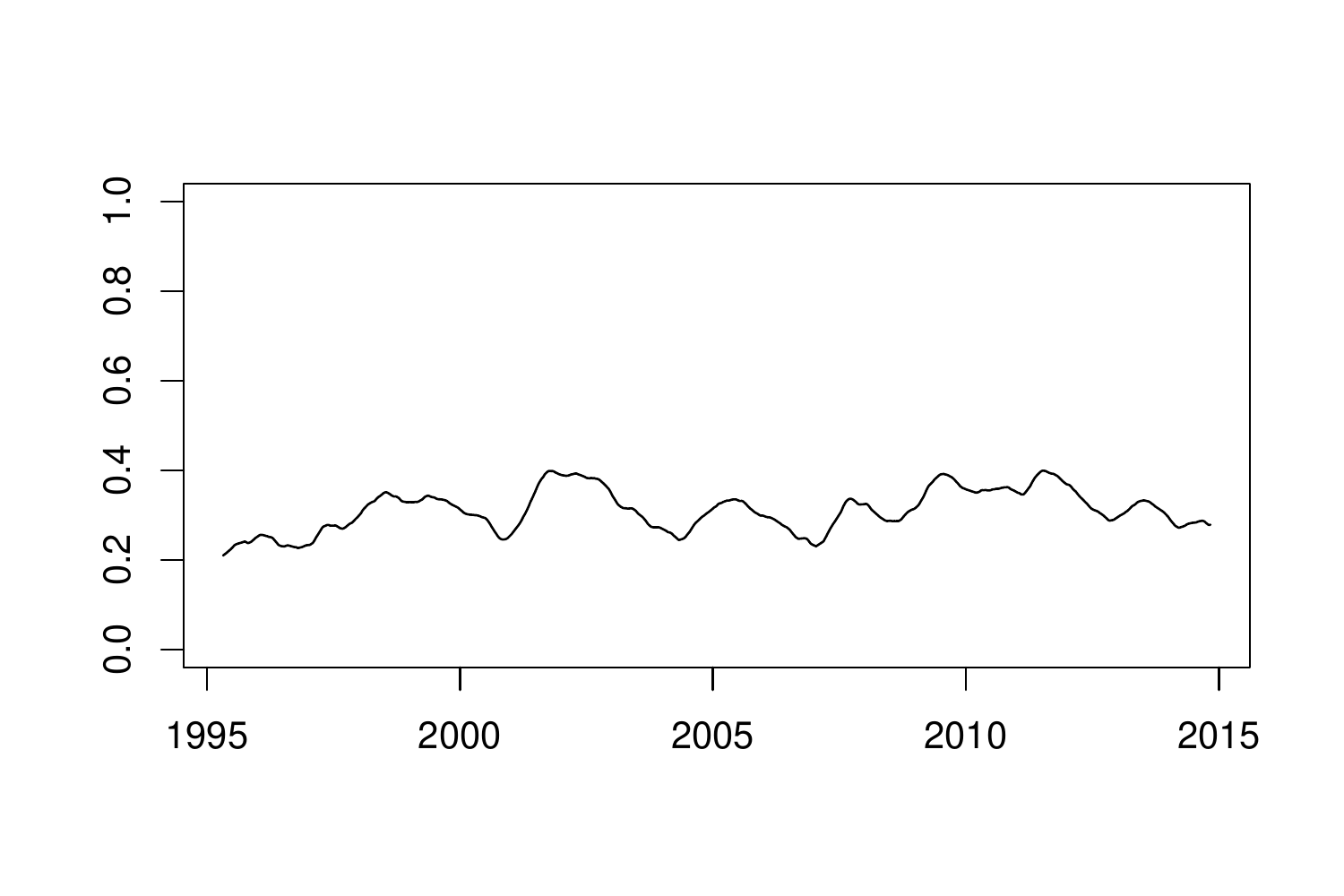}		
\caption{Time-varying proportion of non-missing data $\hat{p}$}\label{fig:Missingness}
\end{figure}

In addition, the data exhibit strong seasonality, as ethane degrades faster in summer than it does in winter, causing the series to displays peaks during winter and troughs during summer. \citet{Franco} take care of this seasonality with the help of Fourier terms by fitting the following model to ethane measurements $x_t$:
\begin{equation}
x_t=\sum_{j=1}^Ma_j\cos(2j\pi t)+b_j\sin(2j\pi t)+y_t.
\end{equation}
They continue their analysis with the residuals from this estimation, where $M=3$ is selected by inspecting the residual variance. To investigate the sensitivity of the choice of $M$, we perform a frequency domain analysis to give more insight about the form of the periodic pattern present in our data. Due to the missing data, we use the Lomb-Scargle periodogram, which is suitable in this situation (see \citet{Lomb76,Scargle82}). Figure \ref{fig:perio} plots the periodogram of the Jungfraujoch series with the frequency, rescaled to years, on the horizontal axis.

\begin{figure}[t]
	\centering
	\subfigure[Full periodogram]
    {
    \includegraphics[width=0.47\linewidth]{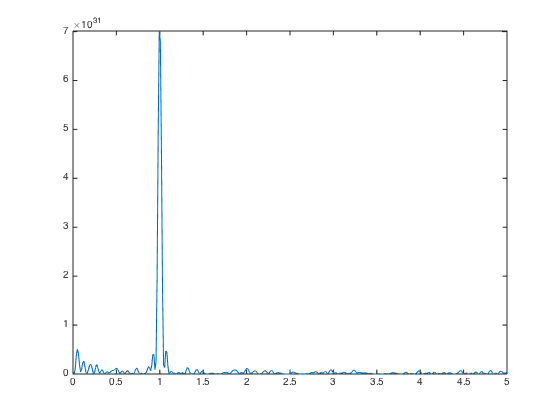}	
    \label{fig:full}
     }
     \subfigure[Zoomed version]
     {
      \includegraphics[width=0.47\linewidth]{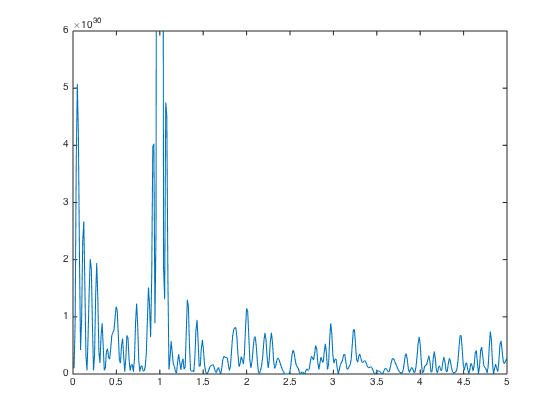}
      \label{fig:zoom}
     }
		\caption{The Lomb-Scargle periodogram of the ethane series}
		\label{fig:perio}
\end{figure}

The peaks around zero are the smooth long-run trend we model with our trend estimator. The present seasonality induces additional peaks at higher frequencies. There is a large peak at 1, representing an annual periodicity, which is so pronounced that it obscures peaks at other frequencies. Therefore, the right panel  displays the same spectrum as the left panel, but with a smaller vertical axis such that other peaks are observed more clearly. Moreover, we can observe that there are peaks at 2 and 3. They are, however, not as clear-cut as the peak at 1 and might not contribute as much to the seasonality, yet provide further justification for the choice $M=3$. In Supplementary Appendix D, we consider the periodograms of the residuals of the regression on 1 up to 4 Fourier terms. These show that while inclusion of one term is clearly needed, including more terms does indeed further reduce periodicity, although for increasing $M$, the effect becomes less pronounced. 

To corroborate our results and to be able to compare our findings to \citet{Franco}, we additionally look at the Akaike and Bayesian information criteria as well as the residual variance (MSE in \citet{Franco}) from the above regression for different values of $M$. The results are summarized in Table \ref{fourier}. While the Akaike criterion (AIC) is indifferent between adding 4 or 5 Fourier terms, BIC selects 2. The residual variance (MSE) decreases by only small increments when more than 3 terms are included. Based on our analysis, it is not clear how many Fourier terms we should include;  any value of $M$ between 1 and 4 seems reasonable. In the following, we report results for applying the trend estimation on the residuals of the regression with $M=3$ Fourier terms, as in \citet{Franco}. In Supplementary Appendix D we perform the same analysis with $M$ varying between 1 and 4, with hardly any difference in the results.

\begin{table}
\centering
\begin{tabular}{@{}clll@{}}
\toprule
\multicolumn{1}{l}{$M$} & \multicolumn{1}{c}{$AIC$} & \multicolumn{1}{c}{$BIC$} & \multicolumn{1}{c}{$MSE$}
                            \\ \cmidrule(l){1-1} \cmidrule(lr){2-2} \cmidrule(lr){3-3} \cmidrule(l){4-4}
1 & 158447.5 & 158464.6 & 2.79909$\times10^{30}$\\
2 & 158420.0 & 158448.6* & 2.76038$\times10^{30}$ \\
3 & 158413.4 & 158453.5 & 2.74750$\times10^{30}$ \\
4 & 158411.3* & 158462.8 & 2.74010$\times10^{30}$ \\
5 & 158411.3* & 158474.3 & 2.73521$\times10^{30}$ \\
6 & 158413.7 & 158488.1 & 2.73329$\times10^{30}$ \\
7 & 158416.8 & 158502.7 & 2.73219$\times10^{30}$ \\
\bottomrule
\end{tabular}
\caption{Fourier term investigation}
\label{fourier}
\end{table}

We next investigate bandwidth selection. As suggested in Section \ref{sec:estimation}, we determine a possible bandwidth using modified cross-validation. In line with the discussion in \citet{ChuMarron}, for our series, the ordinary leave-one-out cross-validation criterion selects a bandwidth which is too small ($h_{cv}=0.0006$). This value for the bandwidth parameter gives almost no smoothing of the data and the resulting trend curve is too wiggly. Leaving out $k=5$ observations on each side of any point, the modified criterion yields a value of $h_{mcv}=0.03$. Albeit a still small bandwidth, this value gives a much more reasonable picture of the trend estimate. The resulting nonparametric estimate as well as 95\% simultaneous confidence bands are depicted in Figure \ref{fig:Jungfraujoch}. The confidence bands are simultaneous over the whole sample. Although the validity has not been established, the algorithm works when we cover the whole sample and the results are easier to interpret. The bands are obtained using $B=999$ replications of the bootstrap procedure and an autoregressive parameter of $\gamma=0.5$.

As a robustness check, we also applied the trend estimation to the data without explicitly modeling the seasonality using Fourier terms. The nonparametric kernel estimator can be interpreted as a low pass filter which suppresses high frequency oscillations. A sufficiently large bandwidth should introduce enough smoothing to provide a trend curve which is not driven by the seasonal component. The bandwidth selected by MCV is too small for this effect to appear in our data. When we increase the bandwidth to $h=0.06$, the resulting trend shows the same pattern as the one in Figure \ref{fig:Jungfraujoch}. This analysis shows that when strong seasonality is present in the data, the nonparametric kernel estimator can be used to filter out the seasonality if the bandwidth is large enough. In this case, however, bandwidth selection becomes a critical issue and the proposed MCV criterion should be applied with care. Further details can be found in Supplementary Appendix D. 

\begin{figure}[!htb]
	\centering
		\includegraphics[width=\linewidth, clip, trim = {0 1.5cm 0 1cm}]{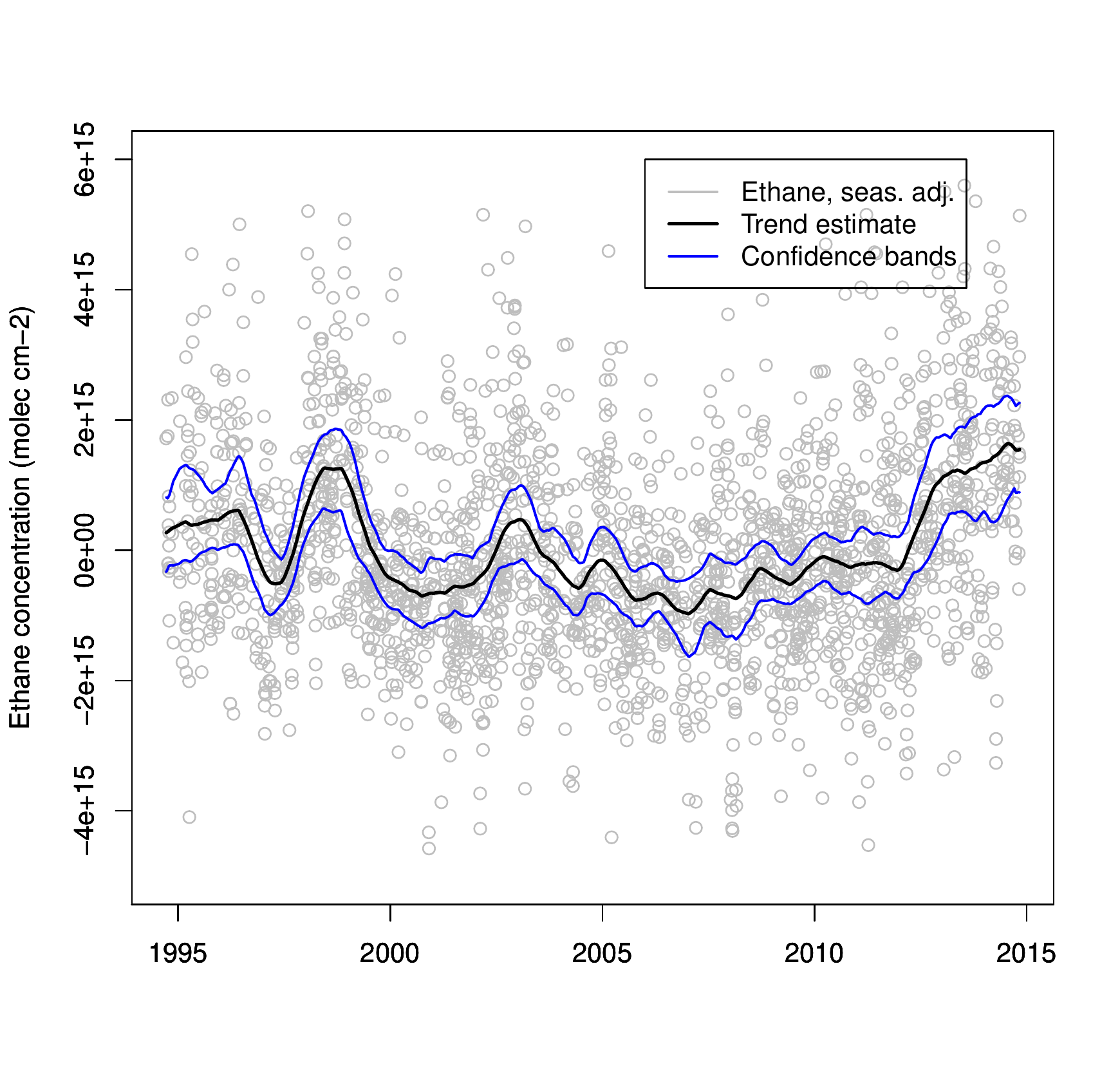}		
		\caption{Trend estimate with uniform 95\% confidence bands for ethane concentration at the Jungfraujoch (with bandwidth selected by MCV).}
		\label{fig:Jungfraujoch}
\end{figure}

We observe a slight downward trend of the ethane time series until around 2009, with local peaks in 1998 and 2002-2003, and an upward trend thereafter. This general development of the trend supports the findings in \citet{Franco} who estimate a linear trend model with a break at the beginning of 2009. They find a negative slope of the trend line before the break and a positive slope after the break. As mentioned by \citet{Franco}, the initial downward trend can be explained by a general emission reduction since the mid 1980's, of the fossil fuel sources in the Northern Hemisphere. This has also been reported by \citet{Simpson}. The upward trend seems to be a more recent phenomenon. Studies attribute it to the recent growth in the exploitation of shale gas and tight oil reservoirs, taking place in North America, see e.g. \citet{Vinci} and \citet{Franco2}. Since previous studies have mainly used methods based on linear trends, the two local peaks have to our knowledge not yet been analyzed. They can potentially be explained by boreal forest fires which were taking place mainly in Russia during both periods. Geophysical studies have investigated these events in association with anomalies in carbon monoxide emissions \citep{Yurganov1, Yurganov2}. In such fires, carbon monoxide is co-emitted with ethane, such that these events are likely explanations for the peaks we observe. 

As a final step, we look at the standard deviation of the residuals. When estimating it with a nonparametric kernel smoother, we see a cyclical pattern with upward trend, similar to the process we generate in our simulations. We plot the estimated standard deviation in Figure \ref{fig:Variance}. This clearly shows that the residuals are heteroskedastic which further underlines the importance of a flexible bootstrap method.

\begin{figure}[!htb]
	\centering
		\includegraphics[width=0.7\linewidth, clip, trim = {0 1.5cm 0 1cm}]{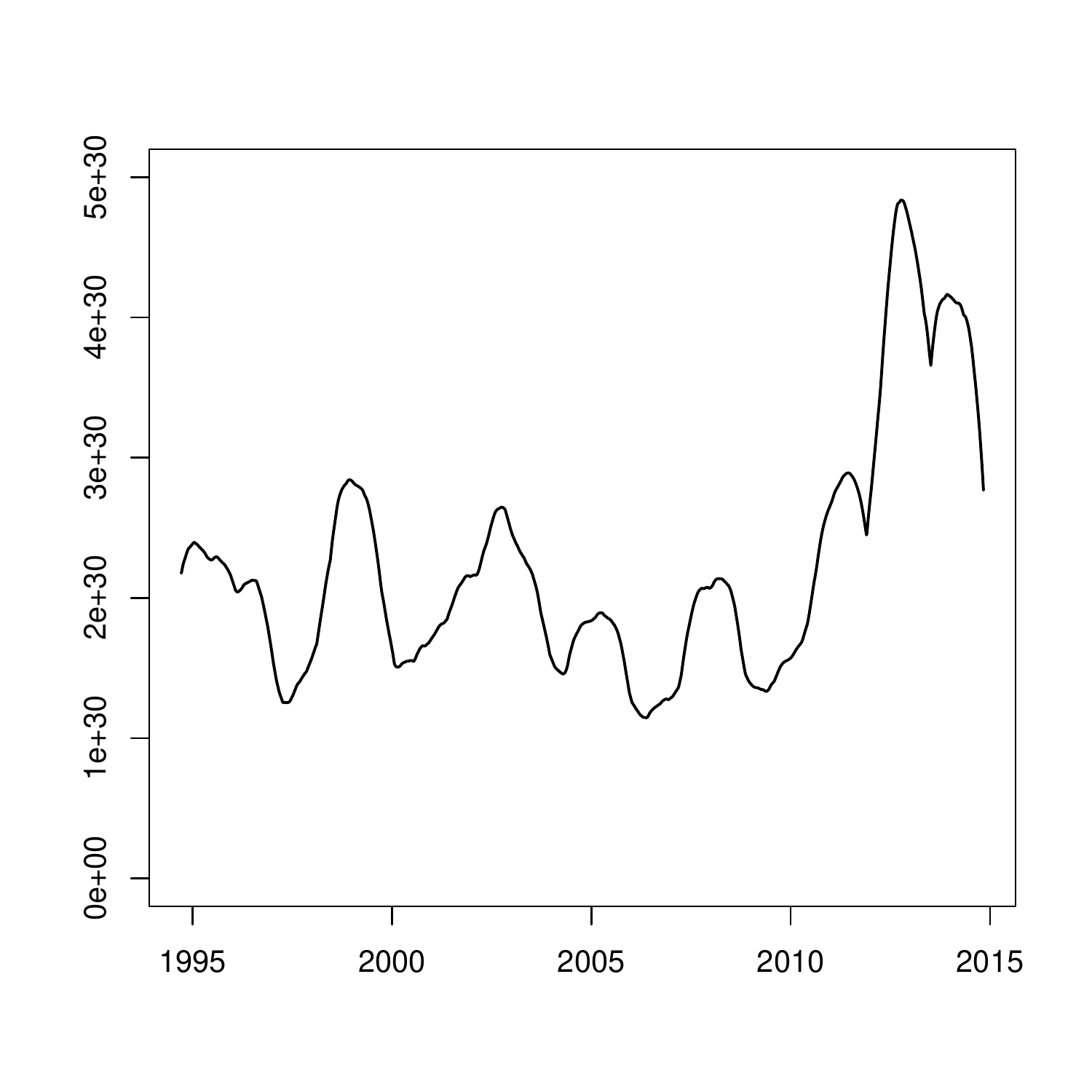}
		\caption{Estimate of the standard deviation of the residuals, obtained using the local constant kernel smoother with Epanechnikov kernel and $h=0.04$}
		\label{fig:Variance}
\end{figure}

\section{Conclusion}
\label{conclusion}
In this paper we have proposed a dependent version of the wild bootstrap, the autoregressive wild bootstrap, to construct confidence intervals around a nonparametrically estimated trend. Consistency of the bootstrap has been established such that it can be used to construct pointwise and simultaneous confidence bands. While the pointwise intervals always show good coverage in finite samples, simulation results for the simultaneous bands indicate that strong positive autocorrelation leads to a drop in coverage whenever simultaneous confidence bands are considered. However, other bootstrap methods such as the dependent wild bootstrap and the sieve bootstrap are equally affected, and overall the autoregressive wild bootstrap performs at least on par with these other methods, and often outperforms them.

One major advantage of the proposed approach is its broad applicability as it can be used under general forms of serial dependence and heteroskedasticity. Furthermore, it can be applied without further adjustments when data points are missing. This feature of the autoregressive wild bootstrap is particularly relevant in economic and climatological applications where the problem of missing data is often encountered. In addition to simulation results, we provide a rigorous asymptotic analysis where asymptotically pervasive missing data are allowed for. While the missing data generating mechanism affects the asymptotic distribution of the estimator, our bootstrap method correctly mimics this and is therefore valid in the presence of general forms of missing data patterns.

An application to atmospheric ethane measurements from Switzerland demonstrates our methodology. An upward trend in this time series is an indication of increasing atmospheric pollution and it has been visible in the data for the last quarter. This finding is in line with previous studies in geophysics and provides further evidence that an increased activity in shale gas extraction might have caused an increase in the ethane burden measured over the Jungfraujoch. In addition, we find two local peaks in the ethane series, which can be explained by boreal forest fires. Natural limitations of linear trend estimation have prevented these peaks from being discovered in previous research. This underlines the flexibility of our approach compared to parametric methods.

An open end to our analysis is the choice of the autoregressive parameter in the autoregressive wild bootstrap. Although our simulation results suggest that a range of values for this parameter perform adequately, its selection in practice remains an open issue. Theoretical results on the choice of this parameter are not trivial; moreover, such theoretical results do not translate directly into selection methods with good properties in small samples. This issue therefore merits deeper study and is left as an exercise for future research. 

\section*{Acknowledgements}
 We would like to thank guest editor Tommaso Proietti, two referees, Eric Beutner, David Hendry, Franz Palm and Hanno Reuvers as well as participants at the conference Econometric Models of Climate Change in Aarhus, the Maastricht Workshop on Advances in Quantitative Economics II and a seminar at Oxford University, for helpful discussions and valuable comments. We further thank Whitney Bader, Bruno Franco, Bernard Lejeune and Emmanuel Mahieu for helpful discussions regarding the ethane application. The second author thanks the Netherlands Organisation for Scientific Research (NWO) for financial support.

\numberwithin{equation}{section}
\numberwithin{lemma}{section}
\numberwithin{assumption}{section}
\setlength{\parskip}{0.7\baselineskip}
\setlength{\parindent}{0cm}

\begin{appendices}
\section{Technical Results}
For the remainder of this appendix, we define some further notation that will help lighten the notational load. Define $k_t(\tau)= K\left(\frac{t/n-\tau}{h}\right)$ and $\tilde{k}_t(\tau)= K\left(\frac{t/n-\tau}{\tilde{h}}\right)$. For any random variable $X$, let $\E_D X = \E (X | \{D_t\}_{t=1}^n)$ and $\var_D X = \E_D [X - (\E_D X)^2]^2$. Furthermore, let $\norm{X}_p = (\E \abs{X}^p)^{1/p}$ denote the $L_p$-norm of $X$, where we use $\norm{X}$ as short-hand notation for $\norm{X}_2$, and define he bootstrap equivalent as $\norm{X^*}_p^* = (\E^* \abs{X^*}^p)^{1/p}$. Finally, $C, C_1, C_2, \ldots$ denote arbitrary fixed positive constants not depending on $n$ or $h$, whose value can change each line, while $\phi_n$ denotes a generic deterministic sequence with the property that $\lim_{n \rightarrow \infty} \phi_n = 0$.

\subsection{Auxiliary lemmas}
We first establish some auxiliary lemmas with basic results that will be needed in the proofs of our main results. The proofs of these auxiliary lemmas are relegated to Supplementary Appendix B.

\begin{lemma} \label{lem:cov}
Let Assumptions \ref{as:LP}-\ref{as:MD} hold. Then for all $n \geq 1$,
\begin{align*}
(i)\; & \sum_{i=0}^{n-1} i \max_{1\leq t \leq n -i} \cov (D_t, D_{t+i}) \leq C_1;\\
(ii)\; & \sum_{i=0}^{n-1} i \max_{1\leq t \leq n -i} \cov (u_t, u_{t+i}) \leq C_2;\\
(iii)\; & \sum_{i=0}^{n-1} i \max_{1\leq t \leq n -i} \cov (D_t u_t, D_{t+i} u_{t+i}) \leq C_3.
\end{align*}
\end{lemma}

\begin{proof}[{\bf Proof of Lemma \ref{lem:cov}}]
See Supplementary Appendix B.
\end{proof}

\begin{lemma}\label{lem:kernel_sum}
Under Assumption \ref{as:kernel}, we have for any $\delta > 0$ that
\begin{align*}
\suptau \max_{0 \leq i \leq n} \frac{1}{nh} \sum_{t=1}^{n-i} k_t (\tau) k_{t+i} (\tau) < \infty.
\end{align*}
\end{lemma}

\begin{proof}[{\bf Proof of Lemma \ref{lem:kernel_sum}}]
See Supplementary Appendix B.
\end{proof}

\begin{lemma} \label{lem:kernel_function_limits}
Let $f_n(\cdot): [0,1] \rightarrow \mathbb{R}$ satisfy Assumption \ref{as:smooth}, and let $g (\cdot, \cdot): [0,1]^2 \rightarrow \mathbb{R}^2$ be a Lipschitz continuous function with Lipschitz constant $C_g = \sup_{(\tau_1, \tau_2) \in [0,1]^2} \abs{g_n(\tau_1, \tau_2)}$. Let Assumption \ref{as:kernel} be satisfied. Then, there exists some $N>0$, such that for any $0< \delta < 1/2$, all $n > N$ and all $i=0,1, \ldots, n$, we have that
\begin{equation*}
\begin{split}
(i) \; & \suptau \abs{\frac{1}{nh} \sum_{t=1}^{n} f\left(\frac{t}{n} \right) k_t(\tau) - f(\tau) \kappa_1 - \frac{1}{2} h^2 f^{(2)} (\tau) \mu_2} \leq C \max \left\{h^3, \frac{1}{nh} \right\};\\
(ii) \; & \suptaun \suptautt \abs{\frac{1}{nh} \sum_{t=1}^{n-i} g \left(\frac{t}{n}, \frac{t+i}{n} \right) k_t(\tau_0 + \tau_1 h) k_{t+i}(\tau_0 + \tau_2 h) - g(\tau_0, \tau_0) \kappa (\tau_1 - \tau_2)} \\
&\quad \leq C_g \left[\frac{C_1 i}{nh} + C_2 \max \left\{h, \frac{1}{nh^2} \right\} + S_{n}(i) \right],
\end{split}
\end{equation*}
where, for any sequence $\{\beta_i\}_{i=0}^{\infty}$ with $\sum_{i=1}^\infty \abs{\beta_i} < \infty$, we have that $\sum_{i=1}^{\infty} \abs{\beta_i} S_{n,i} \leq \phi_n$, where $\lim_{n \rightarrow \infty} \phi_n = 0$.
\end{lemma}

\begin{proof}[{\bf Proof of Lemma \ref{lem:kernel_function_limits}}]
See Supplementary Appendix B.
\end{proof}

\begin{lemma} \label{lem:D_results}
Let Assumptions \ref{as:sigma}-\ref{as:kernel} be satisfied. Define
\begin{equation} \label{eq:Z_D}
Z_{n,D,f,q} (\tau) = \frac{1}{\sqrt{nh}} \sum_{t=1}^{n} k_t(\tau) \left[f\left(\frac{t}{n} \right) - q_n (\tau)\right] \left[D_t - p\left(\frac{t}{n} \right) \right],
\end{equation}
where $f(\cdot)$ satisfies Assumption \ref{as:smooth} and $q_n (\cdot): [0,1] \rightarrow \mathbb{R}$ is a sequence of deterministic functions with $\sup_{\tau \in [0,1]}\abs{ q_n (\tau)} < \infty$ for all $n$. Then, there exists some $N>0$,  such that for any $0< \delta < 1/2$ and all $n > N$ the following hold:

$(i)$ $\E Z_{n,D,f,q} (\tau) = 0$ and
\begin{equation*}
\suptau \abs{\E Z_{n,D,f,l} (\tau)^2 - [f(\tau) - q_n(\tau)]^2 \Omega_D (\tau) \kappa_2} \leq C \max \left\{h, \frac{1}{nh^2}, \phi_n \right\},
\end{equation*}
where $\Omega_D (\tau) = R_{D,0}(\tau,\tau) + 2 \sum_{i=1}^\infty R_{D,i} (\tau, \tau)$ and $\lim_{n \rightarrow \infty} \phi_n = 0$.

$(ii)$ For all $\tau \in (0,1)$,
\begin{equation*}
\begin{split}
\frac{1}{nh} \sum_{t=1}^{n} f\left(\frac{t}{n} \right) k_t(\tau) D_t - p(\tau) f(\tau) - \frac{1}{2} h^2 [fp]^{(2)} (\tau) \mu_2 = \frac{Z_{n,D,f,0}(\tau)} {\sqrt{nh}} + R_{n,D} (\tau),\\
\end{split}
\end{equation*}
where $Z_{n,D,f,0}(\tau)$ is defined as in \eqref{eq:Z_D} with $q_n(\cdot) = 0$ and $\suptau \abs{R_{n,D} (\tau)} < C \max \left\{h^3, \frac{1}{nh} \right\}$.

$(iii)$ Let $\tilde g_i(\cdot, \cdot)$ be a Lipschitz continuous function with Lipschitz constant $C_g(i) = \sup_{(\tau_1, \tau_2) \in [0,1]^2} \abs{\tilde g_i(\tau_1, \tau_2)}$. Then for all $i = 0, 1, \ldots, n - 1$, 
\begin{equation*}
\begin{split}
&\suptau \suptautt \E \abs{\frac{1}{nh} \sum_{t=1}^{n} \tilde g_i \left(\frac{t}{n}, \frac{t+i}{n} \right) k_t(\tau_0 + \tau_1 h) k_{t+i}(\tau_0 + \tau_2 h) D_t D_{t+i} - p(\tau_0) \tilde g(\tau_0, \tau_0) \kappa (\tau_1 - \tau_2)} \\
&\qquad\leq C_g(i) \left[\frac{C_1}{\sqrt{nh}} + C_2 \max \left\{h, \frac{1}{nh^2} \right\} + S_n(i)\right],
\end{split}
\end{equation*}
where $S_n(i)$ is defined in Lemma \ref{lem:kernel_function_limits}.
\end{lemma}

\begin{proof}[{\bf Proof of Lemma \ref{lem:D_results}}]
See Supplementary Appendix B.
\end{proof}

\begin{lemma} \label{lem:lrv}
Let Assumptions \ref{as:smooth}-\ref{as:bandwidth} hold. Then, for all $0 \leq i \leq n-1$ and any $0< \delta < 1/2$,
\begin{equation*}
\begin{split}
&\suptau \E \abs{\frac{1}{nh} \sum_{t=1}^{n-i} k_t(\tau) k_{t+i}(\tau) \sigma_t \sigma_{t+i} D_t D_{t+i} \left[u_t u_{t+i}- \E u_t u_{t+i} \right]} \leq \beta_i \phi_{n} + \frac{\eta_{n}}{\sqrt{nh}},
\end{split}
\end{equation*}
where $\sum_{i=1}^\infty \beta_i < \infty$, $\lim_{n \rightarrow \infty} \phi_{n} = 0$ and $\limsup_{n \rightarrow \infty} \eta_{n} < \infty$.
\end{lemma} 

\begin{proof}[{\bf Proof of Lemma \ref{lem:lrv}}]
See Supplementary Appendix B.
\end{proof}

\subsection{Pointwise Results}
Lemmas \ref{lem:m_decomp} and \ref{lem:mb_decomp} decompose the (bootstrap) estimator in the relevant components, establishing not only consistency but also providing the necessary building blocks towards asymptotic normality. We therefore present these lemmas and their proofs together with the proofs of Theorems \ref{th:asdis} and \ref{th:basdis}.

\begin{lemma} \label{lem:m_decomp}
Let Assumptions \ref{as:smooth}-\ref{as:bandwidth}  be satisfied. Then, for some $N>0$ and any $0< \delta < 1/2$, we have for all $\tau \in (0,1)$ and $n > N$ that
\begin{equation*}
\sqrt{nh} \left[\hat{m}(\tau) - m(\tau) - h^2 B_{as}(\tau) \right] = p(\tau)^{-1} Z_{n,U} (\tau) + R_{n} (\tau),
\end{equation*}
where $B_{as}(\tau)$ is defined in \eqref{eq:as_pars},
\begin{subequations} 
\begin{align}
&Z_{n,U}(\tau) = \frac{1}{\sqrt{nh}} \sum_{t=1}^n k_t(\tau) D_t \sigma_t u_t, \label{eq:z_u}\\
&\sigma_{Z}^2 (\tau) = \plim_{n\rightarrow \infty} \var \left[ Z_{n,U} (\tau) \left.| \{D_t\}_{t=1}^n \right. \right] = p(\tau) \sigma(\tau)^2 \Omega_U \kappa_2, \label{eq:sigma_Z}\\
&\suptau \norm{R_n (\tau)} \leq C \max\left\{h^2, \sqrt{n h^7}, \frac{1}{\sqrt{nh}} \right\}.
\end{align}
\end{subequations}
\end{lemma}

\begin{proof}[{\bf Proof of Lemma \ref{lem:m_decomp}}]
Defining 
\begin{align*}
\hat{p} (\tau) &= \frac{1}{nh}\sum_{t=1}^n k_t(\tau) D_t,
& \breve{m} (\tau) &= \frac{1}{nh} \sum_{t=1}^n k_t(\tau) D_t y_t, \\
\overline{m}_{D} (\tau) &= \frac{1}{nh} \sum_{t=1}^n k_t(\tau) D_t m(t/n) & \overline{m}_p (\tau) &= \frac{1}{nh} \sum_{t=1}^n k_t(\tau) p(t/n) m(t/n),
\end{align*}
and realizing that $Z_{n,U} (\tau) = \sqrt{nh} \left[\breve{m} (\tau) - \overline{m}_{D} (\tau) \right]$, we can write
\begin{equation} \label{eq:m_decomp}
\begin{split}
&\sqrt{nh} \left[\hat{m}(\tau) - m(\tau) - h^2 B_{as} (\tau) \right] = p(\tau)^{-1} Z_{n,U}(\tau) + [\hat{p} (\tau)^{-1} - p(\tau)^{-1}] Z_{n,U} (\tau) / \sqrt{nh} \\
&\quad +\sqrt{nh}\left\{\hat{p} (\tau)^{-1} \overline{m}_{D} (\tau) - p(\tau)^{-1} \overline{m}_p (\tau)\right\} + \sqrt{nh}\left\{p(\tau)^{-1}\overline{m}_p (\tau) - m(\tau) - h^2 B_{as} (\tau) \right\}\\
&= p(\tau)^{-1} Z_{n,U}(\tau) + I_{n}(\tau) + II_{n}(\tau) + III_n(\tau).
\end{split}
\end{equation}

We first derive $\sigma_Z^2 (\tau)$. Let $\E_D (\cdot) = \E (\cdot | \{D_t\}_{t=1}^n)$. Then $\E_D Z_{n,U}(\tau) = 0$ and
\begin{align*}
\E_D Z_{n,U}(\tau)^2 &= \frac{1}{nh} \sum_{s=1}^{n} \sum_{t=1}^{n} k_s(\tau) k_t(\tau) D_s D_t \sigma_s \sigma_t \E u_s u_t \\
&= \frac{1}{nh} \sum_{i=-n+1}^{n-1} R_U(i) \sum_{t=1}^{n-\abs{i}} k_t(\tau) k_{t+\abs{i}}(\tau) \sigma_t \sigma_{t+\abs{i}} D_t D_{t+\abs{i}}.
\end{align*}
Using Lemma \ref{lem:D_results}(iii) with $\tilde g_i (\tau_1, \tau_2) = \sigma(\tau_1) \sigma(\tau_2)$, we find that
\begin{align*}
&\suptau \E \abs{\sum_{i=-n+1}^{n-1} R_U(i) \frac{1}{nh} \sum_{t=1}^{n-\abs{i}} k_t(\tau) k_{t+\abs{i}}(\tau) \sigma_t \sigma_{t+\abs{i}} D_t D_{t+\abs{i}} - p(\tau) \sigma(\tau)^2 \kappa_2 \Omega_U}\\
&\leq 2 \sum_{i=0}^{n-1} \abs{R_U(i)} \left[\frac{C_1}{\sqrt{nh}}  + C_2 \max \left\{h, \frac{1}{n h^2} \right\} + S_n(i) \right] + 2 \sum_{i=n}^\infty \abs{R_U(i)},\\
& \leq \frac{C_3}{\sqrt{nh}} + C_4 \max\left\{h, \frac{1}{n h^2} \right\} + C_5 \phi_n.
\end{align*}
where $\lim_{n \rightarrow \infty} \phi_n = 0$, from which it follows that, for $\tau \in (0,1)$, $\E\abs{\E_D Z_{n,U} (\tau)^2 - \sigma_{Z}^2 (\tau)} \leq C \max\{(nh)^{-1/2}, h, n^{-1} h^2, \phi_n \} = o(1)$, where $\sigma_{Z}^2 (\tau)$ is defined in \eqref{eq:sigma_Z}.

Furthermore, by the arguments used above it follows that $\suptau \norm{Z_{n,U}(\tau)} \leq C$, while by Lemma \ref{lem:D_results}$(i)$
\begin{equation*}
\norm{\hat{p}(\tau) - p(\tau)} \leq \frac{1}{\sqrt{nh}} \norm{Z_{n,D} (\tau)} + C_1 h^2 + C_2 \max\left\{h^3, \frac{1}{nh} \right\} \leq C_3 \max\left\{\frac{1}{\sqrt{nh}}, h^2 \right\}.
\end{equation*}
By combining these results with the fact that $\norm{\hat{p} (\tau)}^{-1} \leq 1/\epsilon^*$ by Assumption \ref{as:bandwidth}, we find that
\begin{align*}
\suptau \norm{I_{n}(\tau)} &\leq \suptau\norm{\hat{p}(\tau) - p(\tau)} \norm{\hat{p}(\tau)^{-2}} \abs{ p(\tau)^{-2}} \norm{Z_{n,U} (\tau)}\leq C_4 \max\left\{\frac{1}{\sqrt{nh}}, h^2 \right\}.
\end{align*}

For $II_n(\tau)$ we note that
\begin{equation} \label{eq:p_inv}
\hat{p}(\tau)^{-1} - p(\tau)^{-1} = p(\tau)^{-2} \left\{p(\tau) - \hat{p}(\tau)\right\} + \hat{p}(\tau)^{-1} p(\tau)^{-2} \left\{p(\tau) - \hat{p}(\tau) \right\}^2,
\end{equation}
such that we can rewrite $II_{n}(\tau)$ as
\begin{align*}
II_n(\tau) / \sqrt{nh} &= p (\tau)^{-1} \left\{\overline{m}_{D} (\tau) - \overline{m}_p (\tau)\right\} + \left\{\hat{p} (\tau)^{-1} - p(\tau)^{-1}\right\} \overline{m}_p (\tau)\\
&\quad + \left\{\hat{p} (\tau)^{-1} - p(\tau)^{-1}\right\} \left\{ \overline{m}_{D} (\tau) - \overline{m}_p (\tau) \right\}\\
&= p (\tau)^{-1} \left\{\overline{m}_{D} (\tau) - \overline{m}_p (\tau)\right\} - p(\tau)^{-2} \left\{\hat{p}(\tau) - p(\tau)\right\} \overline{m}_p (\tau) \\
&\quad + \hat{p}(\tau)^{-1} p(\tau)^{-2} \left\{\hat{p}(\tau) - p(\tau) \right\}^2 \overline{m}_p (\tau) + \left\{\hat{p} (\tau)^{-1} - p(\tau)^{-1}\right\} \left\{ \overline{m}_{D} (\tau) - \overline{m}_p (\tau) \right\}\\
& = \left[II_{n,11}(\tau) - II_{n,12}(\tau) + II_{n,2}(\tau) + II_{n,3}(\tau) \right]/\sqrt{nh}.
\end{align*}
As $II_{n,11} (\tau) = p(\tau)^{-1} \frac{1}{\sqrt{nh}} \sum_{t=1}^n k_t(\tau) m(t/n) \left[ D_t  - p(t/n) \right]$ and $II_{n,12} (\tau) = p(\tau)^{-2} \overline{m}_p(\tau) \frac{1}{\sqrt{nh}} \sum_{t=1}^n k_t(\tau) \allowbreak \times \left[ D_t  - p(t/n) \right]$, we can write
\begin{align*}
II_{n,1}(\tau) &= II_{n,11}(\tau) - II_{n,12}(\tau) = p(\tau)^{-1} \frac{1}{\sqrt{nh}} \sum_{t=1}^n k_t(\tau) \left[m(t/n) - p(\tau)^{-1} \overline{m}_p(\tau) \right] \left[ D_t  - p(t/n) \right]\\
&= p(\tau)^{-1} Z_{n,D} (\tau),
\end{align*}
where $Z_{n,D} (\tau)$ is defined as $Z_{n,D,f,q} (\tau)$ in Lemma \ref{lem:D_results}$(i)$ with $f(\cdot) = m(\cdot)$ and $q_n (\cdot) = p(\cdot)^{-1} \overline{m}_p (\cdot)$. Applying this lemma we find that
\begin{align*}
&\suptau \abs{\E Z_{n,D} (\tau)^2 - \left[m(\tau) - p(\tau)^{-1} \overline{m}_p(\tau) \right]^2 \Omega_{D} (\tau) \kappa_2} \leq C \max\{h, n^{-1} h^2\} + \phi_n.
\end{align*}
As $\abs{x^2 - y^2} \leq \abs{x - y} (\abs{x - y} + 2 \abs{y})$, it follows from Lemma \ref{lem:kernel_function_limits}$(i)$ that
\begin{align*}
\abs{\left[m(\tau) - p(\tau)^{-1} \overline{m}_p(\tau) \right]^2 - h^4 B_{as} (\tau)^2} &\leq C \max\{h^3, (nh)^{-1}\} \left[C \max\{h^3, (nh)^{-1}\} + 2 h^2 B_{as} (\tau) \right] \\
& \leq C_1 \max\{h^5, n^{-1} h, (nh)^{-2}\},
\end{align*}
and therefore, with $\sigma_D^2(\tau) = h^4 B_{as} (\tau)^2 \Omega_D (\tau) \kappa_2$, 
\begin{align*}
&\suptau \abs{\E Z_{n,D} (\tau)^2 - \sigma_D^2(\tau)} \leq C_1 \max\{h^5, n^{-1} h, (nh)^{-2}\} + C_2 \phi_n.
\end{align*}
It then follows directly that 
\begin{align*}
\suptau \norm{ II_{n,1} (\tau)} &\leq \frac{1}{\epsilon^*} \left( \suptau \abs{\E Z_{n,D} (\tau)^2 - \sigma_D^2(\tau)}^{1/2} + \suptau \abs{ \sigma_D(\tau)} \right) \leq C \max\left\{h^{2}, \sqrt{\frac{h}{n}}, \frac{1}{nh} \right\}.
\end{align*}

As $\suptau \norm{\hat{p}(\tau) - p(\tau)} \leq C\max\left\{h^2, \frac{1}{\sqrt{nh}} \right\}$, it follows that $\suptau \norm{II_{n,2}(\tau)} \leq C \max \left\{\sqrt{n h^{9}}, \frac{1}{\sqrt{nh}} \right\}$. Furthermore, by Lemma \ref{lem:D_results}$(i)$, $\norm{\overline{m}_{D} (\tau) - \overline{m}_p(\tau)} \leq \frac{1}{\sqrt{nh}} \norm{Z_{n,D,m,0}(\tau)} \leq \frac{C}{\sqrt{nh}}$, such that it also follows that
\begin{equation*}
\suptau \norm{II_{n,3} (\tau)} \leq C \suptau \norm{\hat{p}(\tau) - p(\tau)} \norm{\overline{m}_{|D} (\tau) - \overline{m}_p(\tau)} \leq C \max\left\{h^2, \frac{1}{\sqrt{nh}} \right\}.
\end{equation*}
Finally, it follows directly from Lemma \ref{lem:kernel_function_limits}$(i)$ that $\suptau \abs{ III_n (\tau)} \leq C \max\left\{\sqrt{n h^7}, \frac{1}{\sqrt{nh}} \right\}$. Collecting all remainder terms in $R_n (\tau) = I_{n,1}(\tau) + II_{n,2} (\tau) + II_{n,3} (\tau) + III_{n} (\tau)$, we find that $\suptau \norm{R_n (\tau)} \leq C \max\left\{h^2, \sqrt{n h^7}, \frac{1}{\sqrt{nh}} \right\}$.
\end{proof}

\begin{proof}[{\bf Proof of Theorem \ref{th:asdis}}]
Given Lemma \ref{lem:m_decomp}, we only have to prove asymptotic normality of $Z_{n,U} (\tau) = \frac{1}{\sqrt{nh}}\sum_{t=1}^n k_t(\tau) D_t \sigma_t u_t$. To simplify the proofs, we first condition on $\{D_t\}_{t=1}^n$ and thus prove conditional asymptotic normality of $Z_{n,U} (\tau) = \frac{1}{\sqrt{nh}}\sum_{t=1}^n k_t(\tau) D_t \sigma_t u_t$. As before, let $\E_D (\cdot) = \E (\cdot|\{D_t\}_{t=1}^{\infty})$. As the limit results do not depend on $\{D_t\}_{t=1}^n$, the results then directly hold unconditionally as well. 

Take an $M$ such that $M\rightarrow\infty$ as $n\rightarrow\infty$, and truncate the MA($\infty$) representation of $u_t$ at $M$ lags, say $u_{t,M} = \sum_{j=1}^{M} \psi_j\epsilon_{t-j}$. Then we can write
\begin{align*}
Z_n (\tau) &= \frac{1}{\sqrt{nh}} \sum_{t=1}^{n} k_t(\tau) D_t \sigma_t u_{t,M} + \frac{1}{\sqrt{nh}} \sum_{t=1}^{n} k_t(\tau) D_t \sigma_t \sum_{j=M+1}^{\infty} \psi_j\epsilon_{t-j}\\
&= \overline{Z}_{n,M} (\tau) + \overline{W}_{n,M} (\tau).
\end{align*}
Let $R_W(k) = \E \left(\sum_{j=M+1}^{\infty} \psi_j\epsilon_{t-j} \right) \left(\sum_{j=M+1}^{\infty} \psi_j\epsilon_{t+k-j} \right) = \sigma_\varepsilon^2 \sum_{j=M+1}^{\infty}\psi_j\psi_{j+|k|}$. Then
\begin{align*}
\E_D \overline{W}_{n,M}^2 &= (nh)^{-1}\sum_{t=1}^n\sum_{s=1}^n D_t D_s \sigma_t\sigma_sk_t(\tau)k_s(\tau)R_{W}(t-s)\\
&\leq 2 (nh)^{-1} \sup_{\tau \in [0,1]} \sigma(\tau)^2 \sum_{i=0}^n \sum_{t=1}^{n-i} k_t (\tau) k_{t+i} (\tau) \abs{R_{W}(i)}\\
&\leq 2 C \sup_{\tau \in [0,1]} \sigma(\tau)^2 \sum_{i=0}^\infty \abs{\sum_{j=M+1}^{\infty}\psi_j\psi_{j+i}}
\leq 2 C \sup_{\tau \in [0,1]} \sigma(\tau)^2 \left(\sum_{j=M+1}^{\infty} \abs{\psi_j}  \right)^2 = o(M^{-2}),
\end{align*}
which follows as $0 \leq D_t D_{t+i} \leq 1$, by Lemma \ref{lem:kernel_sum} and the fact that the summability condition in Assumption \ref{as:LP} implies that $\sum_{j=M+1}^{\infty} \abs{\psi_j} = o(M^{-1})$. As $\E_D \overline{W}_{n,M} = 0$, the Markov inequality implies that the truncation is asymptotically negligible.

We next split $\overline{Z}_{n,M}(\tau)$ into two sequences of blocks: one with small, negligible blocks $Y_{n,j}(\tau)$ and one with dominating blocks $X_{n,j}(\tau)$. Define $B_j = (j-1)(a+b)$, then
\begin{align*}
&X_{n,j} (\tau) = \frac{1}{\sqrt{nh}} \sum_{t = B_j + 1}^{B_j + a} k_t(\tau) D_t \sigma_t u_{t,M}, 
&Y_{n,j}(\tau) = \frac{1}{\sqrt{nh}} \sum_{t = B_j + a + 1}^{B_j + a + b} k_t(\tau) D_t \sigma_t u_{t,M},
\end{align*}
such that $\overline{Z}_{n,M}(\tau) = \sum_{j=1}^k X_{n,i}(\tau)+\sum_{i=1}^k Y_{n,i}(\tau)$, where $k = \lceil n / (a+b) \rceil$, and the final block is truncated to have $n$ observations in total. Now take sequences $a = a(n)$ and $b = b(n) \rightarrow \infty$ such that $a/(n h) + M/a \rightarrow 0$ and $b/a + M/b \rightarrow 0$ as $n \rightarrow \infty$.

We first show that the small blocks are asymptotically negligible. First note that $\E_D \left(\sum_{i=1}^k Y_{n,i}(\tau)\right)=0$. Consider $n$ large enough such that $a>M$ and the blocks $Y_{n,i}$ are mutually independent conditionally on $\{D_t\}_{t=1}^n$. Then, with $R_M(i) = \E u_{t,M} u_{t+i,M}$, where $\sum_{i=0}^\infty \abs{R_M(i)}<\infty$ by Assumption \ref{as:LP}, we have that
\begin{align*}
&\E_D \left(\sum_{j=1}^k Y_{n,j}(\tau)\right)^2 = \sum_{j=1}^k \E_D Y_{n,j}(\tau)^2 = \frac{1}{nh} \sum_{j=1}^k \sum_{s, t = B_j + a + 1}^{B_j + a + b} k_s(\tau) k_t(\tau) D_s D_t \sigma_s \sigma_t R_M(s-t) \\
& \leq 2 C \frac{1}{nh} \sum_{i=0}^{b-1} \abs{R_M(i)} \sum_{j=1}^k \sum_{t=B_j + a + 1}^{B_j + a + b - i} k_{t}(\tau) k_{t+i}(\tau) \leq C_1 \frac{1}{nh} \max_{0\leq i \leq b-1} \sum_{j=1}^k \sum_{t=B_{i} + a + 1}^{B_{i} + a + b - i} k_{t}(\tau) k_{t+i}(\tau) \\
&\leq C_2 \frac{k b h}{nh} \leq C_3 \frac{b}{a} = o(1),
\end{align*}
where we use that $\sum_{j=1}^k \sum_{t=B_{i} + a + 1}^{B_{i} + a + b - i} k_{t}(\tau) k_{t+i}(\tau) \leq C k b h $ and $k\sim an$.

We next employ the Lindeberg central limit theorem \citep[see e.g.][Thm 23.6]{Davidson} to show that $\sum_{j=1}^k X_{n,j}(\tau) \xrightarrow{d}\mathcal{N}(0, p(\tau)^2 \sigma^2_{as}(\tau))$. Consider again $n$ sufficiently large such that such that $b>M$ and the blocks $X_{n,i}(\tau)$ are conditionally independent. Then $\E_D \sum_{j=1}^k X_{n,j}(\tau) = 0$ and
\begin{align*}
\E_D\left(\sum_{j=1}^k X_{n,j}(\tau) \right)^2 &= \sum_{j=1}^k \E_D X_{n,j}(\tau)^2 = \frac{1}{nh} \sum_{i=1}^k \sum_{s, t=B_j + 1}^{B_j + a} k_s(\tau) k_t(\tau) D_s D_t \sigma_s \sigma_t R_M(s-t)\\
&= \frac{1}{nh} \sum_{j=1}^k \sum_{i=-a+1}^{a-1} R_M(i) \sum_{t=B_j + 1}^{B_j + a - \abs{i}} k_t(\tau) k_{t+\abs{i}} (\tau) \sigma_t \sigma_{t+\abs{i}} D_t D_{t+\abs{i}}.
\end{align*}
As $M \rightarrow \infty$, $R_M(k) \rightarrow R_U(k)$. A straightforward extension of Lemma \ref{lem:D_results} then shows that $\E_D\left(\sum_{j=1}^k X_{n,j}(\tau) \right)^2 \xrightarrow{p} p(\tau) \sigma(\tau)^2 \Omega \kappa_2$.

The final step is to verify the Lindeberg condition, that is, we verify that, for every $\kappa>0$, $\sum_{j=1}^k \E_D \left[\frac{X_{n,j}(\tau)^2} {\omega_n^2} \mathbbm{1} \left( \abs{\frac{X_{n,i}(\tau)} {\omega_n}} >\kappa \right) \right] = o_p(1)$, with $\omega_n^2 = \E_D \left(\sum_{j=1}^k X_{n,j}(\tau) \right)^2$. Note that
\begin{align*}
&\sum_{j=1}^k \E_D \left[\frac{X_{n,j}(\tau)^2}{\omega_n^2} \mathbbm{1} \left( \abs{\frac{X_{n,i}(\tau)} {\omega_n}} >\kappa \right) \right] \leq \frac{1}{\kappa^2 \omega_n^{4}} \sum_{j=1}^k \E_D X_{n,j}(\tau)^{4}
\end{align*}
and define $\tilde{X}_t (\tau) = D_t k_t(\tau) \sigma_t u_{t,M}$, which, due to the $M$-dependence of $u_{t,M}$, is an $L_{4}$-mixingale with $\E_D \tilde X_t^{4} \leq C k_t(\tau)^{4} \E u_{t,M}^{4}$, where it follows from Minkowski's inequality that
\begin{align} \label{eq:u4}
\E u_{t,M}^{4} &\leq \left(\sum_{j=0}^M \left(\E\left|\psi_j\varepsilon_{t-j}\right|^{4}\right)^{1/4}\right)^{4} \leq C \E \varepsilon_t^{4} \left(\sum_{j=0}^M \psi_j \right)^{4} < \infty,
\end{align}
by stationarity of $\{\varepsilon_t\}$ and $\E \varepsilon_t^{4} < \infty$. Lemma 2 of \citet{Hansen91} then implies that
\begin{equation*}
\E_D X_{n,j}(\tau)^{4} = \E_D \abs{\frac{1}{\sqrt{nh}} \sum_{t=B_j + 1}^{B_j + a} \tilde X_t (\tau)}^{4} \leq C \frac{1}{(nh)^2} \left(\sum_{t=B_j + 1}^{B_j + a} k_t(\tau)^{2} \right)^{2}.
\end{equation*}
As the blocks are non-overlapping, it follows from the properties of the kernel function and the $c_r$-inequality that
\begin{equation*}
\sum_{j=1}^k \left(\sum_{t=B_j + 1}^{B_j + a} k_t(\tau)^{2} \right)^{2} \leq a \sum_{j=1}^k \sum_{t=B_j + 1}^{B_j + a} k_t(\tau)^{4} \leq C nh a.
\end{equation*}
Therefore, with $\omega_n^{-2} = O(1)$ and $a = o(nh)$, $\sum_{i=1}^k \E_D \left[\frac{X_{n,i}(\tau)^2}{\omega_n^2} \mathbbm{1} \left( \abs{\frac{X_{n,i}(\tau)} {\omega_n}} >\kappa \right)\right] \leq \frac{Ca}{\omega_n^{4}nh} = o(1)$.
\end{proof}

\begin{lemma} \label{lem:mb_decomp}
Let Assumptions \ref{as:smooth}-\ref{as:bandwidth2} hold. Then, for any $0 < \delta < \delta^* < \frac{1}{2}$ and some $N>0$, it holds for all $\tau \in [\delta^*,1-\delta^*]$ and $n>N$ that
\begin{equation*}
\sqrt{nh}\left[\hat{m}^*(\tau) - \tilde{m}(\tau) - h^2 B_{as}(\tau) \right] = p(\tau)^{-1} Z_{n,U}^* (\tau) + R_n^*(\tau),
\end{equation*}
where
\begin{subequations} 
\begin{align}
&Z_{n,U}^*(\tau) = \frac{1}{\sqrt{nh}}\dsum{t} k_t(\tau) D_t z_t \xi_t^* \label{eq:zstar}\\
&\plim_{n \rightarrow \infty} \var^* Z_{n,U}^*(\tau) = \sigma_{Z}^2(\tau) \label{eq:bvar}\\
&\suptaus \E \norm{R_n^* (\tau)}^* \leq C \max\left\{\frac{1}{\sqrt{nh}}, \sqrt{n h^7}, \sqrt{h \tilde{h}^{-1}}, \sqrt{n h^{5} \tilde{h}^4}, \tilde{h}^{2}, \sqrt{\ell\tilde{h}^4}, \sqrt{\frac{\ell}{n\tilde{h}}} \right\},
\end{align}
\end{subequations}
and $B_{as} (\tau)$ and $\sigma_Z^2(\tau)$ are defined in \eqref{eq:as_pars} and \eqref{eq:sigma_Z} respectively.
\end{lemma}

\begin{proof}[{\bf Proof of Lemma \ref{lem:mb_decomp}}]
Analogously to the proof of Lemma \ref{lem:m_decomp}, we define 
\begin{align*}
\breve{m}^* (\tau) &= \frac{1}{nh} \dsum{t} k_t(\tau) D_t y_t^*, & \overline{m}_{D}^* (\tau) &= \frac{1}{nh} \dsum{t} k_t(\tau) D_t \tilde{m}\left(\frac{t}{n}\right)\\
\overline{m}_p^* (\tau) &= \frac{1}{nh} \dsum{t} p\left(\frac{t}{n}\right) k_t(\tau) \tilde{m}\left(\frac{t}{n}\right),
\end{align*}
such that we can write
\begin{equation} \label{eq:mb_decomp}
\begin{split}
&\sqrt{nh} \left[\hat{m}^*(\tau) - \tilde{m}(\tau) - h^2 B_{as} (\tau) \right] = \sqrt{nh} \hat{p} (\tau)^{-1} \left\{\breve{m}^* (\tau) - \overline{m}_{D}^* (\tau)\right\}\\
&\quad +\left\{\hat{p} (\tau)^{-1} \overline{m}_{D}^* (\tau) - p(\tau)^{-1} \overline{m}_p^* (\tau)\right\} + \sqrt{nh}\left\{p(\tau)^{-1}\overline{m}_p^* (\tau) - \tilde{m}(\tau) - h^2 B_{as} (\tau) \right\}\\
&= p(\tau)^{-1} Z_{n,U}^* (\tau) + I_{n}^*(\tau) + II_{n}^*(\tau) + III_n^*(\tau)
\end{split}
\end{equation}

As a general observation, note that taking the sums in the expressions above from $t=1$ to $n$ appears to include the boundary points for $\tilde{m}(\cdot)$, for which the properties of the estimator are not satisfied. However, as we consider $\tau \in [\delta^*, 1-\delta^*]$ and $\tilde{m}\left(\frac{t}{n}\right)$ is always multiplied by $k_t(\tau)$, $k_t(\tau) = 0$ at all points $\frac{t}{n} < \delta$ and $\frac{t}{n} > 1 -\delta$ for large enough $n$, and therefore
\begin{equation*}
\frac{1}{nh} \sum_{t=1}^n k_t (\tau) \tilde{m} \left(\frac{t}{n}\right) = \frac{1}{nh} \sum_{t=[n\delta]+1}^{[n (1- \delta)]} k_t (\tau) \tilde{m} \left(\frac{t}{n}\right),
\end{equation*}
and analogously for all related sums. Hence, in the following we simply take sums from $t=1$ to $n$, under the implicit assumption that $n$ is large enough to do so.

We now first derive \eqref{eq:bvar}. With $\E^* Z_{n,U}^*(\tau)^2 = \frac{1}{nh} \sum_{i=-n+1}^{n-1} \dsum[\abs{i}]{t} k_t(\tau) k_{t+\abs{i}}(\tau) D_t D_{t+\abs{i}} z_t z_{t+\abs{i}} \gamma^{\abs{i}}$, it follows from Lemma \ref{lem:lrv} and the fact that $\sum_{i=0}^{\infty} \gamma^i = \frac{1}{1 - \theta^{1/\ell}} = -\ell / \ln \theta + o(\ell)$,
\begin{equation*}
\E \abs{\E^* Z_{n,U}^*(\tau)^2 - \E_D \E^* Z_{n,U}^*(\tau)^2} \leq \phi_n \sum_{i=0}^{\infty} \gamma^i \beta_i + \frac{\eta_n}{\sqrt{nh}}  \sum_{i=0}^{\infty} \gamma^i \leq \frac{C \ell}{\sqrt{nh}} = o(1).
\end{equation*}
Furthermore
\begin{align*}
\E_D \E^* Z_{n,U}^*(\tau)^2 &= \frac{1}{nh} \sum_{i=-n+1}^{n-1} R_U(i) \dsum[\abs{i}]{t} k_t(\tau) k_{t+\abs{i}}(\tau) D_t D_{t+\abs{i}} \sigma_t \sigma_{t+\abs{i}} \gamma^{\abs{i}} \\
&= \E_D Z_{n,U}^2 + \frac{1}{nh} \sum_{i=-n+1}^{n-1} R_U(i) \dsum[\abs{i}]{t} k_t(\tau) k_{t+\abs{i}} (\tau) \sigma_t \sigma_{t+\abs{i}} D_t D_{t+\abs{i}} (\gamma^{\abs{i}} - 1)
\end{align*}
where $Z_{n,U}$ is defined in Lemma \ref{lem:m_decomp}. For the second term, take $M = M(n)$ such that $1/M + M^2/\ell \rightarrow 0$ as $n \rightarrow \infty$, then we have that
\begin{align*}
&\suptaus \abs{\frac{1}{nh} \sum_{i=-n+1}^{n-1} R_U(i) (\gamma^{\abs{i}} - 1) \dsum[\abs{i}]{t} k_t(\tau) k_{t+\abs{i}} (\tau) \sigma_t \sigma_{t+\abs{i}} D_t D_{t+\abs{i}}} \\
&\quad \leq C \sum_{i=1}^{M} R_U(i) \abs{\theta^{i/\ell} -1} + C \sum_{i=M+1}^{n-1} R_U(i) \abs{\theta^{i/\ell} -1} \\
& \quad \leq C_1 \sum_{i=1}^{M} (1 - \theta^{i/\ell}) + C_2 \sum_{i=M+1}^{\infty} R_U(i) = o(1),
\end{align*}
as $\sum_{i=1}^{M} (1 - \theta^{i/\ell}) \leq C \sum_{i=1}^{M}(-\ln \theta) i / \ell \leq C_1 M^2 /\ell = o(1)$ and $\sum_{i=M+1}^{\infty} R_U(i) = o(M^{-1})$. It then follows directly from the proof of Lemma \ref{lem:m_decomp} that $\plim_{n\rightarrow\infty} \var^* Z_{n,U}^*(\tau) = \sigma_Z^2 (\tau)$ for all $\tau \in [\delta^*, 1 - \delta^*]$.

Now consider $I_{n}^*(\tau)$, which we write as
\begin{align*}
I_{n}^*(\tau) &= \left[\hat{p} (\tau)^{-1} - p(\tau)^{-1} \right] Z_{n,U}^* (\tau) + \frac{1}{\sqrt{nh}} \dsum{t} k_t(\tau) \left[m\left(\frac{t}{n}\right) - \tilde{m}\left(\frac{t}{n}\right)\right] D_t \xi_t^*\\
&= I_{n,1}^* (\tau) + I_{n,2}^*(\tau).
\end{align*}
It follows directly as in the proof of Lemma \ref{lem:m_decomp} and using Jensen's inequality that $\suptaus \E\norm{I_{n,1}^*(\tau)}^* \leq C \max \left\{h^2, \frac{1}{\sqrt{nh}} \right\}$. Furthermore, as
\begin{align*}
\E^* I_{n,2}^*(\tau)^2 &= \frac{1}{(nh)^2} \dsum{s} \dsum{t} k_s(\tau) k_t(\tau) \left[\tilde{m}\left(\frac{s}{n}\right) - m\left(\frac{t}{n}\right) \right] \left[\tilde{m}\left(\frac{t}{n}\right) - m\left(\frac{t}{n}\right) \right] \gamma^{\abs{s-t}} D_s D_t,
\end{align*}
it follows from Jensen's and the Cauchy-Schwartz inequality that
\begin{align*}
\E \norm{I_{n,2}^*(\tau)}^* &\leq \left[\frac{2}{nh} \sum_{i=0}^{n - 1} \gamma^{i} \dsum[i]{t} k_t(\tau) k_{t+i}(\tau) \norm{\tilde{m}\left(\frac{t}{n}\right) - m\left(\frac{t}{n}\right)} \norm{\tilde{m}\left(\frac{t+i}{n}\right) - m\left(\frac{t+i}{n}\right)} \right]^{1/2}\\
&\leq C \left[\sup_{\suptau} \norm{\tilde{m}(\tau)-m(\tau)}^2 \sum_{i=0}^{\infty} \gamma^{i}\right]^{1/2} \leq C_1 \sqrt{\ell} \max \left\{\tilde{h}^2, \frac{1}{\sqrt{n\tilde{h}}} \right\},
\end{align*}
as, using Lemma \ref{lem:m_decomp},
\begin{equation} \label{eq:m_norm}
\begin{split}
\suptau \norm{\tilde{m}(\tau)-m(\tau)} &\leq \tilde{h}^2 \suptau \abs{B_{as}(\tau)} + \frac{1}{\sqrt{n\tilde{h}}} \suptau \norm{Z_{n,U}(\tau)} \\
&\quad + \frac{1}{\sqrt{n\tilde{h}}} \suptau \norm{R_n (\tau)} \leq C \max \left\{\tilde{h}^2, \frac{1}{\sqrt{n\tilde{h}}} \right\}.
\end{split}
\end{equation}

As in the proof of Lemma \ref{lem:m_decomp}, write $II_{n}^*(\tau)$ as
\begin{align*}
II_n^*(\tau) &= p (\tau)^{-1} \left\{\overline{m}_{D}^* (\tau) - p(\tau)^{-1} \hat{p}(\tau) \overline{m}_p^* (\tau)\right\} + \hat{p}(\tau)^{-1} p(\tau)^{-2} \left\{\hat{p}(\tau) - p(\tau) \right\}^2 \overline{m}_p^* (\tau) \\
&\quad + \left\{\hat{p} (\tau)^{-1} - p(\tau)^{-1}\right\} \left\{ \overline{m}_{D}^* (\tau) - \overline{m}_p^* (\tau) \right\}\\
& = II_{n,1}^*(\tau) + II_{n,2}^*(\tau) + II_{n,3}^*(\tau),
\end{align*}
where $II_{n,1}^* (\tau) = p(\tau)^{-1} Z_{n,D}^*(\tau)$ and $Z_{n,D}^*(\tau) = \frac{1}{\sqrt{nh}} \dsum{t} k_t(\tau) \left[\tilde{m}\left(\frac{t}{n}\right) - p(\tau)^{-1} \overline{m}_p^* (\tau) \right] \left[ D_t  - p\left(\frac{t}{n}\right) \right]$. Define $Z_{n,D} (\tau) = \frac{1}{\sqrt{nh}} \dsum{t} k_t(\tau) \left[m\left(\frac{t}{n}\right) - p(\tau)^{-1} \overline{m}_p (\tau) \right] \left[ D_t  - p\left(\frac{t}{n}\right) \right]$. Then
\begin{equation} \label{eq:II_n11}
\begin{split}
&Z_{n,D}^*(\tau) - Z_{n,D} (\tau) = \frac{1}{\sqrt{nh}} \dsum{t} k_t(\tau) \left[\tilde{m}\left(\frac{t}{n}\right) - m\left(\frac{t}{n}\right) \right] \left[ D_t  - p\left(\frac{t}{n}\right) \right]\\
&\quad + \frac{1}{\sqrt{nh}} p(\tau)^{-1} \left[\overline{m}_p^* (\tau) - \overline{m}_p (\tau) \right] \dsum{t} k_t(\tau) \left[ D_t  - p\left(\frac{t}{n}\right) \right] = II_{n,11}^*(\tau) + II_{n,12}^*(\tau).
\end{split}
\end{equation}
It follows directly by \eqref{eq:m_norm} that 
\begin{align*}
\norm{II_{n,11}^*(\tau)} &\leq \left( \suptaus \norm{\tilde{m}(\tau) - m(\tau) }^2 \sum_{i=-n+1}^{n-1} \frac{1}{nh}\dsum[\abs{i}]{t} k_t(\tau) k_{t+\abs{i}}(\tau) \abs{\cov(D_t, D_{t+\abs{i}})} \right)^{1/2}\\
& \leq C \max\left\{\tilde{h}^2, \frac{1}{\sqrt{n\tilde{h}}} \right\}.
\end{align*}
Furthermore, as
\begin{align*}
\suptaus \norm{\overline{m}_p^* (\tau) - \overline{m}_p (\tau)} &\leq \suptaus \norm{\tilde{m}_n (\tau) - m(\tau)} \suptaus \frac{1}{nh} \dsum{t} p\left(\frac{t}{n}\right) k_t(\tau)\\
& \leq C\max\left\{\tilde{h}^2, \frac{1}{\sqrt{n\tilde{h}}} \right\}
\end{align*}
and $\frac{1}{\sqrt{nh}} \norm{\dsum{t} k_t(\tau) \left[ D_t  - p\left(\frac{t}{n}\right) \right]} \leq C_1$ by Lemma \ref{lem:D_results}$(i)$, it follows that $\suptaus \norm{II_{n,12}^*(\tau)} \leq C \max\left\{\tilde{h}^2, \frac{1}{\sqrt{n\tilde{h}}} \right\}$ as well. As it was shown in Lemma \ref{lem:m_decomp} that $\norm{Z_{n,D} (\tau)} \leq C h^2$, it follows that $\norm{ II_{n,1}^* (\tau)} \leq C \max\left\{\tilde{h}^2, \frac{1}{\sqrt{n\tilde{h}}} \right\}$.

As $\suptau\norm{\overline{m}_p^* (\tau) - \overline{m}_p (\tau)} \leq C\max\left\{\tilde{h}^2, \frac{1}{\sqrt{n\tilde{h}}} \right\}$, it follows as in the proof of $II_{n,2}(\tau)$ in Lemma \ref{lem:m_decomp} that $\suptaus \norm{II_{n,2}^*(\tau)} \leq C \max\left\{\sqrt{n h^9}, \frac{1}{\sqrt{nh}}\right\}$, while $\overline{m}_{D}^* (\tau) - \overline{m}_p^* (\tau) = \overline{m}_{D} (\tau) - \overline{m}_p (\tau) + II_{n,11}^*(\tau)$, and $II_{n,11}^*(\tau)$ is defined in \eqref{eq:II_n11}. Therefore it follows directly as in the proof of $II_{n,3}(\tau)$ in Lemma \ref{lem:m_decomp} that $\suptau \norm{II_{n,3}^* (\tau)} \leq C \max \left\{\tilde{h}^{2}, \frac{1}{\sqrt{nh}}\right\}$.

It follows from Lemma \ref{lem:m_decomp} that $\tilde{m}(\tau) = m(\tau) + \tilde{h}^2 B_{as} + \tilde{R}_n (\tau)$, where $\suptau \norm{\tilde{R}_n(\tau)} \leq \frac{C}{\sqrt{n\tilde{h}}}$. Substituting this into $III_n^*(\tau)$ we find that
\begin{align*}
III_n^*(\tau) &= \left[p(\tau)^{-1} \frac{1}{\sqrt{nh}} \dsum{t} p\left(\frac{t}{n}\right) k_t(\tau) m\left(\frac{t}{n}\right) - m(\tau) - h^2 B_{as} (\tau) \right]\\
&\quad + \tilde{h}^2 \left[p(\tau)^{-1} \frac{1}{\sqrt{nh}} \dsum{t} p\left(\frac{t}{n}\right) k_t(\tau) B_{as}\left(\frac{t}{n}\right) - B_{as} (\tau) \right]\\
&\quad + \left[p(\tau)^{-1} \frac{1}{\sqrt{nh}} \dsum{t} p\left(\frac{t}{n}\right) k_t(\tau)  \tilde{R}_n \left(\frac{t}{n}\right) - \tilde{R}_n (\tau)\right]\\
& = III_{n,1}^*(\tau) + III_{n,2}^*(\tau) + III_{n,3}^*(\tau).
\end{align*}
Note that $III_{n,1}^*(\tau)$ is equal to $III_{n}(\tau)$ defined in Lemma \ref{lem:m_decomp}, such that $\suptaus \abs{III_{n,1}^*(\tau)} \leq C \max\left\{\sqrt{n h^7}, \frac{1}{\sqrt{nh}} \right\}$. Furthermore, using the definition of $B_{as}(\tau)$ in \eqref{eq:as_pars}, we can write
\begin{equation*}
\frac{1}{\sqrt{nh}} \dsum{t} p\left(\frac{t}{n}\right) k_t(\tau) B_{as}\left(\frac{t}{n}\right) = \frac{1}{2} \mu_2 \frac{1}{\sqrt{nh}} \dsum{t} k_t(\tau) f\left(\frac{t}{n}\right),
\end{equation*}
where $f(\tau) = p(\tau)^{-1} \left[m p \right]^{(2)} (\tau)$ is Lipschitz continuous. Then, by Lemma \ref{lem:kernel_function_limits}$(i)$ it follows that
\begin{align*}
\abs{III_{n,2}^*(\tau)} &\leq C\sqrt{nh} \tilde{h}^2 \abs{\frac{1}{nh} \dsum{t} k_t(\tau) f\left(\frac{t}{n}\right) - f (\tau)} \leq C_1 \max \left\{\frac{\tilde{h}^2}{\sqrt{nh}}, \sqrt{n h^5 \tilde{h}^4} \right\}.
\end{align*}
Finally, for $III_{n,3}^*(\tau)$ we have
\begin{align*}
\norm{ III_n^*(\tau)} &\leq C \sqrt{nh} \suptau \norm{\tilde{R}_n (\tau)} \frac{1}{nh} \dsum{t} k_t^2(\tau) + \norm{\tilde{R}_n (\tau)} \leq C_1 \sqrt{\frac{h}{\tilde{h}}} + \frac{C_2}{\sqrt{n \tilde{h}}}.
\end{align*}

Collecting all remainder terms, it then follows that
\begin{equation*}
\suptaus \E \norm{ R_n^* (\tau)}^* \leq C \max\left\{\frac{1}{\sqrt{nh}}, \sqrt{n h^7}, \sqrt{\frac{h}{\tilde{h}}}, \sqrt{n h^{5} \tilde{h}^4}, \tilde{h}^{2}, \sqrt{\ell\tilde{h}^4}, \sqrt{\frac{\ell}{n\tilde{h}}} \right\}. \qedhere
\end{equation*}
\end{proof}

\begin{proof}[{\bf Proof of Theorem \ref{th:basdis}}]
Given Lemma \ref{lem:mb_decomp}, we only have to prove asymptotic normality of $Z_{n,U}^* (\tau)$. As in the proof of Theorem \ref{th:asdis}, we establish asymptotic normality of the bootstrap process using a blocking technique. By the stationarity of $\xi_t^*$, we can write $\xi_t^* = \sum_{j=0}^{\infty} \gamma^j \nu_{t-j}^*$ with $\nu_t^*$ for $t \leq 1$ defined analogously as for $t>1$. Take an $M$ such that $M/\ell \rightarrow\infty$ as $n\rightarrow\infty$, and truncate the MA($\infty$) representation of $\xi_t^*$ at $M$ lags to define $\xi_{t,M}^* = \sum_{j=0}^{M} \gamma^j \nu_{t-j}^*$. Then we write
\begin{align*}
Z_{n,U}^* (\tau) &= \frac{1}{\sqrt{nh}} \dsum{t} k_t(\tau) D_t z_t \xi_{t,M}^* + \frac{1}{\sqrt{nh}} \dsum{t} k_t(\tau) D_t z_t \left(\sum_{j=M+1}^{\infty} \gamma^j \nu_{t-j}^* \right) \\
&= \overline{Z}_{n,M}^* (\tau) + \overline{W}_{n,M}^* (\tau).
\end{align*}
Applying Markov's inequality twice, we have that
\begin{align*}
& \E \E^* \overline{W}_{n,M}^{*}(\tau)^2\leq \frac{2}{nh} (1 - \gamma^2) \sum_{i=0}^{n-1} \abs{R_U(i)} \dsum[i]{t} k_t(\tau) k_{t+i}(\tau) \left(\sum_{j=M+1}^{\infty} \gamma^{2 j + i} \right)\\
& \qquad \leq C \gamma^{2M} \sum_{i=0}^{\infty} \gamma^{i} \abs{R_U(i)} = C \theta^{2M/\ell} \sum_{i=0}^{\infty} \theta^{i/\ell} \abs{R_U(i)} \leq C_1 \theta^{2M/\ell} = o(1),
\end{align*}
as $M/\ell \rightarrow \infty$. It then follows that $\overline{W}_{n,M}^{*}(\tau) = o_p^*(1)$ for all $\tau \in [\delta^*, 1 - \delta^*]$.

Now let $\overline{Z}_{t,M}(\tau) = \sum_{j=1}^k X_{n,j}^*(\tau) + \sum_{j=1}^k Y_{n,j}(\tau)$, where
\begin{align*}
&X_{n,j}^* (\tau) = \frac{1}{\sqrt{nh}} \sum_{t = B_j + 1}^{B_j + a} k_t(\tau) D_t z_t \xi_{t,M}^*, 
&Y_{n,j}(\tau) = \frac{1}{\sqrt{nh}} \sum_{t = B_j + a + 1}^{B_j + a + b} k_t(\tau) D_t z_t \xi_{t,M}^*,
\end{align*}
with $B_j = (j-1)(a+b)$ and $k = \lceil n / (a+b) \rceil$. Take sequences $a = a(n)$ and $b = b(n)$ such that $a/(n h) + M/a \rightarrow 0$ and $b/a + M/b \rightarrow 0$ as $n \rightarrow \infty$.

We first show that $\sum_{j=1}^k Y_{n,i}^*(\tau)=o_p^*(1)$. Consider $n$ large enough such that $a(n)>M$ and the blocks $Y_{n,i}$ are mutually independent conditionally on the original data. Then, with
\begin{equation*}
R_M^*(i) = \E^* \xi_t^* \xi_{t+\abs{i}}^* = (1- \gamma^2) \sum_{j=0}^{M} \gamma^j \gamma^{j+\abs{i}} = \theta^{\abs{i}/\ell} (1 - \theta^{2(M+1)/\ell} ) \leq \theta^{\abs{i}/\ell}
\end{equation*}
for $\abs{i} \leq M-1$, we have that
\begin{align*}
&\E \E^* \left(\sum_{j=1}^k Y_{n,j}^*(\tau)\right)^2 = \sum_{j=1}^k \E^* Y_{n,j}(\tau)^2 \leq \frac{1}{nh} \sum_{i=-b+1}^{b-1} \theta^{i/\ell} R_U(i) \sum_{j=1}^k \sum_{t=B_j + a + 1}^{B_j + a + b -\abs{i}} k_{t}(\tau) k_{t+\abs{i}}(\tau) \\
&\leq C_1 \frac{b} {a} = o(1).
\end{align*}

As in the proof of Theorem \ref{th:asdis}, we employ the Lindeberg CLT to establish asymptotic normality of $\sum_{j=1}^k X_{n,j}^*(\tau)$, as for $n$ sufficiently large, $b>M$ and the blocks $X_{n,j}^*(\tau)$ are independent. First we show that the asymptotic variance is equal to $p(\tau)^2 \sigma_{as}^2(\tau)$. Note that $\E^* \sum_{j=1}^k X_{n,j}^* (\tau) = 0$ and
\begin{align*}
&\E^* \left(\sum_{j=1}^k X_{n,j}^* (\tau)\right)^2 = \frac{1}{nh} \sum_{i=-a+1}^{a-1} R_M^*(i) R_U(i) \sum_{j=1}^k \sum_{t=B_j + 1}^{B_j + a - \abs{i}} k_t(\tau) k_{t+\abs{i}}(\tau) D_t D_{t+\abs{i}} \sigma_t \sigma_{t+\abs{i}}\\
&\quad + \frac{1}{nh} \sum_{i=-a+1}^{a-1} \sum_{j=1}^k \sum_{t=B_j + 1}^{B_j + a - \abs{i}} k_t(\tau) k_{t+\abs{i}}(\tau) D_t D_{t+\abs{i}} \sigma_t \sigma_{t+\abs{i}} R_M^*(i) \left[u_t u_{t+\abs{i}} - R_U(i)\right]\\
&= A_{X,n}^* (\tau) + B_{X,n}^* (\tau).
\end{align*}
By adapting Lemma \ref{lem:lrv}, we find that
\begin{align*}
\E \abs{B_{X,n}^* (\tau)} &\leq 2 \sum_{i=0}^{a-1} \E \abs{\sum_{j=1}^k \sum_{t=B_j + 1}^{B_j + a - \abs{i}} k_{t}(\tau) k_{t+\abs{i}}(\tau) D_t D_{t+\abs{i}} \sigma_t \sigma_{t+\abs{i}} \left[u_{t} u_{t+\abs{i}} - R_U(i) \right]} \\
&\leq \phi_n \sum_{i=0}^{a-1} \beta_i + \frac{\eta_n}{\sqrt{nh}} = o(1).
\end{align*}
Furthermore, the arguments used to prove \eqref{eq:bvar} in Lemma \ref{lem:mb_decomp} show that $A_{X,n}^* (\tau) \xrightarrow{p} p(\tau)^2 \sigma_{as}^2 (\tau)$.

The final step is to verify that, for every $\kappa>0$, $\sum_{j=1}^k \E^* \left[\frac{X_{n,j}^*(\tau)^2}{\omega_n^{*2}}\mathbbm{1} \left(\abs{\frac{X_{n,j}^*(\tau)^2}{\omega_n^{*2}}}>\kappa \right)\right] = o_p(1)$, where $\omega_n^{*2} = \E^*\left(\sum_{j=1}^k X_{n,j}^*(\tau) \right)^2$. Similarly as in the proof of Theorem \ref{th:asdis}, we define $\tilde{X}_t^* (\tau) = D_t k_t(\tau) \sigma_t u_t \xi_{t,M}^*$, as an $L_{4}$-mixingale conditionally on the original data, with $\norm{\tilde X_t}_4^{*} \leq C \abs{u_t} k_t(\tau)$, such that Lemma 2 of \citet{Hansen91} then implies that $\E^* X_{n,i}(\tau)^{4} \leq \frac{C}{(nh)^{2}} \left(\sum_{t=B_j+1}^{B_j+a} k_t(\tau)^{2} u_t^2 \right)^{2}$, and therefore
\begin{align*}
&\sum_{j=1}^k \E^* \left[\frac{X_{n,j}^*(\tau)^2}{\omega_n^{*2}}\mathbbm{1} \left(\abs{\frac{X_{n,j}^*(\tau)^2}{\omega_n^{*2}}}>\kappa \right) \right] \leq \frac{C}{(nh)^{2} \omega^{*4}} \sum_{j=1}^k \left(\sum_{t=B_j+1}^{B_j+a} k_t(\tau)^{2} u_t^2 ,\right)^{2},
\end{align*}
where $\omega_n^{*-4} = O_p(1)$. By Minkowski's inequality, stationarity of $u_t$, and $\E u_t^4 < \infty$ -- let $M \rightarrow \infty$ in \eqref{eq:u4} -- we have that
\begin{align*}
& \frac{1}{(nh)^2} \sum_{j=1}^k \E \left(\sum_{t=B_j+1}^{B_j+a} k_t(\tau)^{2} u_t^2 \right)^{2} \leq \frac{1}{(nh)^2} (\E u_t^4 )\sum_{i=1}^k \left(\sum_{t=B_j+1}^{B_j+a} k_t(\tau)^{2} \right)^{2} \leq \frac{C_1 a}{nh} = o(1). \qedhere
\end{align*}
\end{proof}

\subsection{Uniform Results}
Before deriving Theorem \ref{th:uniform}, we propose a few auxiliary lemmas aimed at establishing stochastic equicontinuity, which will be needed to extend the pointwise results to uniformity.

\begin{lemma} \label{lem:SE}
Let $\{X_t\}_{t=1}^n$ be a stochastic process with $\limsup_{n\rightarrow \infty} \sum_{i=1}^{n-1} \max_{1 \leq t \leq n - i} \abs{\E X_t X_{t+i}} < \infty$. Then, for any $\tau_0 \in (0,1)$ and $\tau_1, \tau_2 \in [-1,1]$, there exists an $N > 0$ such that for all $n > N$
\begin{align*}	
&\abs{\frac{1}{\sqrt{nh}} \sum_{t=1}^{n} \left[k_t(\tau_0 + \tau_1 h) - k_t (\tau_0 + \tau_2 h) \right] X_t} \leq B_n \abs{\tau_1 - \tau_2},
\end{align*}
where $B_n$ is a random variable such that $\sup_{n>N, \tau_1,\tau_2 \in [-1,1]} \E B_n^2 < \infty$.
\end{lemma}

\begin{proof}[{\bf Proof of Lemma \ref{lem:SE}}]
See Supplementary Appendix B.
\end{proof}

\begin{lemma} \label{lem:SEb}
Let $\{X_t^*\}_{t=1}^n$ be a bootstrap process defined conditionally on a process $\{X_t\}_{t=1}^n$ with $\limsup_{n\rightarrow \infty} \sum_{i=1}^{n-1} \sup_{1 \leq t \leq n - i} \abs{\E \E^* X_t^* X_{t+i}^*} < \infty$. Then, for any $\tau_0 \in (0, 1)$ and $\tau_1, \tau_2 \in [-1,1]$, there exists an $N > 0$ such that for all $n > N$
\begin{align*}	
&\abs{\frac{1}{\sqrt{nh}} \dsum{t} \left[k_t(\tau_0 + \tau_1 h) - k_t (\tau_0 + \tau_2 h) \right] X_t^*} \leq B_n^* \abs{\tau_1 - \tau_2},
\end{align*}
where $B_n^*$ is a random variable such that $\sup_{n>N, \tau_1,\tau_2 \in [-1,1]} \E \E^* B_n^{*2} < \infty$.
\end{lemma}

\begin{proof}[{\bf Proof of Lemma \ref{lem:SEb}}]
See Supplementary Appendix B.
\end{proof}

\begin{lemma} \label{lem:R_unif}
Let $R_n (\tau)$ and $R_n^* (\tau)$ be defined as in Lemmas \ref{lem:m_decomp} and \ref{lem:mb_decomp} respectively. Then we have, for all $\tau_0 \in (0,1)$, that $\suptaut \abs{R_n(\tau_0 + \tau_1 h)} = o_p(1)$ and $\suptaut \abs{R_n^*(\tau_0 + \tau_1 h)} = o_p^*(1)$.
\end{lemma}

\begin{proof}[{\bf Proof of Lemma \ref{lem:R_unif}}]
See Supplementary Appendix B.
\end{proof}

\begin{proof}[{\bf Proof of Theorem \ref{th:uniform}}]
It follows directly from Lemma \ref{lem:R_unif} that
\begin{align*}
& \suptaut \abs{Z_{\tau_0, n} (\tau) - p(\tau_0 + \tau h)^{-1} Z_{n,U} (\tau_0 + \tau h) - B_{as} (\tau_0)} \leq \suptaut \abs{R_n (\tau_0 + \tau h} = o_p(1),\\ 
& \suptaut \abs{Z_{\tau_0, n}^* (\tau) - p(\tau_0 + \tau h)^{-1}Z_{n,U}^* (\tau_0 + \tau h) - B_{as} (\tau_0)} \leq \suptaut \abs{R_n^* (\tau_0 + \tau h} = o_p^*(1),
\end{align*}
such that we only have to consider $p(\tau_0 + \tau h)^{-1} Z_{n,U} (\tau)$ and $p(\tau_0 + \tau h)^{-1} Z_{n,U}^* (\tau)$ in the following.

We next establish the asymptotic covariances. With
\begin{align*}
\E_D Z_{n,U}(\tau_0 + \tau_1 h) Z_{n,U}(\tau_0 + \tau_2 h) 
&= \frac{1}{nh} \sum_{i=-n+1}^{n-1} R_U(i) \sum_{t=1}^{n-\abs{i}} k_t(\tau_0 + \tau_1 h) k_{t+\abs{i}}(\tau_0 + \tau_2 h) \sigma_t \sigma_{t+\abs{i}} D_t D_{t+\abs{i}},
\end{align*}
we apply Lemma \ref{lem:D_results}$(iii)$ with $\tilde g_i (\tau_1, \tau_2) = \sigma(\tau_1) \sigma(\tau_2)$ to find that
\begin{align*}
&\suptaun \suptautt \norm{ \E_D Z_{n,U}(\tau_0 + \tau_1 h) Z_{n,U}(\tau_0 + \tau_2 h) - p(\tau_0) \sigma^2(\tau_0) \kappa (\tau_1 - \tau_2) \Omega_U}\\
&\qquad\leq \sum_{i=-n+1}^{n-1} \abs{R_U(i)} \left[\frac{C_1}{\sqrt{nh}} + \frac{C_2}{\sqrt{nh}} \max \left\{h, \frac{1}{n h^2} \right\} + S_n(i) \right] + \sum_{i=-\infty}^\infty \mathbbm{1} (\abs{i} \geq n) \abs{R_U(i)},\\
& \qquad\leq \frac{C_3}{\sqrt{nh}} + \frac{C_3}{\sqrt{nh}} \max \left\{h, \frac{1}{n h^2} \right\} + \phi_n,
\end{align*}
from which it follows by the law of iterated expectations that 
\begin{equation*}
\suptau \suptautt \abs{\E Z_{n,U}(\tau_0 + \tau_1 h) Z_{n,U}(\tau_0 + \tau_2 h) - p(\tau_0) \sigma_{W,\tau_0} (\tau_1, \tau_2)} = o(1).
\end{equation*}
It then follows that
\begin{align*}
&\suptau \suptautt \abs{\E Z_{\tau_0, n} (\tau_1) Z_{\tau_0, n} (\tau_2) - \sigma_{W,\tau_0}(\tau_1, \tau_2)} \\
&\leq \suptau \suptautt \abs{\E Z_{\tau_0, n} (\tau_1) Z_{\tau_0, n} (\tau_2) - p(\tau_0 + \tau h)^{-1} p(\tau_0 + \tau_2 h)^{-1} p(\tau_0)^2 \sigma_{W,\tau_0} (\tau_1, \tau_2)} \\
&\quad + \suptau \suptautt \abs{p(\tau_0 + \tau_1 h)^{-1} p(\tau_0 + \tau_2 h)^{-1} p(\tau_0)^2 - 1} \sigma_{W,\tau_0} (\tau_1, \tau_2) \leq o(1) + o(h) = o(1).
\end{align*}

We follow the same steps as in the proof of the asymptotic bootstrap variance in Lemma \ref{lem:mb_decomp} for the bootstrap covariances. Note that
\begin{equation*}
\E^* Z_{n,U}^* (\tau_0 + \tau_1 h) Z_{n,U}^* (\tau_0 + \tau_2 h) = \frac{1}{nh} \sum_{i=1-n}^{n-1} \dsum[\abs{i}]{t} k_t(\tau_0 + \tau_1 h) k_{t+i}(\tau_0 + \tau_2 h) D_t D_{t+i} z_t z_{t+i} \gamma^i,
\end{equation*}
where it follows from a straightforward adaptation of Lemma \ref{lem:lrv} allowing for different $\tau$'s that
\begin{align*}
&\E \abs{\E^* Z_{n,U}^* (\tau_0 + \tau_1 h) Z_{n,U}^* (\tau_0 + \tau_2 h) - \E_D \E^* Z_{n,U}^* (\tau_0 + \tau_1 h) Z_{n,U}^* (\tau_0 + \tau_2 h)} = o(1).
\end{align*}
To conclude this part, we can show as in Lemma \ref{lem:mb_decomp} that
\begin{align*}
\suptauns \suptautt \abs{\E_D \E^* Z_{n,U}^* (\tau_0 + \tau_1 h) Z_{n,U}^* (\tau_0 + \tau_2 h) - \E_D Z_{n,U}(\tau_0 + \tau_1 h) Z_{n,U}(\tau_0 + \tau_2 h)} = o_p(1).
\end{align*}

Finite-dimensional convergence of the vectors $\left(Z_{\tau_0,n} (\tau_1), \ldots, Z_{\tau_0,n} (\tau_m) \right)^\prime$ and $\left(Z_{\tau_0,n}^* (\tau_1), \ldots, Z_{\tau_0,n}^* (\tau_m) \right)^\prime$ for $(\tau_1, \ldots ,\tau_m)^\prime \in \left[-1,1\right]^m$ follows from Theorems \ref{th:asdis} and \ref{th:basdis} and the Cram\'er-Wold device; it remains to show tightness. By applying Lemma \ref{lem:SE} with $X_t = D_t z_t$ and Lemma \ref{lem:SEb} with $X_t^* = D_t z_t \xi_t^*$, it follows directly that
\begin{equation*}
\E\left( Z_{\tau_0,n}(\tau_1) - Z_{\tau_0,n}(\tau_2) \right)^2 \leq \E B_n^2 \abs{\tau_1 - \tau_2 }^2 \quad \text{and} \quad \E \E^* \left( Z_{\tau_0,n}(\tau_1) - Z_{\tau_0,n}(\tau_2) \right)^2 \leq \E \E^* B_n^{*2} (\tau_1 - \tau_2)^2,
\end{equation*}
where $\E B_n^2 \leq C$ and $\E \E^* B_n^{*2} \leq C$ and tightness follows by Theorem 12.3 of \citet{Billingsley}.
\end{proof}

\clearpage
\section{Additional Proofs}

In this appendix we prove the auxiliary lemmas not proven in the main paper.

\begin{proof}[{\bf Proof of Lemma \ref{lem:cov}}]
For $(i)$, note that, by the Cauchy-Schwarz inequality and Assumption \ref{as:MD}, we have that
\begin{align*}
\sum_{i=0}^{n-1} i \max_{1 \leq t \leq n-i} \abs{\cov (D_t, D_{t+i})} &= \sum_{i=0}^{n-1} i \sup_t \abs{\E \left\{ \left[D_t - \E D_t \right] \E \left[ D_{t+i} - \E D_{t+i} | \mathcal{F}_t \right] \right\}}\\
& \leq \sum_{i=0}^{n-1} i \max_{1 \leq t \leq n-i} \norm{ D_t - \E D_t} \norm{\E \left[ D_{t+i} - \E D_{t+i} | \mathcal{F}_t \right]}\\
& \leq \sum_{i=0}^{n-1} i C \zeta_i \leq C \sum_{i=0}^{\infty} i \zeta_i < \infty.
\end{align*}

For $(ii)$, it follows from Assumption \ref{as:LP} that
\begin{equation*}
\begin{split}
\sum_{i=1}^{n-1} i \max_{1 \leq t \leq n-i} \abs{\cov (u_t, u_{t+i})} &= \sum_{i=1}^{n-1} i \abs{R_U(i)} = \sigma_{\varepsilon}^2 \sum_{i=1}^{n-1} i \abs{\sum_{j=0}^\infty \psi_j \psi_{j+\abs{k}} }\\
& \leq \sum_{i=0}^{\infty}  \sum_{j=0}^\infty i\abs{\psi_i} \abs{\psi_j} = \left(\sum_{i=0}^\infty i \abs{\psi_i} \right) \left(\sum_{j=0}^\infty \abs{\psi_j} \right) < \infty.
\end{split}
\end{equation*}

Finally for $(iii)$, as $0 \leq D_t \leq 1$ for all $t$, it follows that $\abs{\E D_t D_{t+i} u_t u_{t+i}} \leq \abs{\E u_t u_{t+i}}$, such that the result directly follows from $(ii)$.
\end{proof}

\begin{proof}[{\bf Proof of Lemma \ref{lem:kernel_sum}}]
By the compact support of $K(\cdot)$, $k_t (\tau) = 0$ if $\abs{t/n - \tau} > C_K h$ for some $C_K > 0$. Therefore
\begin{align*}
&\suptau \max_{0 \leq i \leq n} \frac{1}{nh} \sum_{t=1}^{n} k_t (\tau) k_{t+i} (\tau) \leq \frac{1}{nh} \suptau \sup_\omega K(\omega) \sum_{t=n (\tau - C_K h)}^{n (\tau + C_K h)} k_t (\tau)\\
& \leq \frac{1}{nh} \sup_\omega K(\omega)^2 2 C_K n h \leq C. \qedhere
\end{align*}
\end{proof}

\begin{proof}[{\bf Proof of Lemma \ref{lem:kernel_function_limits}}]
For part $(i)$, note that
\begin{equation*}
\begin{split}
&\frac{1}{nh} \sum_{t=1}^{n} f\left(\frac{t}{n} \right) k_t(\tau) - f(\tau) \kappa_1 = \left[\frac{1}{nh} \sum_{t=1}^{n} f\left(\frac{t}{n} \right) k_t(\tau) - h^{-1} \int_{0}^{1} f(x) K\left(\frac{x-\tau}{h}\right) \d x \right]\\
&\quad + \left[h^{-1} \int_{0}^{1} f(x) K\left(\frac{x-\tau}{h}\right)\d x - f(\tau) \kappa_1 \right] = I_{1,n} (\tau) + I_{2,n} (\tau).
\end{split}
\end{equation*}
We first consider $I_{1,n} (\tau)$. For a continuous and Riemann-integrable function $g(\cdot)$, we have the integral approximation bound \citep[cf.][p.~79, (6.5)]{Buhlmann}
\begin{equation}
\abs{\frac{1}{n}\sum_{t=1}^n g\left(\frac{t}{n} \right)-\int_0^1 g(z) dz }\leq\sup_{\abs{x-y} \leq n^{-1}} \abs{g(x)-g(y)}.
\label{eq:int_approx}
\end{equation}
Take $g(x)=\frac{1}{h} f(x) K\left(\frac{x-\tau}{h}\right)$, then the bound yields
\begin{align*}
\abs{I_{1,n} (\tau)} &\leq\sup_{|x-y|\leq n^{-1}} h^{-1}\left|K\left(\frac{x-\tau}{h}\right)f(x)-K\left(\frac{y-\tau}{h}\right)f(y)\right|\\
&\leq\sup_{|x-y|\leq n^{-1}}\left[K\left(\frac{x-\tau}{h}\right)\frac{\left|x-y\right|}{h} + f(y) \frac{\left|x-y\right|}{h^2}\right]\leq\frac{C_1}{nh} + \frac{C_2}{nh^2}.
\end{align*}

For $I_{2,n} (\tau)$ perform a change of variables with $u=(x-\tau)/h$ such that $\int_{0}^{1} K\left(\frac{x-\tau}{h}\right) f(x)dx = \int_{-\tau/h}^{(1-\tau)/h} K(u) f(\tau+uh)du$. Take an $N$ such that for all $n>N$, $\min\{\tau/h,(1-\tau)/h\}$ will be larger than the bounds of the compact support of $K(\cdot)$; it then follows that $\int_{-\tau/h}^{(1-\tau)/h} K(u) f(uh+\tau)du = \int_{\mathbb{R}} K(u) f(\tau+uh)du$ for all $n>N$. Then a Taylor expansion of $f(\tau+uh)$ around $f(\tau)$ yields
\begin{equation*}
f(\tau+uh) = f(\tau) + f^{(1)}(\tau) uh + \frac{1}{2}f^{(2)}(\tau^*)(uh)^2,
\end{equation*}
where $\abs{\tau^* - \tau} \leq u h$. As $\mu_1 = \int_{\mathbb{R}} K(u)u \d u= 0$, it follows that
\begin{align*}
\int_{\mathbb{R}} K(u)f(\tau+uh) \d u &= f(\tau) \int_{\mathbb{R}} K(u) \d u + f^{(1)}(\tau) h \int_{\mathbb{R}} K(u)u \d u + \frac{1}{2} h^2 \int_{\mathbb{R}} f^{(2)}(\tau^*) K(u) u^2 \d u \\
&= f(\tau)\kappa_1 + \frac{1}{2} h^2 \int_{\mathbb{R}} f^{(2)}(\tau^*) K(u) u^2 \d u.
\end{align*}
It then follows from the Lipschitz continuity of $f^{(2)}(\cdot)$ that $\abs{f^{(2)}(\tau^*) - f^{(2)}(\tau)} \leq C uh$ and consequently
\begin{align*}
\abs{I_{2,n}(\tau)} &\leq \abs{\int_{\mathbb{R}} K(u) f(uh+\tau)du - f(\tau) \kappa_1 - \frac{1}{2} h^2 \int_{\mathbb{R}} f^{(2)}(\tau^*) K(u) u^2 \d u} \\
&\quad + \frac{1}{2} h^2 \abs{\int_{\mathbb{R}} \left[f^{(2)}(\tau^*) - f^{(2)}(\tau) \right] K(u) u^2 \d u} \leq 0 + \frac{C}{2} h^3 \int_{\mathbb{R}} K(u) \abs{u}^3 \d u \leq C_1 h^3,
\end{align*}
as $\int_{\mathbb{R}} K(u) \abs{u}^3 \d u < \infty$ by the compact support of $K(\cdot)$. As none of the bounds depend on $\tau$, and the Taylor expansion is appropriate for any $\tau \in (0,1)$, the result follows.

For part $(ii)$, we can write
\begin{align*}
&\frac{1}{nh} \sum_{t=1}^{n-i} g \left(\frac{t}{n}, \frac{t+i}{n} \right) k_t(\tau_0 + \tau_1 h) k_{t+i}(\tau_0 + \tau_2 h) - g(\tau_0, \tau_0) \kappa (\tau_1 - \tau_2) \\
&= \frac{1}{nh} \sum_{t=1}^{n-i} k_t(\tau_0 + \tau_1 h) \left[g \left(\frac{t}{n}, \frac{t+i}{n} \right) k_{t+i}(\tau_0 + \tau_2 h) - g \left(\frac{t}{n}, \frac{t}{n} \right) k_t(\tau_0 + \tau_2 h) \right]\\
&\quad + \frac{1}{nh} \sum_{t=1}^{n-i} g \left(\frac{t}{n}, \frac{t}{n}\right) k_t(\tau_0 + \tau_1 h) k_t(\tau_0 + \tau_2 h) - g(\tau_0, \tau_0) \kappa (\tau_1 - \tau_2) = II_{1,n} (\tau) + II_{2,n}(\tau),
\end{align*}
where $\tau = (\tau_0, \tau_1, \tau_2)^\prime$. Then
\begin{align*}
\abs{II_{1,n} (\tau)} &\leq \frac{1}{nh} \sum_{t=1}^{n-i} k_t(\tau_0 + \tau_1 h) \left\{\abs{g \left(\frac{t}{n}, \frac{t+i}{n} \right)} \abs{ k_{t+i}(\tau_0 + \tau_2 h) - k_t(\tau_0 + \tau_2 h)} + k_{t}(\tau_0 + \tau_2 h) \right. \\
&\quad \left. \times \abs{g \left(\frac{t}{n}, \frac{t}{n} \right)- g \left(\frac{t}{n}, \frac{t+i}{n} \right)} \right\} \leq \frac{C_g}{nh} \sum_{t=1}^{n} k_t(\tau_0 + \tau_1 h) \left(\frac{i}{nh} +  \frac{i}{n} \right) \leq C_g C_1 \frac{i}{n h}
\end{align*}
and
\begin{align*}
&\abs{II_{2,n} (\tau)} = \frac{1}{nh} \sum_{t=1}^{n-i} k_t(\tau_0 + \tau_1 h) k_t(\tau_0 + \tau_2 h) \left[g \left(\frac{t}{n}, \frac{t}{n}\right) - g(\tau_0, \tau_0) \right] \\
&\quad + \frac{1}{h} g(\tau_0, \tau_0) \left[\frac{1}{n} \sum_{t=1}^{n-i} k_t(\tau_0 + \tau_1 h) k_t(\tau_0 + \tau_2 h) - \int_{0}^{1} K\left(\frac{x - \tau_0 - \tau_1 h}{h}\right) K\left(\frac{x - \tau_0 - \tau_2 h}{h}\right) \d x \right]\\
&\quad + g(\tau_0, \tau_0) \left[h^{-1} \int_{0}^{1} K\left(\frac{x - \tau_0 - \tau_1 h}{h}\right) K\left(\frac{x - \tau_0 - \tau_2 h}{h}\right) \d x - \kappa (\tau_1 - \tau_2) \right]\\
&= II_{21,n} (\tau) + II_{22,n} (\tau) + II_{23,n} (\tau).
\end{align*}
By the compact support of $K(\cdot)$, $k_t (\tau_0 + \tau_1 h) >0$ only if $\abs{t/n - \tau_0 - \tau_1 h} \leq C_K h$ for some $C_K > 0$, which implies that $\abs{t/n - \tau_0} \leq C h$. Therefore
\begin{align*}
\abs{II_{21,n} (\tau)} &\leq \frac{1}{nh} \sum_{t=1}^{n} k_t(\tau_0 + \tau_1 h) k_t(\tau_0 + \tau_2 h) \abs{g\left(\frac{t}{n}, \frac{t}{n}\right) - g(\tau_0, \tau_0)} \\
&\leq \frac{C_g}{nh} \sum_{t=1}^{n} k_t(\tau_0 + \tau_1 h) k_t(\tau_0 + \tau_2 h) \abs{t/n - \tau_0} \leq C_g C_2 h.
\end{align*}
For $II_{22,n} (\tau)$, we again apply the integral approximation bound \eqref{eq:int_approx} to find
\begin{align*}
\abs{II_{22,n} (\tau)} &\leq \frac{1}{h} g(\tau, \tau)\abs{\frac{1}{n} \sum_{t=1}^{n} k_t(\tau_0 + \tau_1 h) k_t(\tau_0 + \tau_2 h) -  \int_{0}^{1} K\left(\frac{x - \tau_0 - \tau_1 h}{h}\right) K\left(\frac{x - \tau_0 - \tau_2 h}{h}\right) \d x}\\
&\quad +\frac{1}{nh} g(\tau, \tau) \sum_{t=n-i}^{n} k_t(\tau_0 + \tau_1 h) k_t(\tau_0 + \tau_2 h)\\
&\leq \frac{C_g}{h} \sup_{\abs{x-y}\leq \frac{1}{n}} \abs{K\left(\frac{x - \tau_0 - \tau_1 h}{h}\right) K\left(\frac{x - \tau_0 - \tau_2 h}{h}\right) - K\left(\frac{y - \tau_0 - \tau_1 h}{h}\right) K\left(\frac{y - \tau_0 - \tau_2 h}{h}\right)}\\
&\quad +\frac{C_g }{nh} \sum_{t=n-i}^{n} k_t(\tau_0 + \tau_1 h) k_t(\tau_0 + \tau_2 h) \leq \frac{C_g C_3}{nh^2} + C_g S_{n}(i),
\end{align*}
where $S_{n}(i) = \frac{1}{nh} \suptau \sum_{t=n-i}^{n} k_t(\tau)^2$. Now take a sequence $M = M(n)$ such that $M \rightarrow \infty$ as $n \rightarrow \infty$ and $M/n \rightarrow 0$. Then take an $N_1$ such that for all $n > N_1$, $M/n \leq \delta - \epsilon$ for some $\epsilon > 0$, such that $k_t (\tau) = 0$ for all $t \leq M$ and $\tau \in [\delta, 1- \delta]$. Then for all $n > N_1$,
\begin{align*}
\sum_{i=1}^{n-1} \beta_i S_{n} (i) \leq \frac{1}{nh} \sum_{i=M+1}^{\infty} \beta_i \suptau \sum_{t=1}^n k_t(\tau)^2 \leq C \sum_{i=M+1}^{\infty} \beta_i.
\end{align*}
As $M \rightarrow \infty$ and $\sum_{i=1}^\infty \beta_i < \infty$, $\phi_n = C \sum_{i=M+1}^{\infty} \beta_i = o(1)$.

Finally, for $II_{23,n}(\tau)$ perform a change of variables with $u=(x-\tau_0)/h$ such that
\begin{align*}
\frac{1}{h}\int_{0}^{1} K\left(\frac{x - \tau_0 - \tau_1 h}{h}\right) K\left(\frac{x - \tau_0 - \tau_2 h}{h}\right) \d x = \int_{-\tau_0/h}^{(1-\tau_0)/h} K\left(u - \tau_1 \right) K\left(u - \tau_2 \right) \d u
\end{align*}
Take an $N_2$ such that for all $n>N_2$, $\min\{\tau_0/h,(1-\tau_0)/h\}$ will be larger than the bounds of the compact support of $K(\cdot)$ such that $\int_{-\tau_0/h}^{(1-\tau_0)/h} K\left(u - \tau_1 \right) K\left(u - \tau_2 \right) \d u = \int_{-\infty}^{\infty} K\left(u - \tau_1 \right) K\left(u - \tau_2 \right) \d u$ for all $n>N_2$. Finally, a second change of variables with $\omega = u - \tau_1$ shows that $\int_{-\infty}^{\infty} K\left(u - \tau_1 \right) K\left(u - \tau_2 \right) \d u = \kappa_2 (\tau_1 - \tau_2)$ and thus $II_{23,n}(\tau) = 0$ for $n > N_2$. As no bounds depend on $\tau_0$, $\tau_1$ or $\tau_2$, all results hold uniformly for $n > N = \max\{N_1, N_2\}$.
\end{proof}

\begin{proof}[{\bf Proof of Lemma \ref{lem:D_results}}]
For part $(i)$, write $\bar{f}_n (\tau_1, \tau_2) = f (\tau_1) - l_n(\tau_2)$.
\begin{align*}
&\E Z_{n,D,f} (\tau)^2 = \frac{1}{nh} \sum_{s=1}^{n} \sum_{t=1}^{n} \bar{f}_n \left(\frac{s}{n}, \tau \right) \bar{f}_n \left(\frac{t}{n}, \tau \right) k_s(\tau) k_t(\tau) \cov(D_s, D_t)\\
& = \frac{1}{nh} \sum_{i=-n+1}^{n-1} \sum_{t=1}^{n-\abs{i}} k_t(\tau) k_{t+\abs{i}}(\tau) \bar{f}\left(\frac{t}{n}, \tau \right) \bar{f}\left(\frac{t+i}{n}, \tau \right) R_{D,\abs{i}} \left(\frac{t}{n}, \frac{t+i}{n} \right)\\
&= \bar{f}_n(\tau, \tau) \Omega_D(\tau) \kappa_2 + A_{n} (\tau),
\end{align*}
where
\begin{align*}
A_n(\tau) &= 2 \kappa_2 \bar{f}_n(\tau, \tau) \sum_{i=n}^\infty R_{D,i} (\tau, \tau) + \sum_{i=-n+1}^{n-1} \left[\frac{1}{nh} \sum_{t=1}^{n-\abs{i}} k_t(\tau) k_{t+\abs{i}}(\tau) \bar{f}_n\left(\frac{t}{n}, \tau \right) \bar{f}_n\left(\frac{t+\abs{i}}{n}, \tau \right) \right.\\
&\quad \times \left. R_{D,\abs{i}} \left(\frac{t}{n}, \frac{t+\abs{i}}{n} \right) - \bar{f}_n(\tau, \tau)^2 R_{D,\abs{i}} (\tau, \tau) \kappa_2 \right].
\end{align*}
Letting $A_n (t,i,\tau)= k_t(\tau) k_{t+i}(\tau) R_{D,i} \left(\frac{t}{n}, \frac{t+i}{n} \right)$ and $A (i,\tau) = 2 \kappa_2  R_{D,i} \left(\tau, \tau \right)$, we can decompose the inner sum as
\begin{equation} \label{eq:f_q_decomp}
\begin{split}
&\abs{\frac{1}{nh} \sum_{t=1}^{n-i} \bar{f}_n\left(\frac{t}{n}, \tau \right) \bar{f}_n \left(\frac{t+i}{n}, \tau \right) A_n (t,i,\tau) - \bar{f}_n(\tau, \tau)^2 A (i, \tau)} \\
&\leq \abs{\frac{1}{nh}\sum_{t=1}^{n-i} f\left(\frac{t}{n}\right) f\left(\frac{t+i}{n} \right) A_n (t,i,\tau) - f(\tau)^2 A (i,\tau)}\\
&\quad + \abs{l_n(\tau)} \abs{\frac{1}{nh} \sum_{t=1}^{n-i} f\left(\frac{t}{n}\right) A_n (t,i,\tau) - f(\tau) A (i,\tau)}\\
&\quad + \abs{l_n(\tau)} \abs{\frac{1}{nh} \sum_{t=1}^{n-i} f\left(\frac{t+i}{n}\right) A_n (t,i,\tau) - f(\tau) A (i,\tau)} + l_n(\tau)^2 \abs{\frac{1}{nh} \sum_{t=1}^{n-i} A_n (t,i,\tau) - A (i,\tau)}
\end{split}	
\end{equation}
By applying Lemma \ref{lem:kernel_function_limits}$(ii)$ to each of these terms, with $g_i \left(\tau_1, \tau_2 \right) = f(\tau_1) f (\tau_2) R_{D,i} \left(\tau_1, \tau_2\right)$ for the first term and defined analogously for the rest, we find that
\begin{align*}
\sup_{\tau \in [\delta,1-\delta]} \abs{A_n (\tau)} &\leq C_1 \sum_{i=1}^\infty \sup_{\tau \in [0,1]} \abs{R_{D,i} (\tau,\tau)} + \frac{2}{nh} C_1 \sum_{i=1}^{n-1} i C_g (i) + C_2 \left\{h, \frac{1}{nh^2} \right\} \sum_{i=1}^{n-1} C_g (i)\\
&\quad + \sum_{i=1}^{n-1} C_g(i) S_n(i) \leq C \max \left\{h, \frac{1}{nh^2} \right\} + \phi_n,
\end{align*}
with $C_{g}(i) = \sup_{(\tau_1, \tau_2) \in [0,1]^2} \abs{R_{D,i}(\tau_1, \tau_2)} \max\{\sup_{\tau \in[0,1]} f(\tau)^2, 1\}$, where \begin{align*}
&\sum_{i=0}^\infty i C_g(i) \leq C \sum_{i=0}^\infty i \sup_{(\tau_1, \tau_2)} \abs{R_{D,i} (\tau_1, \tau_2)} < \infty,\\
&\lim_{n \rightarrow \infty} \phi_n = C_1 \lim_{n \rightarrow \infty} \sum_{i=n}^\infty \sup_{\tau \in [0,1]} \abs{R_{D,i} (\tau,\tau)} + \lim_{n \rightarrow \infty} \sum_{i=1}^{n-1} C_g(i) S_n(i) = 0.
\end{align*}

Part $(ii)$ follows directly by noting that
\begin{equation*}
\begin{split}
&\frac{1}{nh} \sum_{t=1}^{n} f \left(\frac{t}{n} \right) k_t(\tau) D_t - p(\tau) f(\tau) - h^2 p^{(2)} (\tau) \mu_2 = \left[\frac{1}{nh} \sum_{t=1}^{n} f\left(\frac{t}{n} \right) k_t(\tau) \left[D_t - p\left(\frac{t}{n} \right)\right] \right]\\
&\quad + \left[\frac{1}{nh} \sum_{t=1}^{n} f\left(\frac{t}{n} \right) p\left(\frac{t}{n} \right) k_t(\tau) - p(\tau) f(\tau) - h^2 [fp]^{(2)} (\tau) \mu_2 \right] = Z_{n,D,f,0} (\tau) + R_{n,D} (\tau),
\end{split}
\end{equation*}
as it follows directly from Lemma \ref{lem:kernel_function_limits}$(i)$ that $\suptau \abs{R_{n,D} (\tau)} < C \max \left\{h^3, \frac{1}{nh}\right\}$.

For part $(iii)$, we can write
\begin{equation*}
\begin{split}
&\frac{1}{nh} \sum_{t=1}^{n-i} \tilde g_i\left(\frac{t}{n}, \frac{t+i}{n} \right) k_t(\tau_0 + \tau_1 h) k_{t+i}(\tau_0 + \tau_2 h) D_t D_{t+i} - p(\tau_0) \tilde g_i(\tau_0, \tau_0) \kappa (\tau_1 - \tau_2)\\
&= \left[\frac{1}{nh} \sum_{t=1}^{n-i} \tilde g_i\left(\frac{t}{n}, \frac{t+i}{n} \right) k_t(\tau_0 + \tau_1 h) k_{t+i}(\tau_0 + \tau_2 h) \left[D_t D_{t+i} - \E D_t D_{t+i} \right] \right]\\
&\quad + \left[\frac{1}{nh} \sum_{t=1}^{n-i} \tilde g_i \left(\frac{t}{n}, \frac{t+i}{n} \right) k_t(\tau_0 + \tau_1 h) k_{t+i}(\tau_0 + \tau_2 h) \E D_t D_{t+i} - p(\tau_0) \tilde g_i(\tau_0, \tau_0) \kappa (\tau_1 - \tau_2) \right]\\
& = III_{n,1} (\tau) + III_{n,2} (\tau).
\end{split}
\end{equation*}
Define $X_t (i) = k_{t-i}(\tau_0 + \tau_1 h) k_t(\tau_0 + \tau_2 h) \tilde g_i\left(\frac{t-i}{n}, \frac{t}{n} \right) \left[D_{t-i} D_{t} - \E(D_{t-i} D_{t}) \right]$; it then follows from Assumption \ref{as:MD} that $X_{t}(i)$ is an $L_2$- mixingale with $\norm{\E_{t-j} X_{t}(i)} \leq k_{t-i}(\tau_0 + \tau_1 h) k_{t}(\tau_0 + \tau_2 h) C_g(i) \zeta_j$. Then, by Theorem 1.6 of \citet{McLeish} and Lemma \ref{lem:kernel_function_limits},
\begin{align*}
\E \abs{III_{n,1} (\tau)} &\leq \frac{1}{nh} \E \abs{\sum_{t=i+1}^{n} X_{t}(i)} \leq \frac{C}{nh} \left(\sum_{t=1}^{n-i} k_{t-i}(\tau_0 + \tau_1 h)^2 k_{t}(\tau_0 + \tau_2 h)^2 C_g(i)^2 \right)^{1/2} \\
&\leq \frac{C}{nh} \left( C_g(i)^2\sum_{t=i+1}^{n} \suptau k_{t-i}(\tau)^2 k_{t}(\tau)^2 \right)^{1/2} \leq \frac{C_1}{\sqrt{nh}} C_g(i).
\end{align*}
Finally, as $\E D_t D_{t+i} = p\left(\frac{t}{n} \right) p\left(\frac{t+i}{n} \right) + R_{D,i} \left(\frac{t}{n}, \frac{t+i}{n} \right)$, it follows directly from Lemma \ref{lem:kernel_function_limits}$(ii)$ that 
\begin{equation*}
\suptau \abs{III_{n,2}(\tau)} \leq C_g(i) \left[\frac{C_1}{nh} + C_2 \max \left\{h, \frac{1}{nh^2} \right\} + S_n(i)\right],
\end{equation*}
by taking $g_i (\tau_1, \tau_2) = \tilde g_i (\tau_1, \tau_2) \left[p\left(\frac{t}{n} \right) p\left(\frac{t+i}{n} \right) + R_{D,i} \left(\frac{t}{n}, \frac{t+i}{n} \right)\right]$.
\end{proof}

\begin{proof}[{\bf Proof of Lemma \ref{lem:lrv}}]
The proof is an extension of the proof of Lemma 5 in \citet{Jansson02}. As $\E u_t u_{t+i} = \sigma_\varepsilon^2 \sum_{j} \psi_j \psi_{j+i}$, we can write
\begin{equation*}
\begin{split}
A_{n} (\tau) &= \frac{1}{nh} \sum_{t=1}^{n-i} k_t(\tau) k_{t+i}(\tau) \sigma_t \sigma_{t+i} D_t D_{t+i} \left[u_t u_{t+i} - \E u_t u_{t+i} \right] \\
&= \sigma_\varepsilon^2 \sum_{j=0}^\infty \psi_j \psi_{j+i} \left[\frac{1}{nh} \sum_{t=1}^{n-i}k_t(\tau) k_{t+i} D_t D_{t+i} (\tau) \sigma_t \sigma_{t+i} (\varepsilon_{t-j}^2 - 1) \right]\\
&\quad + \sigma_\varepsilon^2 \sum_{j=0}^\infty \sum_{m=0}^\infty \mathbbm{1}(m \neq j+i) \psi_j \psi_m \frac{1}{nh} \sum_{t=1}^{n-i} k_t(\tau) k_{t+i}(\tau) D_t D_{t+i} \sigma_t \sigma_{t+i} \varepsilon_{t-j} \varepsilon_{t+i-m}.
\end{split}
\end{equation*}
Then
\begin{equation*}
\begin{split}
\E \abs{A_{n} (\tau)} &\leq \sigma_\varepsilon^2 \sup_{j \geq 0} \max_{0 \leq i \leq n-1} \E \abs{ \frac{1}{nh} \sum_{t=1}^{n-i} D_t D_{t+i} k_t(\tau) k_{t+i}(\tau) \sigma_t \sigma_{t+i} (\varepsilon_{t-j}^2 - 1) } \left(\sum_{j=0}^\infty \abs{\psi_j} \abs{\psi_{j+i}} \right) \\
&\quad + \sigma_\varepsilon^2 \frac{1}{\sqrt{nh}} \sup_{j,m \geq 0} \max_{0 \leq i \leq n-1} \E \abs{\frac{1}{\sqrt{nh}} \mathbbm{1}(m \neq j+i) \sum_{t=1}^{n-i} D_t D_{t+i} k_t(\tau) k_{t+i}(\tau) \sigma_t \sigma_{t+i} \varepsilon_{t-j} \varepsilon_{t+i-m}}\\
&\quad \times \sum_{j=0}^\infty \sum_{m=0}^\infty \abs{\psi_j} \abs{\psi_m} = \phi_{n} (\tau) \beta_i + \frac{1}{\sqrt{nh}} \eta_{n} (\tau),
\end{split}
\end{equation*}
where 
\begin{align*}
\phi_n(\tau) &= \sigma_\varepsilon^2 \sup_{j \geq 0} \max_{0 \leq i \leq n-1} \E \abs{\frac{1}{nh} \sum_{t=1}^{n-i} D_t D_{t+i} k_t(\tau) k_{t+i}(\tau) \sigma_t \sigma_{t+i} (\varepsilon_{t-j}^2 - 1)}, \\
\beta_i &= \sum_{j=0}^\infty \abs{\psi_j} \abs{\psi_{j+i}},\\
\eta_n(\tau) &= \sigma_\varepsilon^2\sup_{j,m \geq 0} \max_{0 \leq i \leq n-1} \E \abs{\frac{1}{\sqrt{nh}} \mathbbm{1}(m \neq j+i) \sum_{t=1}^{n-i} D_t D_{t+i} k_t(\tau) k_{t+i}(\tau) \sigma_t \sigma_{t+i} \varepsilon_{t-j} \varepsilon_{t+i-m}}\\
&\quad \times \sum_{j=0}^\infty \sum_{m=0}^\infty \abs{\psi_j} \abs{\psi_m}
\end{align*}
First note that
\begin{equation*}
\sum_{i=0}^{\infty} \beta_i \leq \sum_{i=0}^\infty \sum_{j=0}^\infty \abs{\psi_j} \abs{\psi_{j+i}} \leq \left( \sum_{j=0}^\infty\abs{\psi_j} \right)^2 < \infty.
\end{equation*}
Next, let $w_t = \varepsilon_t^2 -1$. By Assumption \ref{as:LP}, $\E \abs{w_t}^p \leq \E \abs{\varepsilon_t}^{2 p} < \infty$ for some $p>1$,  which implies that $w_t$ is uniformly integrable and for every $\epsilon > 0$, there exists a $\lambda_{\epsilon} > 0$ such that $\sup_{t} \E \abs{w_t} 1(\abs{w_t} > \lambda_{\epsilon}) < \epsilon$. Then define $w_{1,t} = w_t 1(\abs{w_t} \leq \lambda_{\epsilon})$ and $w_{2,t} = w_t - w_{1,t} = w_t 1(\abs{w_t} > \lambda_{\epsilon})$.

As in \citet[Proof of Theorem 2.22]{HH}, by the Marcinkiewicz-Zygmund inequality we have that
\begin{equation*}
\begin{split}
&\E \abs{\frac{1}{nh} \sum_{t=1}^{n-i} D_t D_{t+i} k_t k_{t+i} \sigma_t \sigma_{t+i} w_t} \leq C \frac{1}{nh} \left(\sum_{t=1}^{n-i} k_t^2(\tau) k_{t+i}^2(\tau) \sigma_t^2 \sigma_{t+i}^2 \E w_{t-j}^2 \right)^{1/2}\\
&\quad \leq C_1 \left[\frac{1}{nh}\left(\sum_{t=1}^{n-i} k_t^2 k_{t+i}^2 \sigma_t^2 \sigma_{t+i}^2 \E w_{1,t-j}^2 \right)^{1/2} +  \frac{1}{nh} \left(\sum_{t=1}^{n-i} k_t^2(\tau) k_{t+i}^2(\tau) \sigma_t^2 \sigma_{t+i}^2 \E w_{2,t-j}^2 \right)^{1/2} \right]\\
&= C_1 (\phi_{n,1}(\tau) + \phi_{n,2}(\tau)).
\end{split}
\end{equation*}
As $\abs{w_{1,t-j}} \leq \lambda_\epsilon$, $\sup_{\tau \in[0,1]} \sigma(\tau)^2 < \infty$, it follows from the stationarity of $w_{1,t}$ and Lemma \ref{lem:kernel_sum} that
\begin{equation*}
\begin{split}
\suptau \phi_{n,1} (\tau) \leq \frac{1}{nh} \left(\sup_{\tau \in [0,1]} \sigma(\tau)^4 \E w_{1,t}^2 \suptau \sum_{t=1}^{n-i} k_t^2(\tau) k_{t+i}^2(\tau)\right)^{1/2} \leq \frac{C}{\sqrt{nh}} \lambda_{\epsilon}.
\end{split}
\end{equation*}
Furthermore we get that
\begin{equation*}
\begin{split}
\suptau \phi_{n,2}(\tau) &\leq \frac{1}{nh} \sup_{\tau \in [0,1]} \sigma(\tau)^2 \suptau \sum_{t=1}^{n-i} k_t(\tau) k_{t+i}(\tau) \E \abs{w_{2,t-j}} \leq C_3 \epsilon,
\end{split}
\end{equation*}
from which it follows that $\suptau \phi_n (\tau) \leq C[ (n h)^{-1/2} \lambda_\epsilon + \epsilon] = o(1) + C \epsilon$. As we can make $\epsilon$ arbitrarily small, it follows that $\phi_n = \suptau \phi_n(\tau) \rightarrow 0$ as $n \rightarrow \infty$ for all $\tau \in (0,1)$.

Next we look at $\eta_n (\tau)$. First note that $\sum_{j=0}^\infty \sum_{m=0}^\infty \abs{\psi_j} \abs{\psi_m} = \left( \sum_{j=0}^\infty\abs{\psi_j} \right)^2 < \infty$. Next, note that by Jensen's inequality
\begin{equation*}
\begin{split}
&\E \abs{\frac{1}{\sqrt{nh}} \mathbbm{1}(m \neq j+i) \sum_{t=1}^{n-i} D_t D_{t+i} k_t(\tau) k_{t+i}(\tau) \sigma_t \sigma_{t+i} \varepsilon_{t-j} \varepsilon_{t+i-m}}\\
&\quad \leq \frac{1}{\sqrt{nh}} \mathbbm{1}_{m \neq j+i} \left[\E \left( \sum_{t=1}^{n-i} D_t D_{t+i} k_t(\tau) k_{t+i}(\tau) \sigma_t \sigma_{t+i} \varepsilon_{t-j} \varepsilon_{t+i-m} \right)^2 \right]^{1/2}\\
&\quad \leq \frac{1}{\sqrt{nh}} \mathbbm{1}_{m \neq j+i} \left[\sum_{s=1}^{n-i} \sum_{t=1}^{n-i} k_s(\tau) k_t(\tau) k_{s+i}(\tau) k_{t+i}(\tau) \sigma_s \sigma_t \sigma_{s+i} \sigma_{t+i} \E \varepsilon_{s-j}\varepsilon_{t-j} \varepsilon_{s+i-m} \varepsilon_{t+i-m} \right]^{1/2}.\\
\end{split}
\end{equation*}
As $\mathbbm{1}(m \neq j+i) \E \varepsilon_{s-j}\varepsilon_{t-j} \varepsilon_{s+i-m} \varepsilon_{t+i-m}$ can only be non-zero if $t=s$, we can deduce that 
\begin{equation*}
\begin{split}
&\suptau \E \abs{\frac{1}{\sqrt{nh}} \mathbbm{1}(m \neq j+i) \sum_{t=1}^{n-i} D_t D_{t+i} k_t(\tau) k_{t+i}(\tau) \sigma_t \sigma_{t+i} \varepsilon_{t-j} \varepsilon_{t+i-m}}\\
&\quad \leq \frac{1}{\sqrt{nh}} \left[\suptau \sum_{t=1}^{n-i} k_t^2(\tau) k_{t+i}^2(\tau) \sigma_t^2 \sigma_{t+i}^2 \right]^{1/2} \leq \frac{1}{\sqrt{nh}} \left[C nh \sup_{\tau\in[0,1]} \sigma(\tau)^4 \right]^{1/2} \leq C.
\end{split}
\end{equation*}
Let $\eta_n = \suptau \eta_n (\tau)$, then the above shows that $\limsup_{n \rightarrow \infty} \eta_n < \infty$.
\end{proof}

\begin{proof}[{\bf Proof of Lemma \ref{lem:SE}}]
Define the ``difference'' operator $\Delta_{\tau_1, \tau_2}$ applied to any function $f(\cdot)$ as
\begin{equation} \label{eq:Delta}
\Delta_{\tau_1, \tau_2} f(\tau_0) = f(\tau_0 + \tau_1 h) - f(\tau_0 + \tau_2 h).
\end{equation}

Let $A_n (\tau_0, \tau_1, \tau_2) = \frac{1}{\sqrt{nh}} \sum_{t=1}^{n} \Delta_{\tau_1, \tau_2} k_t(\tau_0) \left[f(t/n) - q_n (\tau)\right] X_t$. By the Markov inequality, for any $\epsilon > 0$
\begin{equation} \label{eq:B_n}
\P \left[\frac{\abs{A_n (\tau_0, \tau_1, \tau_2)}}{ \abs{\tau_1 - \tau_2}} \leq \epsilon  \right] \leq \epsilon^{-2} \frac{\E A_n (\tau_0, \tau_1, \tau_2)^2}{(\tau_1 - \tau_2)^2 }.
\end{equation}
We will now show that, for large enough $n$, \eqref{eq:B_n} can be bounded by a constant, which is sufficient for the result to hold.

Note that
\begin{align*}
\E A_n (\tau_0, \tau_1, \tau_2)^2 &= \frac{1}{nh} \sum_{i=-n+1}^{n-1}\sum_{t=1}^n \Delta_{\tau_1, \tau_2} k_t (\tau_0) \Delta_{\tau_1, \tau_2} k_{t+i} (\tau_0) \E X_t X_{t+\abs{i}}\\
&\leq \frac{1}{nh} \sum_{i=-n+1}^{n-1} \sup_{1 \leq t \leq n - i} \abs{\E X_t X_{t+\abs{i}}} \sum_{t=1}^n \abs{\Delta_{\tau_1, \tau_2} k_t (\tau_0)} \abs{\Delta_{\tau_1, \tau_2} k_{t+i} (\tau_0) }.
\end{align*}
By the compact support of the kernel, there a $C_K >0 $ such that $K(x) = 0$ for all $\abs{x} > C_K$. Without loss of generality we henceforth assume that $C_K=1$. Then note that $k_t (\tau_0 + \tau_1 h) > 0$ only if $\abs{t/n - \tau_0 - \tau_1 h} < h$. As $\{t: \abs{t/n - \tau_0 - \tau_1 h} < h\} \subseteq \{t: \abs{t/n - \tau_0} < 2h\}$, this implies that $\sum_{t=1}^n k_t (\tau_0 + \tau_1 h) = \sum_{t=n(\tau_0-2h)}^{n(\tau_0+2h)} k_t (\tau_0 + \tau_1 h)$. Furthermore, the Lipschitz property of the kernel implies that $\abs{\Delta_{\tau_1, \tau_2} k_t (\tau_0)} \leq C \abs{\tau_1 - \tau_2}$. We then find that
\begin{align*}
\E A_n (\tau_0, \tau_1, \tau_2)^2 &\leq \sum_{i=-n+1}^{n-1} \sup_{1 \leq t \leq n - i} \abs{\E X_t X_{t+\abs{i}}} \frac{1}{nh} \sum_{t=n(\tau_0-2h)}^{n(\tau_0+2h)} C (\tau_1 - \tau_2)^2\\
&\leq C_1 (\tau_1 - \tau_2)^2 \sum_{i=-n+1}^{n-1} \sup_{1 \leq t \leq n - i} \abs{\E X_t X_{t+\abs{i}}} \leq C_2 (\tau_1 - \tau_2)^2,
\end{align*}
such that \eqref{eq:B_n} can be bounded by a constant not depending on $\tau_1$ or $\tau_2$, which concludes the proof.
\end{proof}

\begin{proof}[{\bf Proof of Lemma \ref{lem:SEb}}]
Let $A_n^* (\tau_0, \tau_1, \tau_2) = \frac{1}{\sqrt{nh}} \sum_{t=1}^{n} \Delta_{\tau_1, \tau_2} k_t(\tau_0) \left[f(t/n) - q_n (\tau)\right] X_t^*$, with $\Delta_{\tau_1, \tau_2}$ defined in \eqref{eq:Delta}. Apply the Markov inequality twice, such that for any $\epsilon, \eta > 0$,
\begin{equation} \label{eq:B_nb}
\P\left\{ \P^* \left[\frac{\abs{A_n^* (\tau_0, \tau_1, \tau_2)}}{\abs{f(\tau) - q_n(\tau)} \abs{\tau_1 - \tau_2}} > \epsilon  \right] > \eta \right\} \leq \epsilon^{-2} \eta^{-1} \frac{\E \E^* A_n^* (\tau_0, \tau_1, \tau_2)^2}{(\tau_1 - \tau_2)^2 (f(\tau) - q_n(\tau))^2}.
\end{equation}
We again show that, for large enough $n$, that \eqref{eq:B_nb} can be bounded by a constant, which is sufficient for the result to hold.

By the compact support and the Lipschitz property of the kernel, we then find that
\begin{align*}
\E \E^* A_n^* (\tau_0, \tau_1, \tau_2)^2 &= \frac{1}{nh} \sum_{i=-n+1}^{n-1}\sum_{t=1}^n \Delta_{\tau_1, \tau_2} k_t (\tau_0) \Delta_{\tau_1, \tau_2} k_{t+i} (\tau_0) \E \E^* X_t^* X_{t+\abs{i}}^*\\
&\leq \frac{1}{nh} \sum_{i=-n+1}^{n-1} \sup_{1 \leq t \leq n - i} \abs{\E \E^* X_t^* X_{t+\abs{i}}^*} \dsum{t} \abs{\Delta_{\tau_1, \tau_2} k_t (\tau_0)} \abs{\Delta_{\tau_1, \tau_2} k_{t+i} (\tau_0) }\\
&\leq C (\tau_1 - \tau_2)^2 \sum_{i=-n+1}^{n-1} \sup_{1 \leq t \leq n - i} \abs{\E \E^* X_t^* X_{t+\abs{i}}^*} \leq C_1 (\tau_1 - \tau_2)^2.\qedhere
\end{align*}
\end{proof}

\begin{proof}[{\bf Proof of Lemma \ref{lem:R_unif}}]
Using the terms defined in \eqref{eq:m_decomp}, we define
\begin{equation*}
R_{n,1} (\tau) = I_n (\tau) + II_n (\tau),
\end{equation*}
such that $R_n(\tau) = R_{n,1}(\tau) + III_n(\tau)$. With $\Delta_{\tau_1, \tau_2}$ as defined in \eqref{eq:Delta}, it follows directly from the proof of Lemma \ref{lem:m_decomp} that
\begin{equation*}
\suptautt \abs{\Delta_{\tau_1, \tau_2} III_{n} (\tau_0)} \leq 2\suptau \abs{III_{n} (\tau)} = o(1).
\end{equation*}

For $R_{n,1} (\tau)$, as Lemmas \ref{lem:m_decomp} establishes pointwise convergence, we only need to establish stochastic equicontinuity to have uniform convergence as well. For this purpose we show that there exists some $N>0$ such that for all $n>N$, there exists some random variable $B_n = O_p(1)$ such that
\begin{equation} \label{eq:se}
\abs{\Delta_{\tau_1, \tau_2} R_{n,1}(\tau_0)} \leq B_n \abs{\tau_1 - \tau_2}, \qquad \forall \tau_1, \tau_2 \in [\delta, 1 - \delta],
\end{equation}
which, as shown by \citet[Lemma 1]{Andrews} and \citet[Thm 21.10]{Davidson}, implies stochastic equicontinuity. To establish \eqref{eq:se}, we establish this bound for each term in $R_{n,1} (\tau)$. In the following, we use generic random variables $B_{n}^\prime$ and $B_{n}^{\prime\prime}$, which can change each line, but are always $O_p(1)$ and do not depend on $\tau_1$ and $\tau_2$. First, for $I_n(\tau) = \left[\hat{p}(\tau)^{-1} - p(\tau)^{-1} \right] Z_{n,U} (\tau)$, a straightforward but tedious calculation shows that
\begin{align*}
\abs{\Delta_{\tau_1, \tau_2} I_{n}(\tau_0 )} &\leq B_n^\prime \abs{\Delta_{\tau_1, \tau_2} Z_{n,U} (\tau_0)} + B_n^{\prime \prime} \abs{\Delta_{\tau_1, \tau_2} [\hat{p}(\tau_0) - p(\tau_0)]}.
\end{align*}
By taking $X_t = D_t z_t$, Lemmas \ref{lem:cov} and \ref{lem:SE} directly imply that $\abs{\Delta_{\tau_1, \tau_2} Z_{n,U} (\tau_0)} \leq B_n^\prime \abs{\tau_1 - \tau_2}$. Furthermore, for the second term we can write
\begin{align*}
\abs{\Delta_{\tau_1, \tau_2} [\hat{p}(\tau_0) - p(\tau_0)]} &\leq \abs{\Delta_{\tau_1, \tau_2} [\hat{p}(\tau_0) - \overline{p}(\tau_0)]} + \abs{\Delta_{\tau_1, \tau_2} \left[\overline{p}(\tau_0) - p(\tau_0) - h^2 \frac{\mu_2}{2} p^{(2)} (\tau_0) \right]}\\
&\quad + h^2 \frac{\mu_2}{2} \abs{\Delta_{\tau_1, \tau_2} p^{(2)} (\tau_0) },
\end{align*}
where $\overline{p} (\tau) = \frac{1}{nh} \sum_{t=1}^n k_t (\tau) p\left(\frac{t}{n} \right)$. For the first part we can use Lemma \ref{lem:SE} with $X_t = D_t - p\left(\frac{t}{n} \right)$ to establish the bound $\abs{\Delta_{\tau_1, \tau_2} [\hat{p}(\tau_0) - \overline{p}(\tau_0)]} \leq B_n^\prime \abs{\tau_1 - \tau_2} /\sqrt{nh}$, while for the second part we have uniform convergence by Lemma \ref{lem:kernel_function_limits}(i), which implies stochastic equicontinuity, as
\begin{align*}
\abs{\Delta_{\tau_1, \tau_2} \left[\overline{p}(\tau_0) - p(\tau_0) - h^2 \frac{\mu_2}{2} p^{(2)} (\tau_0) \right]} &\leq 2 \suptau \abs{\overline{p}(\tau_0) - p(\tau_0) - h^2 \frac{\mu_2}{2} p^{(2)} (\tau_0) } \\
&\leq C \max \left\{h^3, \frac{1}{nh} \right\},
\end{align*}
and for any $\tau_1, \tau_2$, we can find an $n$ such that $\max\left\{h^3, \frac{1}{nh} \right\}$ is smaller than $\abs{\tau_1 - \tau_2} / \sqrt{nh}$. Finally, for the third term we get, by the smoothness of $p^{(2)} (\cdot)$, that $\abs{\Delta_{\tau_1, \tau_2} p^{(2)} (\tau_0) } \leq C h \abs{\tau_1 - \tau_2}$. Combining these three results, it therefore follows that
\begin{equation} \label{eq:p_unif}
\abs{\Delta_{\tau_1, \tau_2} [\hat{p}(\tau_0) - p(\tau_0)]} \leq B_n^\prime \abs{\tau_1 - \tau_2} / \sqrt{nh}.
\end{equation}

Similarly, for $II_{n,1} (\tau)$, Lemma \ref{lem:SE} with $X_t = \left[m(\tau_0) - p^{-1} \overline{m}_p (\tau_0) \right] \left[D_t - p\left(\frac{t}{n}\right) \right]$ provides the appropriate bound, as $\abs{\E X_t X_{t+i}} \leq C \cov(D_t, D_{t+i})$. For $II_{n,2} (\tau)$, we have that
\begin{equation*}
\abs{ \Delta_{\tau_1, \tau_2} II_{n,2} (\tau_0)} \leq B_n^\prime \sqrt{nh} \Delta_{\tau_1, \tau_2} [\hat{p} (\tau_0) - p(\tau_0)]^2 + B_n^{\prime \prime} \suptaut [\hat{p} (\tau) - p(\tau)]^2 \Delta_{\tau_1, \tau_2} \overline{m}_p (\tau_0).
\end{equation*}
By \eqref{eq:p_unif}, we have that $\sqrt{nh} \Delta_{\tau_1, \tau_2} [\hat{p} (\tau_0) - p(\tau_0)]^2 \leq B_n^\prime \abs{\tau_1 - \tau_2}$, while, again by \eqref{eq:p_unif} and pointwise convergence, $\sqrt{nh} \suptaut [\hat{p} (\tau) - p(\tau)]^2 = O_p \left( \max\left\{\sqrt{n h^9}, \frac{1}{\sqrt{nh}} \right\} \right)$, while it follows from the Lipschitz properties that $\Delta_{\tau_1, \tau_2} \overline{m}_p (\tau_0) \leq C \abs{\tau_1 - \tau_2}$. Finally, the same arguments used above can also be used to establish that $\abs{\Delta_{\tau_1, \tau_2} II_{n,3} (\tau_0)} \leq B_n^\prime \abs{\tau_1 - \tau_2}$.

We treat the bootstrap terms in exactly the same way using the decomposition in \eqref{eq:mb_decomp}. Again, uniform convergence of $III_n^*(\tau)$ is immediate from the proof of Lemma \ref{lem:mb_decomp}. For the other terms $R_{n,1}^* (\tau) = R_n^*(\tau) - III_n^* (\tau)$, we again establish stochastic equicontinuity, although now conditionally, that is we show that
\begin{equation} \label{eq:seb}
\abs{\Delta_{\tau_1, \tau_2} R_{n,1}^*(\tau_0)} \leq B_n^* \abs{\tau_1 - \tau_2}, \qquad \forall \tau_1, \tau_2 \in [\delta, 1 - \delta],
\end{equation}
where $B_n^* = O_p^*(1)$.

As for $I_n(\tau)$, we can bound $I_{n,1}^*(\tau)$ as
\begin{align*}
\abs{\Delta_{\tau_1, \tau_2} I_{n,1}^*(\tau_0)} &\leq B_n^{\prime} \abs{\Delta_{\tau_1, \tau_2} Z_{n,U} (\tau_0)} + B_n^{*\prime \prime} \abs{\Delta_{\tau_1, \tau_2} [\hat{p}(\tau_0) - p(\tau_0)]}.
\end{align*}
The second part can be bounded using \eqref{eq:p_unif}, while for the first, we apply Lemma \ref{lem:SEb} with $X_t^* = D_t z_t \xi_t^*$. First we show the summability of the bootstrap covariances:
\begin{align*}
&\sum_{i=0}^{n-1} i \sup_t \abs{\E \E^* D_t z_t \xi_t^* D_{t+i} z_{t+i} \xi_{t+i}^*} \leq \sum_{i=0}^{n-1} i R_U (i) \gamma^i \leq \sum_{i=0}^{n-1} i R_U (i) < \infty.
\end{align*}
The bound on $\Delta_{\tau_1, \tau_2} I_{n,1} (\tau_0)$ now follows directly from Lemma \ref{lem:SEb}.

Similarly, for $I_{n,2}^*(\tau)$, take $X_t^* = \left[m\left(\frac{t}{n} \right) - \tilde{m}\left(\frac{t}{n} \right)\right] D_t \xi_t^*$ and note that, by \eqref{eq:m_norm}
\begin{align*}
&\sum_{i=0}^{n-1} i \sup_t \abs{\E \E^* \left[m\left(\frac{t}{n} \right) - \tilde{m}\left(\frac{t}{n} \right)\right] D_t \xi_t^* \left[m\left(\frac{t+i}{n} \right) - \tilde{m}\left(\frac{t+i}{n} \right)\right] D_{t+i} \xi_{t+i}^*}\\
& \leq \sum_{i=0}^{n-1} i \gamma^i \suptau \norm{m\left(\tau\right) - \tilde{m}\left(\tau \right)}^2 \leq C \ell \max \left\{\tilde{h}^4, \frac{1}{n \tilde{h}} \right\} < C_1,
\end{align*}
such that the bound follows directly from Lemma \ref{lem:SEb}. 

The remaining terms do not contain bootstrap quantities, and can be treated entirely analogously to their non-bootstrap counterparts to show that \eqref{eq:B_nb} holds.
\end{proof}

\clearpage
\section{Additional Simulation Results}
We present here additional simulation results that are part of the specifications mentioned in Section \ref{sec:simulation}. Specifically, we first show the homoskedastic results. Second, we present additional tables for heteroskedastic DGPs that vary the parameters $k$ and $a$. They define the volatility process given in equation \eqref{eq:variance} in the main paper. Third, we have further results with missing data. We only present results for one bandwidth $h=0.06$ in the paper. Here, the other bandwidths are considered. The results in this Appendix are obtained using 1000 Monte Carlo repetitions and $B=599$.

\begin{table}
\centering
\begin{tabular}{@{}lllllllllll@{}}
\toprule
 $h$ & \multicolumn{1}{l}{Case} & \multicolumn{1}{c}{$\gamma=0$} & \multicolumn{2}{c}{$\gamma=0.2$} & \multicolumn{2}{c}{$\gamma=0.4$} & \multicolumn{2}{c}{$\gamma=0.6$} & \multicolumn{1}{c}{$-$} \\ \midrule

                                     &  & WB & AWB & DWB & AWB  & DWB & AWB & DWB & SWB
                            \\ \cmidrule(l){3-3} \cmidrule(l){4-5} \cmidrule(l){6-7} \cmidrule(l){8-9} \cmidrule(l){10-10}
\multirow{ 14}{*}{$0.02$} & $0$ & 0.948 & 0.953 & 0.948 & 0.944 & 0.948 & 0.906 & 0.949 & 0.906 \\
&& (0.519) & (0.512) & (0.522) & (0.474) & (0.523) & (0.403) & (0.522) & (0.500) \\
 & $AR_{0.2}$ & 0.904 & 0.919 & 0.905 & 0.913 & 0.906 & 0.873 & 0.904 & 0.848 \\
&& (0.505) & (0.509) & (0.507) & (0.482) & (0.505) & (0.423) & (0.506) & (0.480) \\
 & $AR_{0.5}$ & 0.790 & 0.826 & 0.789 & 0.832 & 0.787 & 0.797 & 0.786 & 0.675 \\
&& (0.430) & (0.451) & (0.430) & (0.449) & (0.429) & (0.414) & (0.431) & (0.375) \\
\bottomrule
\end{tabular}
\caption{Pointwise coverage probabilities: homoskedasticity}
\label{pointwise0}
\end{table}

\begin{table}
\centering
\begin{tabular}{@{}lllllllllll@{}}
\toprule
 $h$ & \multicolumn{1}{l}{Case} & \multicolumn{1}{c}{$\gamma=0$} & \multicolumn{2}{c}{$\gamma=0.2$} & \multicolumn{2}{c}{$\gamma=0.4$} & \multicolumn{2}{c}{$\gamma=0.6$} & \multicolumn{1}{c}{$-$} \\ \midrule

                                     &  & WB & AWB & DWB & AWB  & DWB & AWB & DWB & SWB
                            \\ \cmidrule(l){3-3} \cmidrule(l){4-5} \cmidrule(l){6-7} \cmidrule(l){8-9} \cmidrule(l){10-10}
\multirow{ 14}{*}{$0.02$} & $0$ & 0.890 & 0.883 & 0.881 & 0.873 & 0.919 & 0.768 & 0.906 & 0.846 \\
&& (0.479) & (0.473) & (0.481) & (0.437) & (0.482) & (0.371) & (0.482) & (0.462) \\
 & $AR_{0.2}$ & 0.814 & 0.844 & 0.829 & 0.817 & 0.832 & 0.739 & 0.845 & 0.748 \\
&& (0.467) & (0.469) & (0.468) & (0.445) & (0.466) & (0.390) & (0.466) & (0.443) \\
 & $AR_{0.5}$ & 0.634 & 0.693 & 0.661 & 0.676 & 0.644 & 0.623 & 0.643 & 0.400 \\
&& (0.397) & (0.416) & (0.397) & (0.414) & (0.396) & (0.381) & (0.398) & (0.346) \\
\bottomrule
\end{tabular}
\caption{Simultaneous coverage probabilities over $G.sub$: homoskedasticity}
\label{simultGsub0}
\end{table}

\begin{table}
\centering
\begin{tabular}{@{}lllllllllll@{}}
\toprule
 $h$ & \multicolumn{1}{l}{Case} & \multicolumn{1}{c}{$\gamma=0$} & \multicolumn{2}{c}{$\gamma=0.2$} & \multicolumn{2}{c}{$\gamma=0.4$} & \multicolumn{2}{c}{$\gamma=0.6$} & \multicolumn{1}{c}{$-$} \\ \midrule

                                     &  & WB & AWB & DWB & AWB  & DWB & AWB & DWB & SWB
                            \\ \cmidrule(l){3-3} \cmidrule(l){4-5} \cmidrule(l){6-7} \cmidrule(l){8-9} \cmidrule(l){10-10}
\multirow{ 14}{*}{$0.02$} & $0$ & 0.880 & 0.883 & 0.885 & 0.868 & 0.901 & 0.756 & 0.895 & 0.807 \\
&& (0.380) & (0.374) & (0.381) & (0.346) & (0.382) & (0.294) & (0.381) & (0.366) \\
 & $AR_{0.2}$ & 0.770 & 0.832 & 0.798 & 0.802 & 0.807 & 0.690 & 0.799 & 0.631 \\
&& (0.369) & (0.372) & (0.370) & (0.353) & (0.369) & (0.310) & (0.369) & (0.351) \\
 & $AR_{0.5}$ & 0.554 & 0.640 & 0.559 & 0.596 & 0.532 & 0.536 & 0.534 & 0.230 \\
&& (0.314) & (0.330) & (0.314) & (0.329) & (0.313) & (0.304) & (0.315) & (0.274) \\
\bottomrule
\end{tabular}
\caption{Simultaneous coverage probabilities over $G$: homoskedasticity}
\label{simultG0}
\end{table}

\clearpage

\begin{table}
\centering
\begin{tabular}{@{}lllllllllll@{}}
\toprule
 $h$ & \multicolumn{1}{l}{Case} & \multicolumn{1}{c}{$\gamma=0$} & \multicolumn{2}{c}{$\gamma=0.2$} & \multicolumn{2}{c}{$\gamma=0.4$} & \multicolumn{2}{c}{$\gamma=0.6$} & \multicolumn{1}{c}{$-$} \\ \midrule

                                     &  & WB & AWB & DWB & AWB  & DWB & AWB & DWB & SWB
                            \\ \cmidrule(l){3-3} \cmidrule(l){4-5} \cmidrule(l){6-7} \cmidrule(l){8-9} \cmidrule(l){10-10}
\multirow{ 14}{*}{$0.02$} & $0$ & 0.946 & 0.952 & 0.946 & 0.944 & 0.945 & 0.906 & 0.947 & 0.926 \\
&& (0.727) & (0.711) & (0.729) & (0.653) & (0.729) & (0.551) & (0.730) & (0.717) \\
 & $AR_{0.2}$ & 0.902 & 0.919 & 0.902 & 0.914 & 0.903 & 0.869 & 0.901 & 0.861 \\
&& (0.706) & (0.708) & (0.708) & (0.665) & (0.704) & (0.580) & (0.707) & (0.686) \\
 & $AR_{0.5}$ & 0.780 & 0.819 & 0.778 & 0.829 & 0.779 & 0.791 & 0.777 & 0.673 \\
&& (0.600) & (0.626) & (0.600) & (0.614) & (0.596) & (0.560) & (0.599) & (0.536) \\
\bottomrule
\end{tabular}
\caption{Pointwise coverage probabilities for $k=2$ and $a=0.3$}
\label{pointwise1}
\end{table}

\begin{table}
\centering
\begin{tabular}{@{}lllllllllll@{}}
\toprule
 $h$ & \multicolumn{1}{l}{Case} & \multicolumn{1}{c}{$\gamma=0$} & \multicolumn{2}{c}{$\gamma=0.2$} & \multicolumn{2}{c}{$\gamma=0.4$} & \multicolumn{2}{c}{$\gamma=0.6$} & \multicolumn{1}{c}{$-$} \\ \midrule
                                     &  & WB & AWB & DWB & AWB  & DWB & AWB & DWB & SWB
                            \\ \cmidrule(l){3-3} \cmidrule(l){4-5} \cmidrule(l){6-7} \cmidrule(l){8-9} \cmidrule(l){10-10}
\multirow{ 14}{*}{$0.02$} & $0$ & 0.905 & 0.914 & 0.912 & 0.903 & 0.918 & 0.806 & 0.919 & 0.881 \\
&& (0.671) & (0.656) & (0.673) & (0.602) & (0.673) & (0.509) & (0.673) & (0.662) \\
 & $AR_{0.2}$ & 0.847 & 0.876 & 0.845 & 0.853 & 0.846 & 0.781 & 0.863 & 0.781 \\
&& (0.652) & (0.653) & (0.654) & (0.614) & (0.651) & (0.534) & (0.653) & (0.633) \\
 & $AR_{0.5}$ & 0.659 & 0.718 & 0.641 & 0.701 & 0.635 & 0.629 & 0.648 & 0.409 \\
&& (0.554) & (0.577) & (0.554) & (0.566) & (0.550) & (0.516) & (0.553) & (0.495) \\
\bottomrule
\end{tabular}
\caption{Simultaneous coverage probabilities over $G.sub$ for $k=2$ and $a=0.3$}
\label{simultGsub1}
\end{table}

\begin{table}
\centering
\begin{tabular}{@{}lllllllllll@{}}
\toprule
 $h$ & \multicolumn{1}{l}{Case} & \multicolumn{1}{c}{$\gamma=0$} & \multicolumn{2}{c}{$\gamma=0.2$} & \multicolumn{2}{c}{$\gamma=0.4$} & \multicolumn{2}{c}{$\gamma=0.6$} & \multicolumn{1}{c}{$-$} \\ \midrule
                                     &  & WB & AWB & DWB & AWB  & DWB & AWB & DWB & SWB
                            \\ \cmidrule(l){3-3} \cmidrule(l){4-5} \cmidrule(l){6-7} \cmidrule(l){8-9} \cmidrule(l){10-10}
\multirow{ 14}{*}{$0.02$} & $0$ & 0.904 & 0.920 & 0.911 & 0.905 & 0.914 & 0.810 & 0.925 & 0.856 \\
&& (0.532) & (0.519) & (0.533) & (0.477) & (0.532) & (0.403) & (0.533) & (0.524) \\
 & $AR_{0.2}$ & 0.817 & 0.886 & 0.845 & 0.869 & 0.830 & 0.748 & 0.829 & 0.687 \\
&& (0.515) & (0.518) & (0.518) & (0.486) & (0.515) & (0.424) & (0.517) & (0.502) \\
 & $AR_{0.5}$ & 0.576 & 0.678 & 0.589 & 0.638 & 0.555 & 0.574 & 0.566 & 0.237 \\
&& (0.438) & (0.457) & (0.438) & (0.449) & (0.435) & (0.411) & (0.437) & (0.391) \\
\bottomrule
\end{tabular}
\caption{Simultaneous coverage probabilities over $G$ for $k=2$ and $a=0.3$}
\label{simultG1}
\end{table}

\clearpage

\begin{table}
\centering
\begin{tabular}{@{}lllllllllll@{}}
\toprule
 $h$ & \multicolumn{1}{l}{Case} & \multicolumn{1}{c}{$\gamma=0$} & \multicolumn{2}{c}{$\gamma=0.2$} & \multicolumn{2}{c}{$\gamma=0.4$} & \multicolumn{2}{c}{$\gamma=0.6$} & \multicolumn{1}{c}{$-$} \\ \midrule

                                     &  & WB & AWB & DWB & AWB  & DWB & AWB & DWB & SWB
                            \\ \cmidrule(l){3-3} \cmidrule(l){4-5} \cmidrule(l){6-7} \cmidrule(l){8-9} \cmidrule(l){10-10}
\multirow{ 14}{*}{$0.02$} & $0$ & 0.946 & 0.952 & 0.946 & 0.942 & 0.945 & 0.905 & 0.947 & 0.920 \\
&& (0.721) & (0.706) & (0.723) & (0.649) & (0.721) & (0.548) & (0.724) & (0.702) \\
 & $AR_{0.2}$ & 0.901 & 0.918 & 0.902 & 0.913 & 0.903 & 0.868 & 0.901 & 0.857 \\
&& (0.700) & (0.702) & (0.702) & (0.659) & (0.699) & (0.575) & (0.701) & (0.674) \\
 & $AR_{0.5}$ & 0.780 & 0.819 & 0.778 & 0.827 & 0.779 & 0.789 & 0.777 & 0.671 \\
&& (0.594) & (0.619) & (0.594) & (0.607) & (0.591) & (0.553) & (0.594) & (0.526) \\
\bottomrule
\end{tabular}
\caption{Pointwise coverage probabilities for $k=2$ and $a=0.5$}
\label{pointwise2}
\end{table}

\begin{table}
\centering
\begin{tabular}{@{}lllllllllll@{}}
\toprule
 $h$ & \multicolumn{1}{l}{Case} & \multicolumn{1}{c}{$\gamma=0$} & \multicolumn{2}{c}{$\gamma=0.2$} & \multicolumn{2}{c}{$\gamma=0.4$} & \multicolumn{2}{c}{$\gamma=0.6$} & \multicolumn{1}{c}{$-$} \\ \midrule

                                     &  & WB & AWB & DWB & AWB  & DWB & AWB & DWB & SWB
                            \\ \cmidrule(l){3-3} \cmidrule(l){4-5} \cmidrule(l){6-7} \cmidrule(l){8-9} \cmidrule(l){10-10}
\multirow{ 14}{*}{$0.02$} & $0$ & 0.902 & 0.909 & 0.906 & 0.892 & 0.916 & 0.790 & 0.913 & 0.877 \\
&& (0.665) & (0.652) & (0.668) & (0.598) & (0.667) & (0.506) & (0.668) & (0.648) \\
 & $AR_{0.2}$ & 0.838 & 0.871 & 0.857 & 0.835 & 0.838 & 0.764 & 0.854 & 0.790 \\
&& (0.647) & (0.648) & (0.648) & (0.608) & (0.645) & (0.529) & (0.647) & (0.622) \\
 & $AR_{0.5}$ & 0.647 & 0.704 & 0.640 & 0.684 & 0.631 & 0.620 & 0.643 & 0.412 \\
&& (0.548) & (0.571) & (0.549) & (0.560) & (0.546) & (0.509) & (0.549) & (0.485) \\
\bottomrule
\end{tabular}
\caption{Simultaneous coverage probabilities over $G.sub$ for $k=2$ and $a=0.5$}
\label{simultGsub2}
\end{table}

\begin{table}
\centering
\begin{tabular}{@{}lllllllllll@{}}
\toprule
 $h$ & \multicolumn{1}{l}{Case} & \multicolumn{1}{c}{$\gamma=0$} & \multicolumn{2}{c}{$\gamma=0.2$} & \multicolumn{2}{c}{$\gamma=0.4$} & \multicolumn{2}{c}{$\gamma=0.6$} & \multicolumn{1}{c}{$-$} \\ \midrule

                                     &  & WB & AWB & DWB & AWB  & DWB & AWB & DWB & SWB
                            \\ \cmidrule(l){3-3} \cmidrule(l){4-5} \cmidrule(l){6-7} \cmidrule(l){8-9} \cmidrule(l){10-10}
\multirow{ 14}{*}{$0.02$} & $0$ & 0.898 & 0.912 & 0.908 & 0.897 & 0.908 & 0.789 & 0.914 & 0.845 \\
&& (0.527) & (0.516) & (0.528) & (0.474) & (0.528) & (0.400) & (0.529) & (0.513) \\
 & $AR_{0.2}$ & 0.810 & 0.875 & 0.841 & 0.857 & 0.824 & 0.736 & 0.824 & 0.682 \\
&& (0.511) & (0.513) & (0.513) & (0.482) & (0.511) & (0.420) & (0.512) & (0.492) \\
 & $AR_{0.5}$ & 0.574 & 0.665 & 0.581 & 0.624 & 0.545 & 0.560 & 0.554 & 0.241 \\
&& (0.433) & (0.453) & (0.434) & (0.444) & (0.432) & (0.405) & (0.434) & (0.384) \\
\bottomrule
\end{tabular}
\caption{Simultaneous coverage probabilities over $G$ for $k=2$ and $a=0.5$}
\label{simultG2}
\end{table}

\clearpage

\begin{table}
\centering
\begin{tabular}{@{}lllllllllll@{}}
\toprule
 $h$ & \multicolumn{1}{l}{Case} & \multicolumn{1}{c}{$\gamma=0$} & \multicolumn{2}{c}{$\gamma=0.2$} & \multicolumn{2}{c}{$\gamma=0.4$} & \multicolumn{2}{c}{$\gamma=0.6$} & \multicolumn{1}{c}{$-$} \\ \midrule

                                     &  & WB & AWB & DWB & AWB  & DWB & AWB & DWB & SWB
                            \\ \cmidrule(l){3-3} \cmidrule(l){4-5} \cmidrule(l){6-7} \cmidrule(l){8-9} \cmidrule(l){10-10}
\multirow{ 14}{*}{$0.02$} & $0$ & 0.945 & 0.951 & 0.945 & 0.941 & 0.944 & 0.903 & 0.947 & 0.914 \\
&& (0.713) & (0.699) & (0.715) & (0.645) & (0.714) & (0.546) & (0.715) & (0.685) \\
 & $AR_{0.2}$ & 0.901 & 0.918 & 0.902 & 0.911 & 0.903 & 0.866 & 0.900 & 0.850 \\
&& (0.692) & (0.696) & (0.693) & (0.654) & (0.692) & (0.571) & (0.693) & (0.657) \\
 & $AR_{0.5}$ & 0.781 & 0.819 & 0.779 & 0.827 & 0.780 & 0.789 & 0.778 & 0.666 \\
&& (0.588) & (0.614) & (0.590) & (0.601) & (0.586) & (0.548) & (0.589) & (0.512) \\
\bottomrule
\end{tabular}
\caption{Pointwise coverage probabilities for $k=2$ and $a=0.7$}
\label{pointwise3}
\end{table}

\begin{table}
\centering
\begin{tabular}{@{}lllllllllll@{}}
\toprule
 $h$ & \multicolumn{1}{l}{Case} & \multicolumn{1}{c}{$\gamma=0$} & \multicolumn{2}{c}{$\gamma=0.2$} & \multicolumn{2}{c}{$\gamma=0.4$} & \multicolumn{2}{c}{$\gamma=0.6$} & \multicolumn{1}{c}{$-$} \\ \midrule

                                     &  & WB & AWB & DWB & AWB  & DWB & AWB & DWB & SWB
                            \\ \cmidrule(l){3-3} \cmidrule(l){4-5} \cmidrule(l){6-7} \cmidrule(l){8-9} \cmidrule(l){10-10}
\multirow{ 14}{*}{$0.02$} & $0$ & 0.894 & 0.904 & 0.893 & 0.872 & 0.907 & 0.781 & 0.907 & 0.877 \\
&& (0.658) & (0.645) & (0.660) & (0.596) & (0.660) & (0.504) & (0.660) & (0.633) \\
 & $AR_{0.2}$ & 0.822 & 0.856 & 0.851 & 0.819 & 0.832 & 0.744 & 0.852 & 0.781 \\
&& (0.640) & (0.641) & (0.641) & (0.604) & (0.639) & (0.526) & (0.640) & (0.607) \\
 & $AR_{0.5}$ & 0.637 & 0.696 & 0.641 & 0.673 & 0.632 & 0.598 & 0.625 & 0.415 \\
&& (0.542) & (0.567) & (0.545) & (0.555) & (0.542) & (0.504) & (0.544) & (0.473) \\
\bottomrule
\end{tabular}
\caption{Simultaneous coverage probabilities over $G.sub$ for $k=2$ and $a=0.7$}
\label{simultGsub3}
\end{table}

\begin{table}
\centering
\begin{tabular}{@{}lllllllllll@{}}
\toprule
 $h$ & \multicolumn{1}{l}{Case} & \multicolumn{1}{c}{$\gamma=0$} & \multicolumn{2}{c}{$\gamma=0.2$} & \multicolumn{2}{c}{$\gamma=0.4$} & \multicolumn{2}{c}{$\gamma=0.6$} & \multicolumn{1}{c}{$-$} \\ \midrule

                                     &  & WB & AWB & DWB & AWB  & DWB & AWB & DWB & SWB
                            \\ \cmidrule(l){3-3} \cmidrule(l){4-5} \cmidrule(l){6-7} \cmidrule(l){8-9} \cmidrule(l){10-10}
\multirow{ 14}{*}{$0.02$} & $0$ & 0.887 & 0.906 & 0.898 & 0.880 & 0.905 & 0.781 & 0.904 & 0.847 \\
&& (0.521) & (0.511) & (0.522) & (0.472) & (0.521) & (0.399) & (0.523) & (0.501) \\
 & $AR_{0.2}$ & 0.797 & 0.868 & 0.829 & 0.840 & 0.810 & 0.724 & 0.812 & 0.682 \\
&& (0.506) & (0.508) & (0.507) & (0.479) & (0.506) & (0.418) & (0.506) & (0.480) \\
 & $AR_{0.5}$ & 0.561 & 0.654 & 0.569 & 0.612 & 0.534 & 0.546 & 0.547 & 0.243 \\
&& (0.430) & (0.449) & (0.431) & (0.440) & (0.429) & (0.401) & (0.430) & (0.374) \\
\bottomrule
\end{tabular}
\caption{Simultaneous coverage probabilities over $G$ for $k=2$ and $a=0.7$}
\label{simultG3}
\end{table}

\clearpage

\begin{table}
\centering
\begin{tabular}{@{}lllllllllll@{}}
\toprule
 $h$ & \multicolumn{1}{l}{Case} & \multicolumn{1}{c}{$\gamma=0$} & \multicolumn{2}{c}{$\gamma=0.2$} & \multicolumn{2}{c}{$\gamma=0.4$} & \multicolumn{2}{c}{$\gamma=0.6$} & \multicolumn{1}{c}{$-$} \\ \midrule

                                     &  & WB & AWB & DWB & AWB  & DWB & AWB & DWB & SWB
                            \\ \cmidrule(l){3-3} \cmidrule(l){4-5} \cmidrule(l){6-7} \cmidrule(l){8-9} \cmidrule(l){10-10}
\multirow{ 14}{*}{$0.02$} & $0$ & 0.946 & 0.952 & 0.946 & 0.944 & 0.945 & 0.906 & 0.947 & 0.927 \\
&& (0.728) & (0.712) & (0.730) & (0.653) & (0.730) & (0.550) & (0.731) & (0.719) \\
 & $AR_{0.2}$ & 0.902 & 0.919 & 0.902 & 0.914 & 0.903 & 0.869 & 0.901 & 0.862 \\
&& (0.708) & (0.710) & (0.709) & (0.665) & (0.705) & (0.578) & (0.707) & (0.690) \\
 & $AR_{0.5}$ & 0.780 & 0.819 & 0.778 & 0.829 & 0.780 & 0.791 & 0.778 & 0.674 \\
&& (0.601) & (0.626) & (0.601) & (0.614) & (0.597) & (0.561) & (0.599) & (0.538) \\
\bottomrule
\end{tabular}
\caption{Pointwise coverage probabilities for $k=3$ and $k=0.3$}
\label{pointwise4}
\end{table}

\begin{table}
\centering
\begin{tabular}{@{}lllllllllll@{}}
\toprule
 $h$ & \multicolumn{1}{l}{Case} & \multicolumn{1}{c}{$\gamma=0$} & \multicolumn{2}{c}{$\gamma=0.2$} & \multicolumn{2}{c}{$\gamma=0.4$} & \multicolumn{2}{c}{$\gamma=0.6$} & \multicolumn{1}{c}{$-$} \\ \midrule

                                     &  & WB & AWB & DWB & AWB  & DWB & AWB & DWB & SWB
                            \\ \cmidrule(l){3-3} \cmidrule(l){4-5} \cmidrule(l){6-7} \cmidrule(l){8-9} \cmidrule(l){10-10}
\multirow{ 14}{*}{$0.02$} & $0$ & 0.905 & 0.915 & 0.912 & 0.904 & 0.917 & 0.814 & 0.922 & 0.879 \\
&& (0.672) & (0.657) & (0.674) & (0.603) & (0.674) & (0.508) & (0.675) & (0.664) \\
 & $AR_{0.2}$ & 0.845 & 0.880 & 0.843 & 0.842 & 0.853 & 0.778 & 0.861 & 0.787 \\
&& (0.653) & (0.655) & (0.655) & (0.614) & (0.651) & (0.532) & (0.653) & (0.636) \\
 & $AR_{0.5}$ & 0.657 & 0.713 & 0.643 & 0.704 & 0.643 & 0.609 & 0.640 & 0.406 \\
&& (0.555) & (0.577) & (0.554) & (0.567) & (0.552) & (0.516) & (0.553) & (0.497) \\
\bottomrule
\end{tabular}
\caption{Simultaneous coverage probabilities over $G.sub$ for $k=3$ and $a=0.3$}
\label{simultGsub4}
\end{table}

\begin{table}
\centering
\begin{tabular}{@{}lllllllllll@{}}
\toprule
 $h$ & \multicolumn{1}{l}{Case} & \multicolumn{1}{c}{$\gamma=0$} & \multicolumn{2}{c}{$\gamma=0.2$} & \multicolumn{2}{c}{$\gamma=0.4$} & \multicolumn{2}{c}{$\gamma=0.6$} & \multicolumn{1}{c}{$-$} \\ \midrule

                                     &  & WB & AWB & DWB & AWB  & DWB & AWB & DWB & SWB
                            \\ \cmidrule(l){3-3} \cmidrule(l){4-5} \cmidrule(l){6-7} \cmidrule(l){8-9} \cmidrule(l){10-10}
\multirow{ 14}{*}{$0.02$} & $0$ & 0.906 & 0.916 & 0.904 & 0.903 & 0.915 & 0.810 & 0.921 & 0.861 \\
&& (0.532) & (0.520) & (0.533) & (0.477) & (0.533) & (0.402) & (0.533) & (0.525) \\
 & $AR_{0.2}$ & 0.815 & 0.896 & 0.848 & 0.859 & 0.824 & 0.738 & 0.830 & 0.687 \\
&& (0.517) & (0.519) & (0.518) & (0.487) & (0.514) & (0.423) & (0.516) & (0.504) \\
 & $AR_{0.5}$ & 0.567 & 0.679 & 0.585 & 0.650 & 0.544 & 0.554 & 0.564 & 0.238 \\
&& (0.439) & (0.457) & (0.439) & (0.450) & (0.436) & (0.411) & (0.438) & (0.393) \\
\bottomrule
\end{tabular}
\caption{Simultaneous coverage probabilities over $G$ for $k=3$ and $a=0.3$}
\label{simultG4}
\end{table}

\clearpage

\begin{table}
\centering
\begin{tabular}{@{}lllllllllll@{}}
\toprule
 $h$ & \multicolumn{1}{l}{Case} & \multicolumn{1}{c}{$\gamma=0$} & \multicolumn{2}{c}{$\gamma=0.2$} & \multicolumn{2}{c}{$\gamma=0.4$} & \multicolumn{2}{c}{$\gamma=0.6$} & \multicolumn{1}{c}{$-$} \\ \midrule

                                     &  & WB & AWB & DWB & AWB  & DWB & AWB & DWB & SWB
                            \\ \cmidrule(l){3-3} \cmidrule(l){4-5} \cmidrule(l){6-7} \cmidrule(l){8-9} \cmidrule(l){10-10}
\multirow{ 14}{*}{$0.02$} & $0$ & 0.946 & 0.952 & 0.946 & 0.942 & 0.945 & 0.905 & 0.947 & 0.923 \\
&& (0.722) & (0.706) & (0.723) & (0.648) & (0.723) & (0.546) & (0.724) & (0.707) \\
 & $AR_{0.2}$ & 0.902 & 0.919 & 0.902 & 0.912 & 0.903 & 0.867 & 0.901 & 0.859 \\
&& (0.702) & (0.703) & (0.703) & (0.659) & (0.699) & (0.573) & (0.700) & (0.679) \\
 & $AR_{0.5}$ & 0.780 & 0.819 & 0.778 & 0.828 & 0.780 & 0.790 & 0.778 & 0.672 \\
&& (0.594) & (0.619) & (0.594) & (0.608) & (0.592) & (0.554) & (0.593) & (0.530) \\
\bottomrule
\end{tabular}
\caption{Pointwise coverage probabilities for $k=3$ and $a=0.5$}
\label{pointwise5}
\end{table}

\begin{table}
\centering
\begin{tabular}{@{}lllllllllll@{}}
\toprule
 $h$ & \multicolumn{1}{l}{Case} & \multicolumn{1}{c}{$\gamma=0$} & \multicolumn{2}{c}{$\gamma=0.2$} & \multicolumn{2}{c}{$\gamma=0.4$} & \multicolumn{2}{c}{$\gamma=0.6$} & \multicolumn{1}{c}{$-$} \\ \midrule

                                     &  & WB & AWB & DWB & AWB  & DWB & AWB & DWB & SWB
                            \\ \cmidrule(l){3-3} \cmidrule(l){4-5} \cmidrule(l){6-7} \cmidrule(l){8-9} \cmidrule(l){10-10}
\multirow{ 14}{*}{$0.02$} & $0$ & 0.896 & 0.905 & 0.911 & 0.892 & 0.914 & 0.784 & 0.911 & 0.884 \\
&& (0.667) & (0.651) & (0.668) & (0.598) & (0.668) & (0.504) & (0.668) & (0.653) \\
 & $AR_{0.2}$ & 0.827 & 0.869 & 0.840 & 0.826 & 0.836 & 0.754 & 0.853 & 0.785 \\
&& (0.648) & (0.649) & (0.650) & (0.608) & (0.646) & (0.528) & (0.647) & (0.627) \\
 & $AR_{0.5}$ & 0.649 & 0.695 & 0.639 & 0.687 & 0.637 & 0.592 & 0.628 & 0.417 \\
&& (0.549) & (0.571) & (0.549) & (0.561) & (0.547) & (0.511) & (0.548) & (0.489) \\
\bottomrule
\end{tabular}
\caption{Simultaneous coverage probabilities over $G.sub$ for $k=3$ and $a=0.5$}
\label{simultGsub5}
\end{table}

\begin{table}
\centering
\begin{tabular}{@{}lllllllllll@{}}
\toprule
 $h$ & \multicolumn{1}{l}{Case} & \multicolumn{1}{c}{$\gamma=0$} & \multicolumn{2}{c}{$\gamma=0.2$} & \multicolumn{2}{c}{$\gamma=0.4$} & \multicolumn{2}{c}{$\gamma=0.6$} & \multicolumn{1}{c}{$-$} \\ \midrule

                                     &  & WB & AWB & DWB & AWB  & DWB & AWB & DWB & SWB
                            \\ \cmidrule(l){3-3} \cmidrule(l){4-5} \cmidrule(l){6-7} \cmidrule(l){8-9} \cmidrule(l){10-10}
\multirow{ 14}{*}{$0.02$} & $0$ & 0.901 & 0.910 & 0.898 & 0.896 & 0.913 & 0.794 & 0.911 & 0.862 \\
&& (0.527) & (0.515) & (0.528) & (0.473) & (0.528) & (0.399) & (0.528) & (0.516) \\
 & $AR_{0.2}$ & 0.802 & 0.877 & 0.840 & 0.836 & 0.814 & 0.722 & 0.822 & 0.686 \\
&& (0.512) & (0.514) & (0.514) & (0.482) & (0.510) & (0.419) & (0.511) & (0.496) \\
 & $AR_{0.5}$ & 0.560 & 0.664 & 0.567 & 0.633 & 0.533 & 0.538 & 0.554 & 0.250 \\
&& (0.434) & (0.453) & (0.434) & (0.445) & (0.432) & (0.407) & (0.433) & (0.387) \\
\bottomrule
\end{tabular}
\caption{Simultaneous coverage probabilities over $G$ for $a=3$ and $a=0.5$}
\label{simultG5}
\end{table}

\clearpage

\begin{table}
\centering
\begin{tabular}{@{}lllllllllll@{}}
\toprule
 $h$ & \multicolumn{1}{l}{Case} & \multicolumn{1}{c}{$\gamma=0$} & \multicolumn{2}{c}{$\gamma=0.2$} & \multicolumn{2}{c}{$\gamma=0.4$} & \multicolumn{2}{c}{$\gamma=0.6$} & \multicolumn{1}{c}{$-$} \\ \midrule

                                     &  & WB & AWB & DWB & AWB  & DWB & AWB & DWB & SWB
                            \\ \cmidrule(l){3-3} \cmidrule(l){4-5} \cmidrule(l){6-7} \cmidrule(l){8-9} \cmidrule(l){10-10}
\multirow{ 14}{*}{$0.02$} & $0$ & 0.946 & 0.951 & 0.946 & 0.941 & 0.945 & 0.903 & 0.947 & 0.920 \\
&& (0.713) & (0.699) & (0.716) & (0.642) & (0.715) & (0.541) & (0.714) & (0.691) \\
 & $AR_{0.2}$ & 0.901 & 0.918 & 0.902 & 0.911 & 0.903 & 0.866 & 0.901 & 0.856 \\
&& (0.694) & (0.696) & (0.695) & (0.653) & (0.692) & (0.568) & (0.693) & (0.665) \\
 & $AR_{0.5}$ & 0.783 & 0.819 & 0.779 & 0.827 & 0.781 & 0.790 & 0.780 & 0.670 \\
&& (0.589) & (0.613) & (0.589) & (0.602) & (0.586) & (0.547) & (0.588) & (0.518) \\
\bottomrule
\end{tabular}
\caption{Pointwise coverage probabilities for $k=3$ and $a=0.7$}
\label{pointwise6}
\end{table}

\begin{table}
\centering
\begin{tabular}{@{}lllllllllll@{}}
\toprule
 $h$ & \multicolumn{1}{l}{Case} & \multicolumn{1}{c}{$\gamma=0$} & \multicolumn{2}{c}{$\gamma=0.2$} & \multicolumn{2}{c}{$\gamma=0.4$} & \multicolumn{2}{c}{$\gamma=0.6$} & \multicolumn{1}{c}{$-$} \\ \midrule

                                     &  & WB & AWB & DWB & AWB  & DWB & AWB & DWB & SWB
                            \\ \cmidrule(l){3-3} \cmidrule(l){4-5} \cmidrule(l){6-7} \cmidrule(l){8-9} \cmidrule(l){10-10}
\multirow{ 14}{*}{$0.02$} & $0$ & 0.891 & 0.894 & 0.895 & 0.869 & 0.908 & 0.764 & 0.902 & 0.874 \\
&& (0.659) & (0.645) & (0.660) & (0.592) & (0.660) & (0.499) & (0.660) & (0.639) \\
 & $AR_{0.2}$ & 0.824 & 0.845 & 0.832 & 0.805 & 0.818 & 0.723 & 0.846 & 0.787 \\
&& (0.642) & (0.642) & (0.643) & (0.603) & (0.639) & (0.523) & (0.640) & (0.614) \\
 & $AR_{0.5}$ & 0.648 & 0.672 & 0.636 & 0.666 & 0.629 & 0.584 & 0.633 & 0.432 \\
&& (0.544) & (0.565) & (0.544) & (0.554) & (0.542) & (0.504) & (0.543) & (0.479) \\
\bottomrule
\end{tabular}
\caption{Simultaneous coverage probabilities over $G.sub$ for $k=3$ and $a=0.7$}
\label{simultGsub6}
\end{table}

\begin{table}
\centering
\begin{tabular}{@{}lllllllllll@{}}
\toprule
 $h$ & \multicolumn{1}{l}{Case} & \multicolumn{1}{c}{$\gamma=0$} & \multicolumn{2}{c}{$\gamma=0.2$} & \multicolumn{2}{c}{$\gamma=0.4$} & \multicolumn{2}{c}{$\gamma=0.6$} & \multicolumn{1}{c}{$-$} \\ \midrule

                                     &  & WB & AWB & DWB & AWB  & DWB & AWB & DWB & SWB
                            \\ \cmidrule(l){3-3} \cmidrule(l){4-5} \cmidrule(l){6-7} \cmidrule(l){8-9} \cmidrule(l){10-10}
\multirow{ 14}{*}{$0.02$} & $0$ & 0.890 & 0.900 & 0.893 & 0.880 & 0.906 & 0.770 & 0.906 & 0.855 \\
&& (0.521) & (0.510) & (0.522) & (0.468) & (0.522) & (0.396) & (0.522) & (0.505) \\
 & $AR_{0.2}$ & 0.791 & 0.856 & 0.824 & 0.809 & 0.789 & 0.699 & 0.818 & 0.690 \\
&& (0.507) & (0.508) & (0.508) & (0.478) & (0.505) & (0.415) & (0.506) & (0.486) \\
 & $AR_{0.5}$ & 0.555 & 0.640 & 0.557 & 0.617 & 0.531 & 0.508 & 0.542 & 0.257 \\
&& (0.430) & (0.448) & (0.430) & (0.440) & (0.428) & (0.402) & (0.429) & (0.379) \\
\bottomrule
\end{tabular}
\caption{Simultaneous coverage probabilities over $G$ for $k=3$ and $a=0.7$}
\label{simultG6}
\end{table}

\clearpage

\begin{table}
\centering
\begin{tabular}{@{}lllllllllll@{}}
\toprule
 $h$ & \multicolumn{1}{l}{Case} & \multicolumn{1}{c}{$\gamma=0$} & \multicolumn{2}{c}{$\gamma=0.2$} & \multicolumn{2}{c}{$\gamma=0.4$} & \multicolumn{2}{c}{$\gamma=0.6$} & \multicolumn{1}{c}{$-$} \\ \midrule

                                     &  & WB & AWB & DWB & AWB  & DWB & AWB & DWB & SWB
                            \\ \cmidrule(l){3-3} \cmidrule(l){4-5} \cmidrule(l){6-7} \cmidrule(l){8-9} \cmidrule(l){10-10}
\multirow{ 14}{*}{$0.02$} & $0$ & 0.946 & 0.952 & 0.946 & 0.943 & 0.945 & 0.906 & 0.947 & 0.928 \\
&& (0.731) & (0.713) & (0.731) & (0.653) & (0.730) & (0.552) & (0.732) & (0.722) \\
 & $AR_{0.2}$ & 0.902 & 0.919 & 0.903 & 0.914 & 0.903 & 0.870 & 0.901 & 0.863 \\
&& (0.710) & (0.710) & (0.710) & (0.666) & (0.706) & (0.582) & (0.708) & (0.691) \\
 & $AR_{0.5}$ & 0.780 & 0.820 & 0.779 & 0.829 & 0.780 & 0.792 & 0.778 & 0.674 \\
&& (0.604) & (0.628) & (0.601) & (0.615) & (0.599) & (0.562) & (0.600) & (0.539) \\
\bottomrule
\end{tabular}
\caption{Pointwise coverage probabilities for $k=4$ and $a=0.3$}
\label{pointwise7}
\end{table}

\begin{table}
\centering
\begin{tabular}{@{}lllllllllll@{}}
\toprule
 $h$ & \multicolumn{1}{l}{Case} & \multicolumn{1}{c}{$\gamma=0$} & \multicolumn{2}{c}{$\gamma=0.2$} & \multicolumn{2}{c}{$\gamma=0.4$} & \multicolumn{2}{c}{$\gamma=0.6$} & \multicolumn{1}{c}{$-$} \\ \midrule

                                     &  & WB & AWB & DWB & AWB  & DWB & AWB & DWB & SWB
                            \\ \cmidrule(l){3-3} \cmidrule(l){4-5} \cmidrule(l){6-7} \cmidrule(l){8-9} \cmidrule(l){10-10}
\multirow{ 14}{*}{$0.02$} & $0$ & 0.921 & 0.932 & 0.921 & 0.926 & 0.923 & 0.865 & 0.930 & 0.883 \\
&& (0.674) & (0.658) & (0.674) & (0.602) & (0.674) & (0.509) & (0.675) & (0.666) \\
 & $AR_{0.2}$ & 0.856 & 0.898 & 0.871 & 0.872 & 0.863 & 0.806 & 0.865 & 0.783 \\
&& (0.655) & (0.656) & (0.656) & (0.615) & (0.652) & (0.536) & (0.654) & (0.638) \\
 & $AR_{0.5}$ & 0.661 & 0.738 & 0.660 & 0.733 & 0.657 & 0.659 & 0.659 & 0.389 \\
&& (0.557) & (0.578) & (0.555) & (0.567) & (0.552) & (0.517) & (0.554) & (0.497) \\
\bottomrule
\end{tabular}
\caption{Simultaneous coverage probabilities over $G.sub$ for $k=4$ and $a=0.3$}
\label{simultGsub7}
\end{table}

\begin{table}
\centering
\begin{tabular}{@{}lllllllllll@{}}
\toprule
 $h$ & \multicolumn{1}{l}{Case} & \multicolumn{1}{c}{$\gamma=0$} & \multicolumn{2}{c}{$\gamma=0.2$} & \multicolumn{2}{c}{$\gamma=0.4$} & \multicolumn{2}{c}{$\gamma=0.6$} & \multicolumn{1}{c}{$-$} \\ \midrule

                                     &  & WB & AWB & DWB & AWB  & DWB & AWB & DWB & SWB
                            \\ \cmidrule(l){3-3} \cmidrule(l){4-5} \cmidrule(l){6-7} \cmidrule(l){8-9} \cmidrule(l){10-10}
\multirow{ 14}{*}{$0.02$} & $0$ & 0.908 & 0.923 & 0.910 & 0.913 & 0.915 & 0.850 & 0.925 & 0.855 \\
&& (0.534) & (0.521) & (0.534) & (0.477) & (0.533) & (0.404) & (0.534) & (0.527) \\
 & $AR_{0.2}$ & 0.823 & 0.887 & 0.853 & 0.874 & 0.831 & 0.746 & 0.815 & 0.688 \\
&& (0.518) & (0.520) & (0.520) & (0.487) & (0.516) & (0.425) & (0.518) & (0.505) \\
 & $AR_{0.5}$ & 0.577 & 0.674 & 0.574 & 0.640 & 0.563 & 0.562 & 0.573 & 0.236 \\
&& (0.440) & (0.458) & (0.439) & (0.450) & (0.437) & (0.412) & (0.438) & (0.394) \\
\bottomrule
\end{tabular}
\caption{Simultaneous coverage probabilities over $G$ for $k=4$ and $a=0.3$}
\label{simultG7}
\end{table}

\clearpage

\begin{table}
\centering
\begin{tabular}{@{}lllllllllll@{}}
\toprule
 $h$ & \multicolumn{1}{l}{Case} & \multicolumn{1}{c}{$\gamma=0$} & \multicolumn{2}{c}{$\gamma=0.2$} & \multicolumn{2}{c}{$\gamma=0.4$} & \multicolumn{2}{c}{$\gamma=0.6$} & \multicolumn{1}{c}{$-$} \\ \midrule

                                     &  & WB & AWB & DWB & AWB  & DWB & AWB & DWB & SWB
                            \\ \cmidrule(l){3-3} \cmidrule(l){4-5} \cmidrule(l){6-7} \cmidrule(l){8-9} \cmidrule(l){10-10}
\multirow{ 14}{*}{$0.02$} & $0$ & 0.946 & 0.952 & 0.945 & 0.942 & 0.945 & 0.904 & 0.947 & 0.926 \\
&& (0.725) & (0.708) & (0.724) & (0.648) & (0.724) & (0.549) & (0.726) & (0.712) \\
 & $AR_{0.2}$ & 0.902 & 0.919 & 0.903 & 0.913 & 0.903 & 0.868 & 0.901 & 0.861 \\
&& (0.704) & (0.706) & (0.705) & (0.660) & (0.701) & (0.576) & (0.702) & (0.682) \\
 & $AR_{0.5}$ & 0.781 & 0.820 & 0.780 & 0.829 & 0.781 & 0.791 & 0.779 & 0.673 \\
&& (0.599) & (0.622) & (0.596) & (0.610) & (0.594) & (0.557) & (0.596) & (0.531) \\
\bottomrule
\end{tabular}
\caption{Pointwise coverage probabilities for $k=4$ and $a=0.5$}
\label{pointwise8}
\end{table}

\begin{table}
\centering
\begin{tabular}{@{}lllllllllll@{}}
\toprule
 $h$ & \multicolumn{1}{l}{Case} & \multicolumn{1}{c}{$\gamma=0$} & \multicolumn{2}{c}{$\gamma=0.2$} & \multicolumn{2}{c}{$\gamma=0.4$} & \multicolumn{2}{c}{$\gamma=0.6$} & \multicolumn{1}{c}{$-$} \\ \midrule

                                     &  & WB & AWB & DWB & AWB  & DWB & AWB & DWB & SWB
                            \\ \cmidrule(l){3-3} \cmidrule(l){4-5} \cmidrule(l){6-7} \cmidrule(l){8-9} \cmidrule(l){10-10}
\multirow{ 14}{*}{$0.02$} & $0$ & 0.926 & 0.933 & 0.921 & 0.934 & 0.929 & 0.864 & 0.936 & 0.883 \\
&& (0.669) & (0.653) & (0.668) & (0.597) & (0.668) & (0.507) & (0.669) & (0.657) \\
 & $AR_{0.2}$ & 0.865 & 0.898 & 0.871 & 0.881 & 0.861 & 0.813 & 0.868 & 0.784 \\
&& (0.650) & (0.651) & (0.651) & (0.609) & (0.646) & (0.530) & (0.648) & (0.630) \\
 & $AR_{0.5}$ & 0.665 & 0.739 & 0.665 & 0.745 & 0.660 & 0.665 & 0.659 & 0.384 \\
&& (0.552) & (0.574) & (0.551) & (0.562) & (0.548) & (0.513) & (0.549) & (0.490) \\
\bottomrule
\end{tabular}
\caption{Simultaneous coverage probabilities over $G.sub$ for $k=4$ and $a=0.5$}
\label{simultGsub8}
\end{table}

\begin{table}
\centering
\begin{tabular}{@{}lllllllllll@{}}
\toprule
 $h$ & \multicolumn{1}{l}{Case} & \multicolumn{1}{c}{$\gamma=0$} & \multicolumn{2}{c}{$\gamma=0.2$} & \multicolumn{2}{c}{$\gamma=0.4$} & \multicolumn{2}{c}{$\gamma=0.6$} & \multicolumn{1}{c}{$-$} \\ \midrule

                                     &  & WB & AWB & DWB & AWB  & DWB & AWB & DWB & SWB
                            \\ \cmidrule(l){3-3} \cmidrule(l){4-5} \cmidrule(l){6-7} \cmidrule(l){8-9} \cmidrule(l){10-10}
\multirow{ 14}{*}{$0.02$} & $0$ & 0.904 & 0.918 & 0.905 & 0.910 & 0.912 & 0.837 & 0.928 & 0.848 \\
&& (0.530) & (0.517) & (0.529) & (0.473) & (0.529) & (0.401) & (0.530) & (0.520) \\
 & $AR_{0.2}$ & 0.822 & 0.876 & 0.849 & 0.866 & 0.823 & 0.731 & 0.815 & 0.685 \\
&& (0.514) & (0.515) & (0.515) & (0.483) & (0.512) & (0.421) & (0.513) & (0.499) \\
 & $AR_{0.5}$ & 0.576 & 0.655 & 0.560 & 0.629 & 0.558 & 0.554 & 0.564 & 0.229 \\
&& (0.437) & (0.455) & (0.436) & (0.447) & (0.434) & (0.408) & (0.435) & (0.388) \\
\bottomrule
\end{tabular}
\caption{Simultaneous coverage probabilities over $G$ for $k=4$ and $a=0.5$}
\label{simultG8}
\end{table}

\clearpage

\begin{table}
\centering
\begin{tabular}{@{}lllllllllll@{}}
\toprule
 $h$ & \multicolumn{1}{l}{Case} & \multicolumn{1}{c}{$\gamma=0$} & \multicolumn{2}{c}{$\gamma=0.2$} & \multicolumn{2}{c}{$\gamma=0.4$} & \multicolumn{2}{c}{$\gamma=0.6$} & \multicolumn{1}{c}{$-$} \\ \midrule

                                     &  & WB & AWB & DWB & AWB  & DWB & AWB & DWB & SWB
                            \\ \cmidrule(l){3-3} \cmidrule(l){4-5} \cmidrule(l){6-7} \cmidrule(l){8-9} \cmidrule(l){10-10}
\multirow{ 14}{*}{$0.02$} & $0$ & 0.945 & 0.951 & 0.945 & 0.941 & 0.945 & 0.903 & 0.947 & 0.923 \\
&& (0.719) & (0.702) & (0.717) & (0.643) & (0.717) & (0.546) & (0.719) & (0.701) \\
 & $AR_{0.2}$ & 0.902 & 0.918 & 0.903 & 0.912 & 0.903 & 0.867 & 0.901 & 0.860 \\
&& (0.697) & (0.700) & (0.698) & (0.655) & (0.694) & (0.570) & (0.695) & (0.672) \\
 & $AR_{0.5}$ & 0.783 & 0.820 & 0.781 & 0.829 & 0.783 & 0.792 & 0.782 & 0.674 \\
&& (0.594) & (0.617) & (0.593) & (0.605) & (0.590) & (0.550) & (0.591) & (0.521) \\
\bottomrule
\end{tabular}
\caption{Pointwise coverage probabilities for $k=4$ and $a=0.7$}
\label{pointwise9}
\end{table}

\begin{table}
\centering
\begin{tabular}{@{}lllllllllll@{}}
\toprule
 $h$ & \multicolumn{1}{l}{Case} & \multicolumn{1}{c}{$\gamma=0$} & \multicolumn{2}{c}{$\gamma=0.2$} & \multicolumn{2}{c}{$\gamma=0.4$} & \multicolumn{2}{c}{$\gamma=0.6$} & \multicolumn{1}{c}{$-$} \\ \midrule

                                     &  & WB & AWB & DWB & AWB  & DWB & AWB & DWB & SWB
                            \\ \cmidrule(l){3-3} \cmidrule(l){4-5} \cmidrule(l){6-7} \cmidrule(l){8-9} \cmidrule(l){10-10}
\multirow{ 14}{*}{$0.02$} & $0$ & 0.929 & 0.932 & 0.924 & 0.938 & 0.934 & 0.877 & 0.939 & 0.884 \\
&& (0.662) & (0.647) & (0.661) & (0.594) & (0.661) & (0.503) & (0.663) & (0.646) \\
 & $AR_{0.2}$ & 0.866 & 0.905 & 0.876 & 0.887 & 0.863 & 0.816 & 0.870 & 0.772 \\
&& (0.643) & (0.645) & (0.645) & (0.604) & (0.640) & (0.525) & (0.641) & (0.619) \\
 & $AR_{0.5}$ & 0.665 & 0.746 & 0.667 & 0.751 & 0.661 & 0.665 & 0.663 & 0.391 \\
&& (0.547) & (0.569) & (0.547) & (0.558) & (0.545) & (0.507) & (0.545) & (0.481) \\
\bottomrule
\end{tabular}
\caption{Simultaneous coverage probabilities over $G.sub$ for $k=4$ and $a=0.7$}
\label{simultGsub9}
\end{table}

\begin{table}
\centering
\begin{tabular}{@{}lllllllllll@{}}
\toprule
 $h$ & \multicolumn{1}{l}{Case} & \multicolumn{1}{c}{$\gamma=0$} & \multicolumn{2}{c}{$\gamma=0.2$} & \multicolumn{2}{c}{$\gamma=0.4$} & \multicolumn{2}{c}{$\gamma=0.6$} & \multicolumn{1}{c}{$-$} \\ \midrule

                                     &  & WB & AWB & DWB & AWB  & DWB & AWB & DWB & SWB
                            \\ \cmidrule(l){3-3} \cmidrule(l){4-5} \cmidrule(l){6-7} \cmidrule(l){8-9} \cmidrule(l){10-10}
\multirow{ 14}{*}{$0.02$} & $0$ & 0.902 & 0.912 & 0.903 & 0.904 & 0.913 & 0.812 & 0.924 & 0.846 \\
&& (0.525) & (0.512) & (0.524) & (0.470) & (0.524) & (0.399) & (0.525) & (0.512) \\
 & $AR_{0.2}$ & 0.812 & 0.870 & 0.836 & 0.854 & 0.827 & 0.712 & 0.813 & 0.693 \\
&& (0.509) & (0.511) & (0.510) & (0.479) & (0.507) & (0.417) & (0.508) & (0.490) \\
 & $AR_{0.5}$ & 0.566 & 0.648 & 0.551 & 0.610 & 0.554 & 0.538 & 0.574 & 0.233 \\
&& (0.434) & (0.451) & (0.433) & (0.443) & (0.431) & (0.404) & (0.431) & (0.381) \\
\bottomrule
\end{tabular}
\caption{Simultaneous coverage probabilities over $G$ for $k=4$ and $a=0.7$}
\label{simultG9}
\end{table}
\clearpage

\begin{table}
\centering
\begin{tabular}{@{}llllllllll@{}}
\toprule
 & \multicolumn{1}{l}{Case} & \multicolumn{1}{c}{$\gamma=0$} & \multicolumn{2}{c}{$\gamma=0.2$} & \multicolumn{2}{c}{$\gamma=0.4$} & \multicolumn{2}{c}{$\gamma=0.6$} & \\ \midrule

                                     &  & WB & AWB & DWB & AWB  & DWB & AWB & DWB 
                            \\ \cmidrule(l){3-3} \cmidrule(lr){4-5} \cmidrule(lr){6-7} \cmidrule(l){8-9}
\multirow{ 13}{*}{pw} & $0$ & 0.949 & 0.947 & 0.947 & 0.935 & 0.948 & 0.896 & 0.947 \\
&& (0.400) & (0.393) & (0.399) & (0.369) & (0.402) & (0.321) & (0.400) \\
 & $AR_{0.2}$ & 0.919 & 0.926 & 0.922 & 0.915 & 0.921 & 0.875 & 0.922 \\
&& (0.388) & (0.391) & (0.390) & (0.370) & (0.389) & (0.330) & (0.391) \\
 & $AR_{0.5}$ & 0.854 & 0.868 & 0.851 & 0.866 & 0.853 & 0.832 & 0.856 \\
&& (0.341) & (0.350) & (0.340) & (0.343) & (0.340) & (0.317) & (0.341) \\
\bottomrule
\end{tabular}
\caption{Coverage with missing data ($h=0.02$)}
\label{missings1}
\end{table}

\begin{table}
\centering
\begin{tabular}{@{}llllllllll@{}}
\toprule
 & \multicolumn{1}{l}{Case} & \multicolumn{1}{c}{$\gamma=0$} & \multicolumn{2}{c}{$\gamma=0.2$} & \multicolumn{2}{c}{$\gamma=0.4$} & \multicolumn{2}{c}{$\gamma=0.6$} & \\ \midrule

                                     &  & WB & AWB & DWB & AWB  & DWB & AWB & DWB 
                            \\ \cmidrule(l){3-3} \cmidrule(lr){4-5} \cmidrule(lr){6-7} \cmidrule(l){8-9}
\multirow{ 13}{*}{pw} & $0$ & 0.956 & 0.953 & 0.954 & 0.941 & 0.955 & 0.909 & 0.955 \\
&& (0.338) & (0.335) & (0.337) & (0.318) & (0.338) & (0.283) & (0.337) \\
 & $AR_{0.2}$ & 0.929 & 0.935 & 0.932 & 0.924 & 0.930 & 0.894 & 0.932 \\
&& (0.330) & (0.334) & (0.331) & (0.322) & (0.331) & (0.291) & (0.330) \\
 & $AR_{0.5}$ & 0.871 & 0.883 & 0.870 & 0.881 & 0.870 & 0.856 & 0.873 \\
&& (0.293) & (0.303) & (0.291) & (0.302) & (0.292) & (0.283) & (0.292) \\
\bottomrule
\end{tabular}
\caption{Coverage with missing data ($h=0.04$)}
\label{missings2}
\end{table}

\clearpage

\section{Additional Empirical Results}

Following up on the frequency domain analysis in the paper, we perform the same analysis on the residual series after a regression of the ethane series on the Fourier terms. We increase the number of Fourier terms, $M$, from 1 to 4. The result is shown in Figure \ref{fig:terms} with an increasing $M$ from \ref{fig:1_term} to \ref{fig:4_terms}. The most important observation in comparison to Figure \ref{fig:perio} is that in all four graphs, the large peak at 1 has disappeared. Oscillations at this frequency have been captured. The peaks at 0 are now the most pronounced ones. In panel \ref{fig:2_terms} for $M=2$, we can see that the tallest peak around 2 vanished, if we compare it to panel \ref{fig:1_term}. The same holds for the remaining two panels: the peaks at 3 and 4 vanish for $M=3$ and $M=4$, respectively. Overall, due to the small scale of the peaks at higher frequencies, we can say that including a minimum of one Fourier term seems the best strategy according to this analysis.

\begin{figure}[hbt]
\centering
\subfigure[$M=1$]
	{
		 \includegraphics[width=0.47\linewidth]{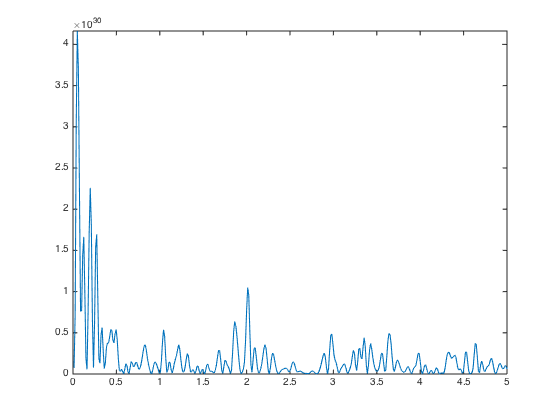}
		 \label{fig:1_term}
      }
\subfigure[$M=2$]
	{
		 \includegraphics[width=0.47\linewidth]{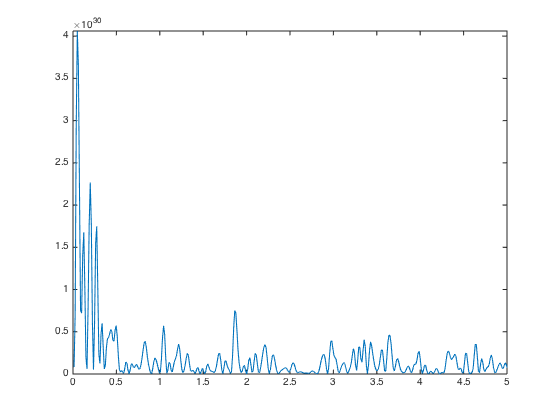}
 		 \label{fig:2_terms}
    }\\
\subfigure[$M=3$]
	{
		 \includegraphics[width=0.47\linewidth]{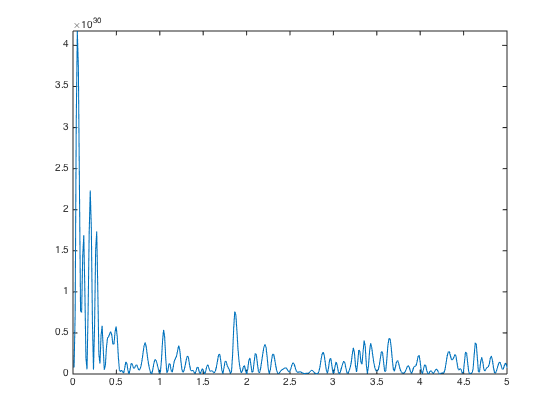}
  		 \label{fig:3_terms}
	}
\subfigure[$M=4$]
	{
		 \includegraphics[width=0.47\linewidth]{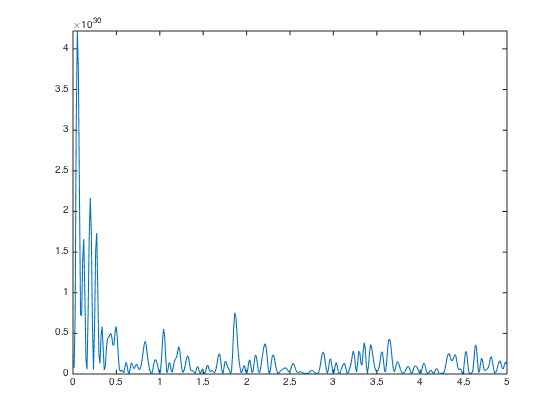}
  		 \label{fig:4_terms}
    }
\caption{The Lomb-Scargle periodogram after projection on an increasing $M$}
\label{fig:terms}
\end{figure}

Next, we show additional plots nonparametric trend plus confidence bands for a different number of Fourier terms ($M$) used to handle seasonality. In the paper, the results shown are for $M=3$. Here, we present additional results for $M=1,2,4$. They are plotted in Figures \ref{fig:1Term}, \ref{fig:2Terms} and \ref{fig:4Terms}, respectively. As pointed out there, we observe that the resulting trend curves and confidence bands are almost identical to those in the paper. Lastly, Figure \ref{fig:bandwidth6} shows that when we increase the bandwidth, there is enough smoothing due to the nonparametric estimator such that no initial regression on Fourier terms is needed. Again, the shape is very similar to the previous graphs.

\begin{figure}[hbt]
	\centering
	\includegraphics[width=\linewidth]{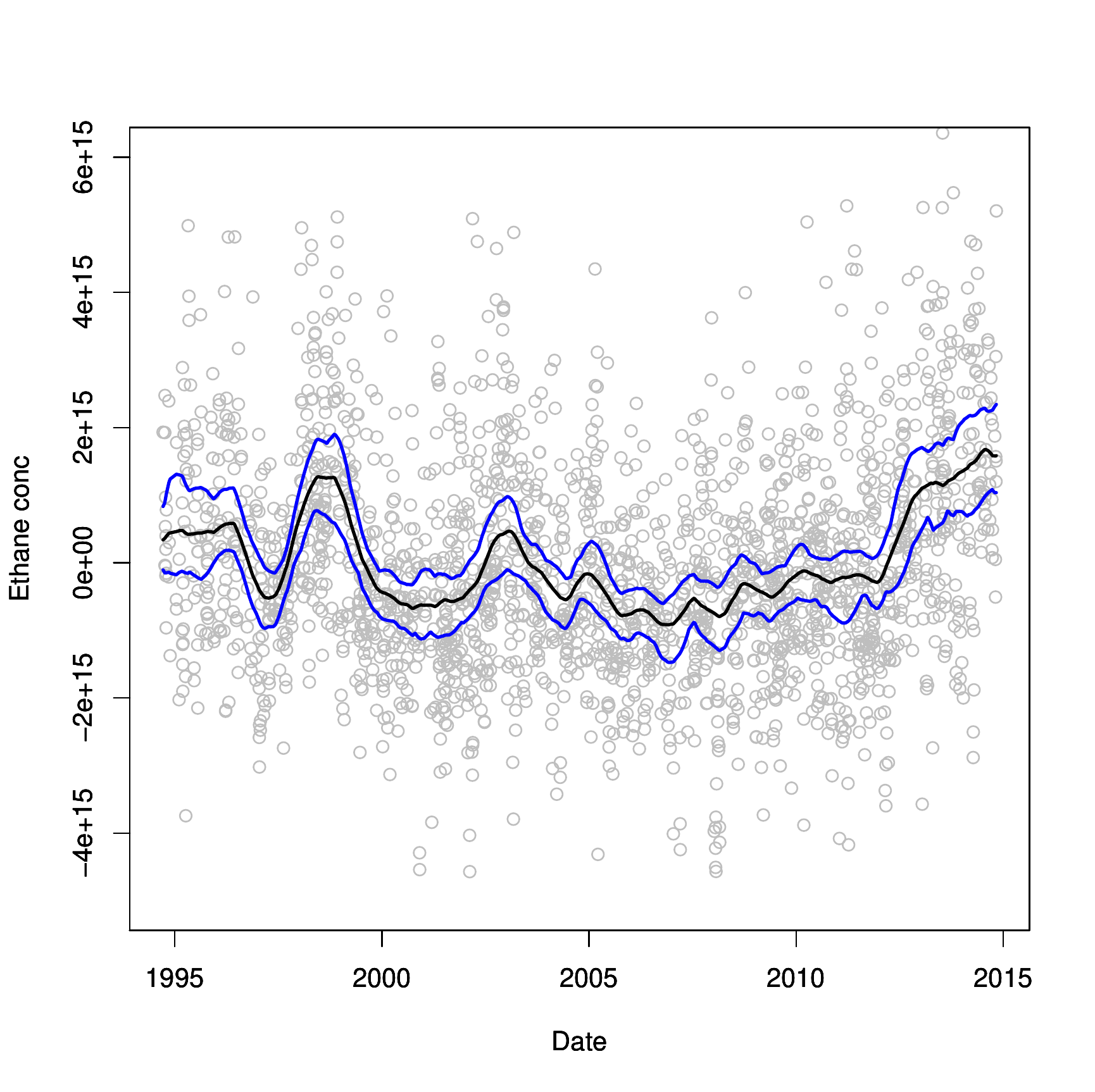}
	\caption{$M=1$, $h=0.03$}
	\label{fig:1Term}
\end{figure}

\begin{figure}
	\centering
	\includegraphics[width=\linewidth]{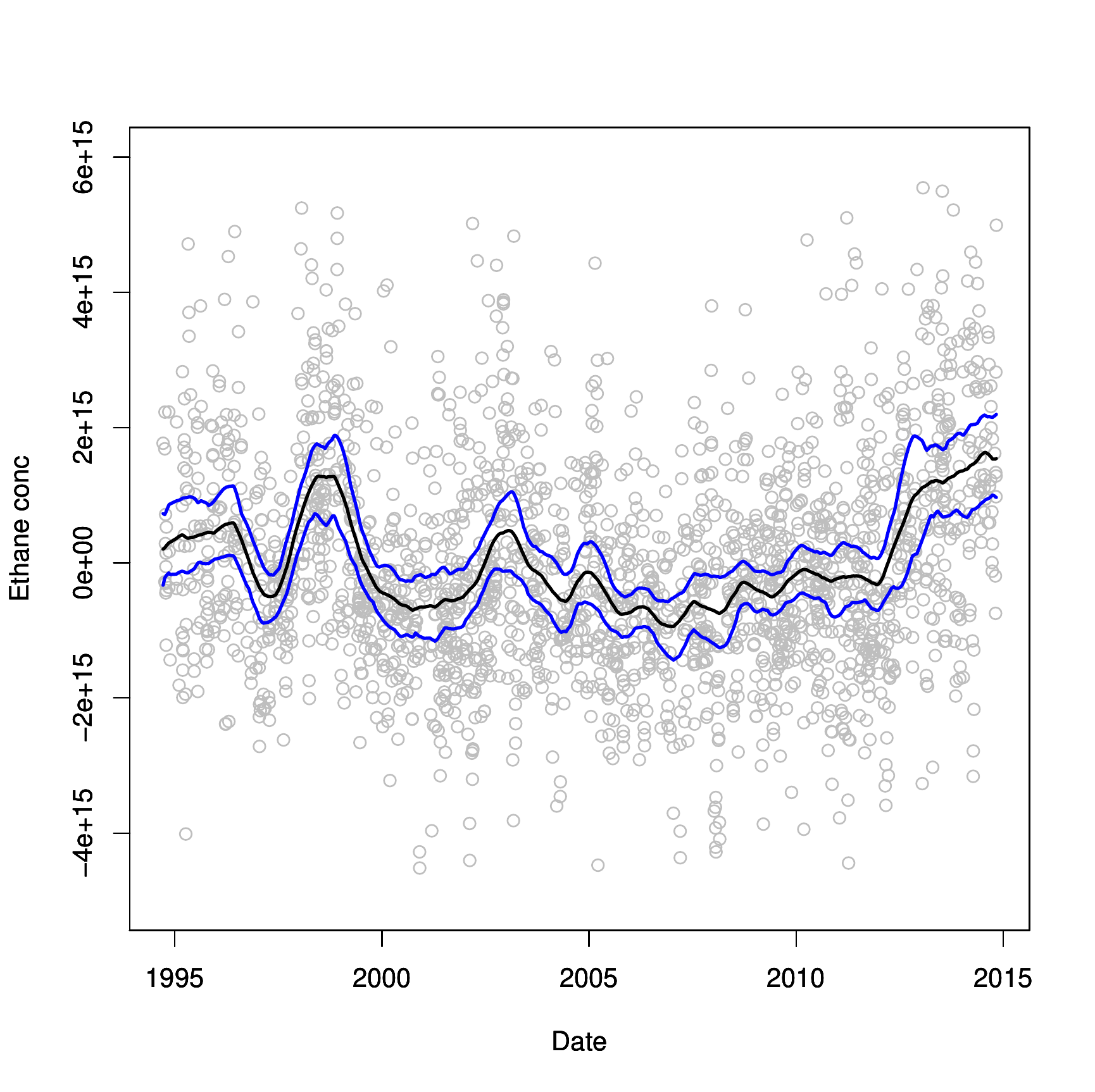}
	\caption{$M=2$, $h=0.03$}
	\label{fig:2Terms}
\end{figure}

\begin{figure}
	\centering
	\includegraphics[width=\linewidth]{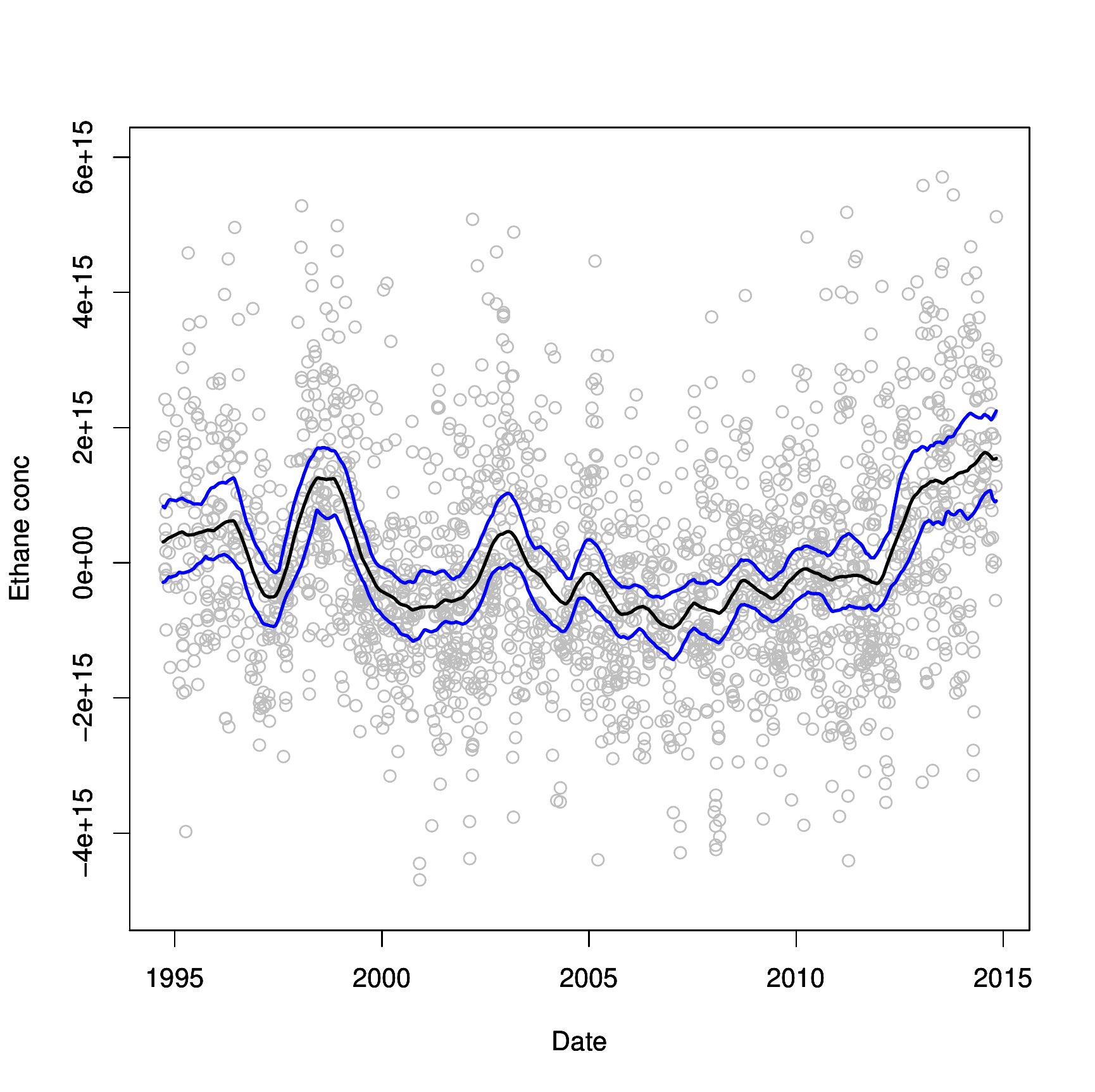}
	\caption{$M=4$, $h=0.03$}
	\label{fig:4Terms}
\end{figure}

\begin{figure}
	\centering
	\includegraphics[width=\linewidth]{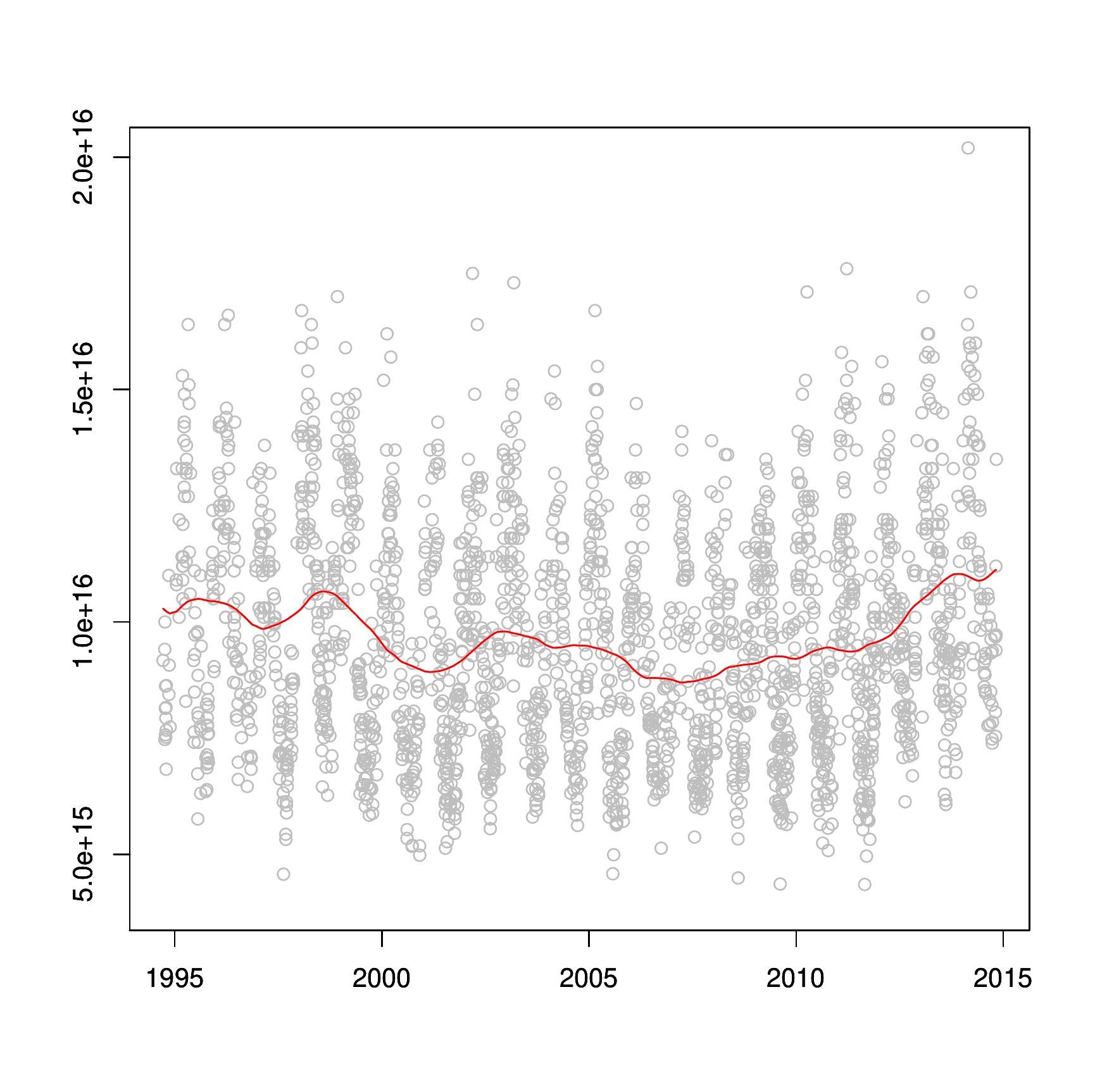}
	\caption{$h=0.06$}
	\label{fig:bandwidth6}
\end{figure}

\end{appendices}

\end{document}